\documentclass{article}
\usepackage{natbib}
\usepackage{fullpage}

\usepackage{microtype}
\usepackage{graphicx}
\usepackage{subfigure}
\usepackage{booktabs}

\usepackage{hyperref}
\hypersetup{
	colorlinks=true,
	linkcolor=red,
	citecolor=blue,
	urlcolor=violet,
	pdftitle={Turnstile \texorpdfstring{$\ell_p$}{lp} leverage score sampling with applications},
	linktocpage=true,
}

\usepackage{amsmath}
\usepackage{amssymb}
\usepackage{mathtools}
\usepackage{amsthm}

\usepackage[capitalize,noabbrev]{cleveref}

\theoremstyle{plain}
\newtheorem{theorem}{Theorem}[section]
\newtheorem{proposition}[theorem]{Proposition}
\newtheorem{lemma}[theorem]{Lemma}
\newtheorem{corollary}[theorem]{Corollary}
\theoremstyle{definition}
\newtheorem{definition}[theorem]{Definition}

\theoremstyle{remark}

\DeclarePairedDelimiter\norm{\lVert}{\rVert}

\newcommand{\at}{\ensuremath{\tilde{a}}}
\let\Pr\relax\DeclareMathOperator*{\Pr}{\mathbf{Pr}}

\DeclareMathOperator*{\median}{\mathbf{median}}
\DeclareMathOperator*{\E}{\mathbb{E}}
\def\eps{\varepsilon}
\renewcommand{\epsilon}{\varepsilon}

\usepackage{algorithm}
\usepackage{bbm}
\usepackage[noend]{algpseudocode}
\usepackage{varwidth}
\newcounter{algsubstate}


\usepackage{eqparbox}

\usepackage{etoolbox} 
\makeatletter
\patchcmd{\algorithmic}{\addtolength{\ALC@tlm}{\leftmargin} }{\addtolength{\ALC@tlm}{\leftmargin}}{}{}
\makeatother

\usepackage[textsize=tiny]{todonotes}

\begin{document}

\title{Turnstile \texorpdfstring{$\ell_p$}{lp} leverage score sampling with applications}

\author{Alexander Munteanu\thanks{Dortmund Data Science Center, Faculties of Statistics and Computer Science, TU Dortmund University, Dortmund, Germany. Email: \texttt{alexander.munteanu@tu-dortmund.de}.}
\and Simon Omlor \thanks{Faculty of Statistics and Lamarr-Institute for Machine Learning and Artificial Intelligence, TU Dortmund University, Dortmund, Germany. Email: \texttt{simon.omlor@tu-dortmund.de}.}}

\vskip 0.3in

\maketitle

\begin{abstract}
The turnstile data stream model offers the most flexible framework where data can be manipulated dynamically, i.e., rows, columns, and even single entries of an input matrix can be added, deleted, or updated multiple times in a data stream. We develop a novel algorithm for sampling rows $a_i$ of a matrix $A\in\mathbb{R}^{n\times d}$, proportional to their $\ell_p$ norm, when $A$ is presented in a turnstile data stream. Our algorithm not only returns the set of sampled row indexes, it also returns slightly perturbed rows $\tilde{a}_i \approx a_i$, and approximates their sampling probabilities up to $\varepsilon$ relative error. When combined with preconditioning techniques, our algorithm extends to $\ell_p$ leverage score sampling over turnstile data streams. With these properties in place, it allows us to simulate subsampling constructions of coresets for important regression problems to operate over turnstile data streams with very little overhead compared to their respective off-line subsampling algorithms. For logistic regression, our framework yields the first algorithm that achieves a $(1+\varepsilon)$ approximation and works in a turnstile data stream using polynomial sketch/subsample size, improving over $O(1)$ approximations, or $\exp(1/\varepsilon)$ sketch size of previous work. We compare experimentally to plain oblivious sketching and plain leverage score sampling algorithms for $\ell_p$ and logistic regression.
\end{abstract}

\clearpage
\tableofcontents
\clearpage

\section{Introduction}

When analyzing huge amounts of data, even linear time and space algorithms may require large computing resources or even reach the limits of tractability. When dealing with data streams, or distributed data, we face additional restrictions regarding their accessibility or communication. In massively unordered models, huge amounts of data are stored and need to be processed in arbitrary order. To deal with such situations, it is necessary to preprocess the dataset and reduce its size before classical data analysis algorithms can perform on a compressed substitute data set. Two main techniques can be identified in the literature, referred to as \emph{coresets} and \emph{sketching}, that quickly compute some sort of smaller data summary while data is presented under the various restrictions mentioned above, and thereby provide mathematical guarantees on the approximation error obtained from analyzing the proxy \citep{Phillips17,Munteanu23}.

Coreset constructions often work by importance subsampling or selection of original rows of a data matrix and reweighting them reciprocally to their sampling probability \citep{MunteanuS18}. This yields unbiased and precise estimates using few rows of high importance that are likely to be included, while many low contributions are redundant and can be subsampled in a near uniform way \citep{LangbergS10,FeldmanSS20}.

Sketching is often seen as a descendant of random projections and aims at randomly isolating rows that have a very high impact on the objective function \citep{Woodruff14}. The idea behind the type of sketches considered in this paper is that these high impact contributions can be separated with high probability from each other by hashing them randomly into buckets, and collisions with less important data add only little noise \citep{CharikarCF04,Woodruff14,Mahabadi2020}.

Coresets admit batch-wise processing of data points using a black-box technique called Merge \& Reduce \citep{BentleyS80,GeppertIMS20,FeldmanSS20,CohenAddadWZ23}, and a lot of effort has been put recently into developing on-line algorithms that simulate $\ell_p$ norm subsampling in a data stream, when the input points are presented row-by-row \citep{ChhayaC0S20,CohenMP20,MunteanuOP22,WoodruffY23}. Dynamic data structures, allowing to remove points after their insertion \citep{FrahlingS05,FrahlingIS08,BravermanFLSY17}, are slightly less common in the coreset literature.

While the above models are often sufficient in practice, massively unordered and distributed data bases require handling so called \emph{turnstile} data streams \citep{Muthukrishnan05} that allow multiple additive updates to change single coordinates of a data matrix in an arbitrary order. Starting from an initial zero matrix $A=0$, data is represented as a sequence of updates of the form $(i,j,v)$ meaning that the previous value $A_{ij}$ is updated to $A_{ij} + v$. Note that this model can simply simulate (multiple) row- or column-wise updates and deletions as in the previous models. Allowing the full flexibility of {turnstile} data streams seems to lie in the domain of linear sketching algorithms, as most known turnstile streaming algorithms can be interpreted as linear sketches. Indeed, under certain conditions, linear sketching \citep{LiNW14,AiHLW16} is optimal for handling turnstile data streams.

Linearity provides a couple of useful properties. For instance in distributed systems, each computing node can calculate their own sketch $\Pi A_{(i)}$ and the final sketch representing the full data is obtained by summing all sketches $\Pi A= \Pi \sum_{i} A_{(i)} = \sum_{i} \Pi A_{(i)}$ at a central node. Linear sketches allow certain database operations to be applied in the sketch space. For instance, when a time varying signal is sketched at time instances $t_1 < t_2$, then the difference of the two sketches $\Pi A_{(t_2)}-\Pi A_{(t_1)}=\Pi A_{(t_1,t_2]}$ represents a sketch of all changes between the two time stamps. Associativity of matrix multiplication also enables projection operations in the sketch space since a sketch of projected data equals the projected sketch: $\Pi (A P) = (\Pi A) P$. Additionally, state of the art sketching techniques make heavy use of sparsity, which allows for fast updates with little, often constant or logarithmic overhead over the time spent on just reading the data. This is commonly referred to as input sparsity time or $\tilde O(\texttt{nnz}(A))$, where $\texttt{nnz}(A)$ denotes the number of non-zero entries in the representation of $A$.

For some problems, this flexibility comes at a price, as lower bounds for sketching $\ell_p$ related loss functions for $p>2$ indicate near linear $\Omega(n^{1-2/p}\log n)$ sketching size \citep{AndoniNPW13}, while subsampling can produce coresets of size $d^{O(p)}$ \citep{DasguptaDHKM09,WoodruffZ13,MunteanuOP22,WoodruffY23,WoodruffY23sens}. The situation is different for $1\leq p \leq 2$, where sketching is more powerful in compressing data. 

But recent research again indicates certain limitations. For logistic regression, data oblivious sketches were only known to give constant factor approximations until recently a first $(1+\eps)$-approximation was developed \citep{MunteanuOW23}, albeit with an exponential dependence on ${1}/{\eps}$. Similarly, a classic result \citep{Indyk06} on sketching the $\ell_1$ norm of vectors had $\exp({1}/{\eps})$ dependencies and this is likely necessary as indicated by impossibility results of \citet{CharikarCF04,LiW021,WangW22}. These seem to suggest that sketching cannot yield $(1+\eps)$ approximations for all queries below $\exp{(1/\eps)}$ or $\exp{(\Omega(\sqrt{d}))}$ size. However, we note that these impossibility results are derived under the assumption that the sketch must be taken as a final data approximation, and is not allowed to be post-processed, which is a major difference to our work.

We remark here that \citet{Indyk06} gave fully polynomial $(1\pm\eps)$-approximations for $\ell_p$ norms, using median operators that turn convex optimization problems to non-convex optimization problems in the sketch space. 
The considered sort of convex loss functions $f(Az)$ remains convex with respect to $z$ for any fixed dataset $A$ directly by rules of combining convex functions. In particular, if $A$ is replaced by any other fixed $A'$ such as a weighted subsample or a sketch, then $f_w(A'z)$ remains convex. It is probably more instructive to explain the source of non-convexity of previous $\ell_p$-norm sketches with $(1+\epsilon)$ guarantee within polynomial size. This came from the fact that for each query $z$, the estimate came from a different row $a'_i$ of $A'$ (namely the median row among all $|a'_i z|_p^p$). Now, imagine this as a dataset that is not fixed, but it is changing in a non-convex way for each query. The median technique is still useful for single estimations, but we avoid to use these methods for the final sketch, so as to preserve convexity and thus the efficient tractability of the optimization problem. 

Again, in contrast to sketching, sampling based coresets are known for $\ell_1$, and logistic regression within $\operatorname{poly}(d,{1}/{\eps},\log n)$ size and without affecting the efficiency of optimizing over the reduced data.
We thus ask the question if it is possible to get the best of the two worlds:

\begin{centering}\label{question}
    \textbf{Question 1:} \textit{Is it possible to obtain the full flexibility of turnstile streaming updates, and fully polynomial sketching/sampling size, while preserving a $(1\pm\eps)$ factor approximation, and convexity of the reduced problem?}
\end{centering}

In particular, we resolve the above question by developing a new algorithm for $\ell_p$ sampling over turnstile data streams.
\begin{definition}[$L_{p,p}$ sampling]\label{def:lpsampling}
    Let $A\in\mathbb{R}^{n\times d}$ with rows $a_i\in\mathbb{R}^d$, and $k\in\mathbb{N}$. An $L_{p,p}$ sampler is a turnstile streaming algorithm that returns a subset $S\subseteq [n]$ of size $S=\Theta(k)$, such that the probability that $S$ contains index $i$ is given by 
    $$\Pr[i\in S] \geq \min\left\{1, (1\pm\eps)\frac{k\|a_i\|_p^p}{\|A\|_p^p}\right\},$$
    where $\|A\|_p = (\sum_{ij} |A_{ij}|^p)^{1/p}$ denotes the entry-wise $p$ norm.
    Moreover, we call it an $\ell_p$ leverage score sampler, if the inclusion probabilities satisfy 
    \begin{align}\label{eq:overestim}
        \Pr[i\in S] \geq \min\left\{ 1,{k u_i^{(p)}}{} \right\},
    \end{align}
    where $u_i^{(p)}= \sup_{z \in \mathbb{R}^d\setminus\{0\}} \frac{|a_iz |^p}{\|Az \|^p_p}$ for $ i\in[n]$ are the $\ell_p$ leverage scores of $A$, see \cref{def:leveragescores}.
\end{definition}
We remark that the amount of overestimation in \cref{eq:overestim} translates into an increase in the sample size, and will thus be controlled by a constant that possibly depends on the dimension $d$, though not on the number of input points $n$.

\subsection{Our contributions}
We answer Question \ref{question} in the affirmative. We first develop an $L_{p,p}$ sampler that processes data presented in a turnstile data stream. After another stage of postprocessing, it identifies $\Theta(k)$ many indexes $i\in[n]$ whose inclusion probabilities satisfy the requirements of \cref{def:lpsampling}. We use known $\ell_p$ subspace embeddings that can be calculated in parallel while reading the turnstile data stream, and obtain a conditioning matrix $P\in\mathbb{R}^{d\times d}$. Post right-multiplication of the $L_{p,p}$ sampler sketch with $P$ yields a well-conditioned basis so that the sampler becomes an $\ell_p$ leverage score sampler. In addition to the row indexes $i\in S$, it returns slightly perturbed rows $\at_i \approx a_i$ such that $\|\at_i-a_i\|_p\leq O(\eps)\|a_i\|_p$, as well as accurate $(1\pm\eps)$-estimates on the sampling probabilities, which translate to $(1\pm\eps)$-approximations of the weights required by various importance sampling coreset constructions.

Our main contributions can be summarized as follows:
\begin{enumerate}
    \item[1)] We simplify and generalize the $L_{2,2}$ sampler of \citet{Mahabadi2020} to arbitrary $L_{p,p},$ for $p \in [1, 2]$, by developing new statistical test procedures on the sketch and providing a tailored analysis of our new algorithm.
    
    \item[2)] We show how our algorithm can be used to sample with probability approximately proportional to $\frac{\norm{a_i}_p^p}{\norm{A}_p^p }+1/n$ as well as $\frac{\norm{a_i}_p^p}{\norm{A}_p^p }+ \frac{\norm{a_i}_q^q}{\norm{A}_q^q}$ for distinct $p, q \in [1,2]$.
    \item[3)] We apply our algorithm to construct $\varepsilon$-coresets over turnstile data streams for a wide array of regression loss functions including linear-, ReLU-, probit-, and logistic regression, as well as their $\ell_p$ generalizations.
    \item[4)] We provide an experimental comparison to previous reduction algorithms for {$\ell_p$ and logistic regression} that were purely based \emph{either} on sketching \emph{or} subsampling.
\end{enumerate}
To our knowledge, we give the first algorithm that returns an $\varepsilon$-coreset for logistic regression and requires only polynomial space in the turnstile data stream setting, improving over the $\exp(1/\eps)$ dependence of \citet{MunteanuOW23}.
Given the impossibility results of \citep{LiW021,WangW22} mentioned above, it may seem surprising that we can circumvent exponential $1/\eps$ dependence. We can get around these limitations by first sketching obliviously, then post-processing the sketch and selecting the right information. These latter steps of 'cherry-picking' from the sketch are crucial to obtain our results. In particular, they violate pure obliviousness required by previous impossibility results.

\subsection{Comparison to related work}
Our work builds upon and extends the work of \citet{Mahabadi2020} on $L_{2,2}$ samplers to arbitrary $L_{p,p}$. The authors claimed that a generalization to other values of $p$ is possible, but out of scope of their paper, which focused on $L_{2,2}$, and the sum of $\ell_2$ norms, denoted $L_{1,2}$. We note that \citet{DrineasMMW12} gave a high level description for the case $p=2$ but required a second pass to collect the samples from the original data instead of extracting samples from the sketch. A similar $L_{1,1}$ sampling technique was developed in \citet{SohlerW11} in the context of $\ell_1$ regression. However, the paper gives only an outline of the proof and the full details apparently never appeared. Other classic literature on $\ell_p$ sampling, and recent advances improving the error of the subsampling distribution to zero, focused on the special case of sampling entries from a vector proportional to their $\ell_p$ norm contributions \citep{MonemizadehW10,AndoniKO11,JowhariST11,JayaramW21,JayaramWZ22}, rather than sampling rows of a matrix. We refer the interested reader to \citet{CormodeJ19} for a survey on this line of research.

The work of \citet{Mahabadi2020} requires generalizations of the well-known AMS \citep{AlonMS99} and CountSketch \citep{CharikarCF04} algorithms to estimate the Frobenius norm of their (transformed) input matrices and identify the rows that exceed a certain fraction thereof. Our techniques also rely on the CountSketch but the AMS sketch using Rademacher random variables is a special choice that does not allow to generalize beyond the case $p=2$. There exist alternatives for sketching $\ell_p$ norms via $p$-stable random variables, but these distributions are not expressible in closed form except for $p\in\{1,2\}$ and are cumbersome to analyze \citep{Indyk06,MaiMMRSW23}. On our quest for a unifying algorithm for all $p\in [1,2]$, we exploit the percentiles of norms sketched in independent repetitions of the Count\-Sketch data structure and do not require additional sketches to estimate the required thresholds. In particular, there is no special treatment across different values of $p\in[1,2]$, which simplifies our algorithms. We note that \citet{LiW16} developed similar ideas for a subroutine for estimating $\|A\|_p^p$ in special cases.
\begin{algorithm}[ht!]
  \caption{Finding $\ell_p$ heavy hitters.}\label{alg:findhh}
  \begin{algorithmic}[1]
    \Statex \textbf{Input:} data matrix $A \in \mathbb{R}^{n \times d}$ presented as a turnstile data stream, and parameters $s$, $r$ and $\varepsilon$.;
    \Statex \textbf{Output:} list $L\subseteq [n] \times \mathbb{R}^d$ of slightly perturbed rows of $A$ with large $\ell_p$ norms, each $(i, \at_i ) \in L $ satisfying $\norm{\at_i-a_i}_p \leq (\varepsilon/3) \norm{a_i}_p$; 
    \State  For $i \in [n]$ and $j \in [s]$ generate $h_{i, j}\in [r]$ uniformly at random;
    \State  For $i \in [n]$ and $j \in [s]$ generate a sign $\sigma_{i, j}\in \{-1, 1\}$ uniformly at random; 
    \State \textbf{//* sketching stage *//}
    \State  For $j \in [s]$ initialize $B_j \in \mathbb{R}^{r \times d}$ as $0$-matrix; 
        \For{$l=1\ldots N$}
            \State Let update $u_l$ be of the form $a_i=a_i+x_l$;
 	    \State For $j \in [s]$ set $B_{j, h_{i, j}}= B_{j, h_{i, j}}+ \sigma_{i, j} x_l$; 
        \EndFor
    \State \textbf{//* extraction stage *//}
    \State Let $L$ be an empty list ;
    \State Let $M_0:=M_0(A)$ be the $0.65$-percentile of the set $\{ \norm{B_{j, 1}}_p^p \mid j \in [s] \}$
    \For{$i \in [n]$} 
        \State For $j \in [s] $ denote $ \at_{i, j}=\sigma_{i, j} B_{j, h_{i, j}}$; 
 	\State Compute $v_i=\median_{j \in [s]} \norm{\at_{i, j}}_p^p $ ;
        \If{$v_i \geq (12/\varepsilon)^p M_0$} 
            \State Find $j \in [s]$ minimizing \\\qquad\qquad $\median_{j' \in [s]}\{ \norm{\at_{i, j}-\at_{i, j'}}_p^p \}$ ;
            \State Add $(i, \at_{i, j})$ to $L$ ;
        \EndIf
    \EndFor
    \State RETURN $L$;
  \end{algorithmic}
\end{algorithm}

As mentioned in the introduction, there are a lot of works on subsampling based on $\ell_p$ row norms, in particular using $\ell_p$ leverage scores \citep{DrineasMM2006a,DrineasMMW12,DasguptaDHKM09,MolinaMK18,MunteanuSSW18,MunteanuOP22,WoodruffY23sens,FrickKM24}, and related measures such as Lewis weights \citep{CohenP15,MaiMR21,WoodruffY23}.
Many of the above sampling algorithms can be handled in row-wise insertion data streams using a standard technique called Merge \& Reduce \citep{BentleyS80,GeppertIMS20,FeldmanSS20,CohenAddadWZ23}, or via online algorithms \citep{ChhayaC0S20,CohenMP20,MunteanuOP22,WoodruffY23}.

Our work extends $\ell_p$ leverage score sampling to the most flexible and dynamic setting of turnstile data streams. We simulate $\ell_p$ norm sampling algorithms by means of first sketching the data obliviously. After postprocessing the sketches, they allow us to extract an approximate sample that satisfies the coreset guarantee. Hereby, we provide a general framework that allows $\ell_p$ leverage score sampling based coreset constructions to be simulated almost generically with little overhead compared to the off-line construction. The approximate weights and probabilities are readily of such form as to provide $(1\pm O(\eps))$ factor guarantees. Thus, if we had access to the original data rows once again, our sampler would apply in a black-box manner to any off-line construction that uses $\ell_p$ leverage score sampling.
\begin{algorithm}[ht!]
  \caption{$\ell_p$ norm sampling.}\label{alg:turnstilesampling}
  \begin{algorithmic}[1]
    \Statex \textbf{Input:} data matrix $A \in \mathbb{R}^{n \times d}$ presented as a turnstile data stream, matrix $P \in \mathbb{R}^{d \times d}$ (identity matrix $P=I_d$ if not specified), and parameters $k$, $s$ and $r$.;
    \Statex \textbf{Output:} a sample $S$ consisting of tuples $(i, \at_i, w_i)$ where for $i \in [n]$, $\at_i \approx a_i $ and $w_i$ is roughly the inverse sampling probability of $i$; 
    \State  For $i \in [n]$ generate independent scaling factors $t_i \in (0, 1)$ uniformly at random;
    
    \State Let $A'=TA$ be the matrix where the rows $a_i$ of $A$ are multiplied by $t_i^{-1/p}$;
    \State Forward turnstile updates for $A'$ to \cref{alg:findhh};

    \State For $j \in [s]$ set $B_j=B_jP$ in \cref{alg:findhh}; 
    \State Let $L$ be the output of \cref{alg:findhh};
    \State Let $S_k$ be the set of $k$ elements of $L$ with the largest $\ell_p$ norms;
    \State Set $\alpha=\min_{i \in S_{k}} \norm{\at_i'}_p^p$;
    
    \State For $(i, \at_i') \in L$ we set $\at_i= \at_i't_i^{1/p}$ ; 
    \State Set $S=\{ (i, \at_iP^{-1},1/\min\{ 1, \frac{\norm{\at_i}_p^p}{\alpha} \}) \mid \norm{\at_i'}_p^p\geq\!\alpha \}$ ;
    \State RETURN $S$;
  \end{algorithmic}
\end{algorithm}
There is only one \emph{additional} requirement for full turnstile processing, where after seeing the data once, we only have access to the sketches instead of the original data. 
Namely, the loss function needs to be \emph{robust} to the small perturbations of the original rows returned by our algorithm. To provide a wide array of applications as a corollary of our methods, we prove the robustness property for wide classes of functions such as the linear regression loss, ReLU loss, logistic regression, probit regression, and their $\ell_p$-generalizations.

In particular, we give the first turnstile streaming algorithm for logistic regression that achieves a $(1+\eps)$-approximation with fully polynomial dependence on the input dimensions, improving over the $O(1)$-factor oblivious sketching algorithms of \citet{MunteanuOW21,MunteanuOW23}, and over the $(1+\eps)$-approximation of \citet{MunteanuOW23} that had an $\exp(1/\eps)$ dependence in its sketching dimension.
We point out that their sketches were directly the final approximations and input to the optimization algorithm, in which case the aforementioned impossibility results \citep{LiW021,WangW22} apply. To circumvent these limitations, our new algorithm uses oblivious sketches as intermediate data structures from which we extract an approximate coreset in a postprocessing stage. This might seem minor, but is actually a crucial point that allows to get below the exponential dependence and yields sketches and coresets of fully polynomial size with respect to all input parameters.

\section{Algorithms and technical overview}
As we have mentioned above, the sketching algorithm is similar to previous $\ell_p$ samplers using the CountSketch and randomized scaling. It is usual in this line of research to analyze the algorithms under the assumption of full independence of generated random numbers. Since this assumption implies $\Omega(n)$ space complexity, we will provide the necessary arguments to reduce this overhead to only a $\log(n)$ factor at the end of the section.

Our sketching matrix can be written as a concatenation of a diagonal $n \times n$ matrix $T=\operatorname{diag}(t_1^{-1/p},\ldots,t_n^{-1/p})$, where $t_i\sim U(0,1)$ and a CountSketch $S$ with $r$ rows and $s$ independent repetitions. Each repetition $S_j, j\in [s]$ is an $r \times n$ matrix with one single non-zero entry indexed by a uniform random value $h_{i,j}\in [r]$ in each column $i\in[n]$, that takes a uniform value $\sigma_{i,j}\in \{-1,1\}$. Each sketch of an input matrix $A\in \mathbb{R}^{n\times d}$ is then calculated by $B_j = \Pi_j A=S_j T A,$ for $j\in [s]$. The exact update procedure is given in \cref{alg:findhh} resp. \cref{alg:turnstilesampling}.

The idea behind the CountSketch algorithm (\cref{alg:findhh}) is that there cannot be too many large entries $i\in[n]$ and thus they get separated with good probability when they are mapped to the target coordinates by the functions $h$. Collisions still happen, but only with small entries, whose contributions become even smaller by summing them using random signs $\sigma$. This ensures that very large entries $a_i$ are approximately preserved not only with respect to their norm but also regarding their orientation, as their sketched approximations $\at_i$ after bringing them back to the original scale and sign satisfy $$\|\at_i-a_i\|_p\leq O(\eps) \|a_i\|_p.$$

The purpose of the uniform random values $t_i \sim U(0,1)$ is to randomly upscale the contributions to become heavy coordinates with probability proportional to our desired target $\ell_p$ distribution. The idea is illustrated by the fact that $$\Pr\left[\bigg\|\frac{a_i}{t_i^{1/p}}\bigg\|^p_p\geq \frac{\|A\|_p^p}{k}\right] = \Pr\left[t_i \leq \frac{k\|a_i\|_p^p}{\|A\|_p^p}\right] = \frac{k\|a_i\|_p^p}{\|A\|_p^p},$$ which is (up to clipping at $1$) exactly the right distribution for sampling $\Theta(k)$ elements proportional to their $\ell_p$ norm contribution with good probability.

Since $\|A\|_p^p$ is not easy to calculate over a turnstile data stream, previous work approximated the required threshold from an AMS sketch or using a sketch with i.i.d. Cauchy entries, i.e., specific methods designed for the special choices of $p\in\{1,2\}$. The Cauchy sketch is in principle extendable using $p$-stable distributions, which exist for $p\in[1,2]$, but except for the special cases $p\in \{1,2\}$, they do not admit closed form expressions and are cumbersome to analyze \citep{Indyk06,MaiMMRSW23}. We thus follow a different statistical idea for extracting the relevant information directly from the CountSketch.

\subsection{Idea 1: thresholding the CountSketch}
To calculate the required threshold, we select an arbitrary row/bucket out of the independent repetitions of the Count\-Sketch. W.l.o.g., we simply take the first bucket $B_{j,1}, j\in [s]$, and we let $M_0$ be the $.65$-percentile of the realized $\ell_p^p$ norm of the sketched buckets, i.e., of the set $\{\|B_{j,1}\|_p^p\mid j\in[s]\}$. The idea behind this value is that it can be upper bounded in terms of $M=\sum_{i \in S_R} \norm{a_i}_p^p$, the $\ell_p^p$ norm of the tail, ignoring the largest $r/20$ rows in $\ell_p^p$ norm, divided by the number of rows $r$ of the sketch. $M_0$ can also be lower bounded by the theoretical $.6$-percentile of the $\ell_p^p$ norm contributions of the buckets in the CountSketch, i.e., by $M'=\inf \{ w \in \mathbb{R}_{\geq 0} \mid P(\|B\|_p^p \leq w)\geq 0.6 \}$. With these quantities in place and choosing sufficiently large number of repetitions $s\gtrsim \log(n/\delta)$, we can give the following bound $$M' \leq M_0 \leq 4 M/r.$$
A direct analysis using $M_0$ is not possible but we can estimate this threshold by theoretical upper and lower bounds. The upper bound is used to show that all heavy elements with $\|a_i\|_p^p\gtrsim M/(\eps^p r)$ are included in the sample. The lower bound $M'$ allows us to prove that the elements whose median $\ell_p^p$ norm estimates $v_i$ in the sketch are large w.r.t. this threshold, are actually large in their original magnitude. It can further be shown for these elements that their median estimates are $(1\pm\eps)$-approximations to their true $\ell_p^p$ norm and thus that they are in the set of returned large elements.
Finally, we show that at least half of the sketches not only preserve the norm up to $(1\pm\eps)$ but also preserve the orientation up to a small relative error perturbation, i.e., $S_i:=\{j \in [s] \mid \norm{\at_{i, j}-a_i}_p\leq \varepsilon\|a_i\|_p/9\} \geq s/2$. Therefore, taking the repetition that minimizes the median $\ell_p$ distance to all other repetitions and applying the triangle inequality over the original element, yields an approximation $\at_{i}$ that is close to the original element, i.e., it satisfies $\|\at_{i}-a_i\| \leq (\eps/3)\|a_i\|_p$. 

Now, with these properties in place, we are able to prove that if the number of rows $r$ and repetitions $s$ are chosen sufficiently large, then all the items returned by the algorithm satisfy the desired approximation guarantees.
Overall, we conclude that all sufficiently large elements have an approximate representative in the output and all elements in the output are sufficiently close approximations of their respective original input points.

\begin{theorem}\label{thm_main:findhh}
    Let $\eps,\delta\in(0,1/20],\gamma\in(0,1)$. Let $L$ be the list of tuples in the output of \cref{alg:findhh}. Further let $S_R(r/20)$ be the subset of rows excluding the $r/20$ largest $\ell_p$ norms and let $M=\sum_{i \in S_R} \norm{a_i}_p^p$.
    If $r = 8\gamma^{-1} \cdot (12/\varepsilon)^p$ and $s \geq 3\ln(6n /\delta)/0.025^3$ then with probability at least $1-\delta$, the following properties hold: for any element $(i, \at_i) \in L $ it holds that $\norm{\at_i-a_i}_p\leq (\varepsilon/3)\norm{a_i}_p$ and $\norm{\at_i}_p^p=(1 \pm \varepsilon)\norm{a_i}_p^p$.
    Further, for any $i \in [n]$ with $\norm{a_i}_p^p \geq \gamma M$ it holds that $i \in L$.
\end{theorem}

\subsection{Idea 2: controlling random rescaling by means of the harmonic series}
For the sake of presenting the high level idea, we fix $p=1$ for the moment and consider the matrix $A \in \mathbb{R}^{n \times 1}$ consisting of $n$ copies of the row $a_i=1$.
If we multiply each row with $t_i^{-1}$, where $t_i \sim U(0, 1)$ are drawn uniformly at random, then the new matrix $A'=TA$ with rows $a'_i= a_i/t_i$ consists roughly of the entries $n, n/2, n/3 , \dots , n/(n-1), 1$ in expectation.
Summing over these entries forms a harmonic series that yields $\norm{A'}_1= \Theta(n \log(n))$ and the $k$ largest elements of $A'$ are bounded from below by $n/k$.

In other words, the previous threshold becomes $M=\Theta( n \log(n))$, i.e., it increases by a $\log n$ factor and we aim to find all rows with $\ell_1$ norm greater or equal to $n/k$.
If we now apply \cref{alg:findhh} to $A'$ with $r= O( k \log(n) /\varepsilon)$ then all elements with $a_i'\geq {n/k}=\Theta(M/(k\log(n)))$ will be in $L$ with high probability. The challenge is to control the randomness of the variables $t_i$ since by the uniform distribution they have a high variance, and to generalize the idea to arbitrary non-uniform instances and to different
$p\in[1,2]$.

In our detailed analysis, \cref{alg:turnstilesampling} is slightly modified by applying \cref{alg:findhh} twice in parallel to avoid dependencies between the threshold $\alpha$ and the final sample $S$.\footnote{See \cref{mod: alg2} for details.} The main purpose of this modification is to keep the analysis clean and simple while running time and space complexities remain bounded to within a factor of two. The plain algorithm as presented here in \cref{alg:turnstilesampling} is likely to have the same properties up to small constant factors but its analysis would require additional technicalities that distract from understanding the main ideas behind the algorithm. Moreover, we assume that the matrix $P$ equals the default choice of the identity matrix $I\in\mathbb{R}^{d\times d}$; other choices are discussed later in the applications of \cref{sec:applications}.

We summarize the properties of the sample returned by \cref{alg:turnstilesampling} as follows:

\begin{theorem}\label{thm_main:alg2}
    If we apply the modified version of \cref{alg:turnstilesampling} (see \cref{mod: alg2}) with $0 < \varepsilon,\delta \leq 1/20$, $k \geq 160\ln(12/\delta)$, $r \geq 32 k\ln(n) \cdot (72/\varepsilon)^p$, and $s \geq 3\ln(36n/\delta)/0.025^3$, then with probability at least $1- \delta$ it holds that
    \begin{itemize}
        \item[1)] $|S| \in [k, 2k]$,
        \item[2)] index $ i \in S$ is sampled with probability\\[5pt]\hspace*{33pt}$p_i:= P(i \in S) \geq \min \left\{1, \frac{k \norm{a}_p^p}{\norm{A}_p^p}\right\},$
        \item[3)] if $i \in S$ then $\norm{\at_i-a_i}_p \leq (\varepsilon/3) \norm{a_i}_p $,
        \item[4)] if $i \in S$ then $w_i =(1 \pm \varepsilon)\frac{1}{p_i}$,
        \item[5)] $\sum_{i \in S} w_i \norm{\at_i}_p^p= (1 \pm \varepsilon)\norm{A}_p^p$.
    \end{itemize}
\end{theorem}

The first item ensures that the sample size will be within a constant factor to the required size $k$.\footnote{Note that the plain \cref{alg:turnstilesampling} returns exactly $k$ elements, which is desirable for our experiments with fixed subsample sizes.} The second item ensures that the marginal sampling probabilities satisfy the right distribution of \cref{def:lpsampling}. The third item yields that each sample is a close approximation of their corresponding original input point. The fourth item ensures that the weight corresponds up to $(1\pm\eps)$ to the inverse inclusion probability, which is required to obtain an unbiased estimate of a sum by their weighted importance subsample. Finally, item five shows that the weighted sum over $\ell_p^p$ norms gives an $(1\pm\eps)$ estimate for the entry-wise $\ell_p^p$ norm of the full original data.

The proof of \cref{thm_main:alg2} is subdivided into several technical lemmas. The full details are in \cref{mod: alg2}. Here, we provide a high level overview:

First, we determine the expected norm of the $k$-th largest row of $A'$. Note that $\norm{a_i'}\geq \norm{a_i}$.
Instead of assuming that $ \norm{a_i}_p^p \geq \norm{A}_p^p/k$, we define $A(k) \in \mathbb{R}^{n \times d}$ to be the truncated matrix that we get by scaling down the largest rows of $A$ so that all rows $a_i(k)$ of $A(k)$ satisfy $\norm{a_i(k)}_p^p \geq \norm{A(k)}_p^p/k $.
The exact value of $\norm{a_i}_p^p$ does not matter but the analysis becomes more complicated for very large values.
We use this to show that rows with $\norm{a_i}_p^p \geq \norm{A}_p^p/k \geq \norm{A(k)}_p^p/k$ remain large rows after multiplying with $t_i$.

After truncating the large rows of $A'$ in this way, we show that the total sum $M''=\sum_{i\in S_R(r/20)}\norm{a_i'}_p^p$, excluding the largest contributions is small enough to guarantee that all rows of $A'$ with the $k$ largest norms are in $L$. Note that a $\gamma$ fraction of $M''$ serves as a threshold for the event $i\in L$ in \cref{thm_main:findhh}, so we would like $M''$ to be not much larger than the original $M$.

When proving that this is indeed the case, we need to take care of one complication. Namely, the expected value of 
$\norm{a_i'}_p^p=\norm{a_i}_p^p/t_i$ is unbounded.
However, after truncation, we know that $t_i>\max\{ 1/n,\norm{a_i}_p^p/u \}$ for some $u \in \mathbb{R}_{\geq 0}$, which enables to bound the expected value of $ \norm{a_i'}_p^p$ by $\norm{a_i}_p^p \log(n)$ and the variance by $ 2u \norm{a_i}_p^p$.

Using these properties, we can prove that the total contribution of the elements that are not large, is bounded by $M'' = O(M \log(n))$ as already indicated in the introductory example.
Then, we show that we can make the same analysis work up to further $(1\pm\eps)$ errors when we only have access to the sketched approximations $\at_i'$ instead of the exact values of $a_i'$.
Finally, we approximate the sampling probabilities, whose inverses serve as $(1\pm\eps)$ approximate weights.
Combining these additional uncertainties with the properties of \cref{alg:findhh} provided in \cref{thm_main:findhh}, we conclude the proof of \cref{thm_main:alg2}.

\subsection{Sublinear space with logarithmic overhead}
The hash functions denoted by $h$ as well as the random signs $\sigma$ admit random variables of bounded independence, for which hashing based random number generators are available that require only a seed of size $O(\log n)$ and are able to produce the entries instantly when they are required \citep{AlonBI86,AlonMS99,Dietzfelbinger96,RusuD07}. Derandomization of the random scalars $t_i$, as well as other random variables used in the applications of the next section, seems more complicated. To this end, we use in a black-box manner, a standard psedorandom number generator of \citet{Nisan92} that also produces its random numbers on the fly as required and uses only polylogarithmic overhead to simulate a polynomial amount of independent random bits required in our analysis.

\begin{proposition}[\citealt{Nisan92}, cf. \citealt{JayaramWZ22}]
    Let $\mathcal A$ be an algorithm that uses $S = \Omega(\log n)$ space and $R$ random bits. Then there exists a pseudorandom number generator for $\mathcal A$ that succeeds with high probability and runs within $O(S \log R)$ bits.
\end{proposition}

\section{Applications}\label{sec:applications}
Our algorithms provide a fairly general framework for turnstile streaming algorithms that simulates under mild conditions any off-line coreset construction that builds upon $\ell_p$ leverage score sampling, up to little overheads in the sketch resp. subsample size. In this section, we discuss the additional conditions and give a brief overview over the analysis for the loss functions of several important regression problems, showing that they can be handled within our framework.
In the presented form, our algorithms simulate -- by means of sketching a turnstile data stream -- drawing a subsample of the rows from the input matrix proportional to their $\ell_p^p$ norm contribution, i.e., proportional to $\|a_i\|_p^p/\|A\|_p^p$. This is commonly referred to as row-norm sampling and usually yields only additive error guarantees. For the desired multiplicative $(1\pm\eps)$ guarantees, the probabilities need to be replaced by (approximate) $\ell_p$ leverage scores obtained from a \emph{well-conditioned} basis $U$ so as to sample proportionally to $\|u_i\|_p^p/\|U\|_p^p$. In addition, many algorithms require sampling from a mixture of $\ell_p$ leverage scores with another, e.g., a uniform distribution. To sample approximately from such distributions, we need some additional ideas.

\subsection{Idea 3: sampling from mixture distributions and \texorpdfstring{$\ell_p$}{lp} conditioning}
Say, we would like to sample from a mixture of two distributions $p$ and $q$. Then we can show by simple algebraic manipulations that if $S_1 \sim p$ and $S_2 \sim q$ then $S=S_1\cup S_2$ is a sample whose marginal inclusion probabilities are in $\Pr[i\in S]=\Theta(p_i+q_i)$. And if $p$ and $q$ are only known up to $(1\pm\eps)$ factors, as is the case with our $\ell_p$ samplers, then $\Pr[i\in S]$ can be approximated up to $(1\pm\eps)$ factors, which implies that all properties ensured by the sampler continue to hold for the combined sample. The second distribution is often a simple uniform sample, in which case it can be included into the sketching algorithm for the $\ell_p$ distribution by only hashing the entries $i\in [n]$ that satisfy $t_i > c/n$ and otherwise including them in the uniform sample.
\begin{corollary}\label{cor:combinedsamplingmain}
    Combining a sample $S_1$ from \cref{alg:turnstilesampling} with parameter $k$ and a uniform sample $S_2$ with sampling probability $k/n$ we get a sample $S_1 \cup S_2$ of size $\Theta(k)$ and the sampling probability of $i$ is $\Omega\left(k \left(\frac{\norm{a_i}_p^p}{\norm{A}_p^p} +1/n\right)\right)$, for any sample $\at_i$ we have that $\norm{\at_i-a_i}_p\leq (\varepsilon/3) \norm{a_i}_p $.
    Further, the sampling probability and thus appropriate weights can be approximated up to a factor of $(1\pm\eps)$.
\end{corollary}

To obtain $(1\pm\eps)$ relative error guarantees by $\ell_p$ leverage score sampling, we need to be able to transform the input to a so called well-conditioned basis $U$ for the $\ell_p$ column space of $A$ \citep{DasguptaDHKM09}. This is a generalization of the orthonormal basis in $\ell_2$ to general $\ell_p$ which are not rotationally invariant and therefore require more complicated constructions to ensure low bounded distortions. 

\begin{definition}[\citealt{DasguptaDHKM09}]
\label{def:good_basismain}
Let $A$ be an $n\times d$ matrix, let $p \in [1,\infty)$, and let $q \in (1, \infty]$ be its dual norm, satisfying $\frac{1}{p}+\frac{1}{q}=1$.
Then an $n \times d$ matrix $V$ is an \emph{$(\alpha,\beta,p)$-well-conditioned basis} for the column space of $A$ if\\
(1) $\Vert V \Vert_p:=\left( \sum_{i \leq n, j \leq d}|V_{ij}|^p\right)^{1/p}\leq \alpha$, and\\
(2) for all $z\in\mathbb{R}^d$, $\Vert z \Vert_q \leq \beta \Vert V z\Vert_p$.

We say that $V$ is an \emph{$\ell_p$-well-conditioned basis} for the column
space of $A$ if $\alpha$ and $\beta$ are in $d^{O(1)}$,
independent of $n$.
\end{definition}

The required basis transformations involve right-multiplication of our sketches with a conditioning matrix $P$. To this end, we can simply use the associativity of matrix multiplication to postprocess the sketches. I.e., it holds that $\Pi U = \Pi (AP) = (\Pi A) P$ (see \cref{alg:turnstilesampling}). To obtain $P$, we run in parallel to the $\ell_p$ row-sampler another turnstile sketch $\Pi_2 A$ that gives an $\ell_p$ subspace embedding in low dimensions, from which a $QR$-decomposition yields via $\Pi_2 A=QR$ the desired conditioning matrix $P=R^{-1}$. This idea goes back to \citet{SohlerW11,DrineasMMW12,WoodruffZ13} and has become a standard technique in recent literature. Using the oblivious $\ell_p$ subspace embeddings of \citet{WoodruffY23lpSE}, we get the following result.
\begin{proposition}\label{pro:Rmatrixmain}
    There exists a turnstile sketching algorithm that for a given $p\in[1,2]$ computes an invertible matrix $R$ such that $AR^{-1}$ is $(\alpha,\beta,p)$-well-conditioned with $\alpha=O(d^{2/p-1/2}(\log d)^{1/p-1/2}),$ and $\beta=O((d(\log d)(\log\log d))^{1/p})$, and $(\alpha\beta)^p = O(d^{3-p/2}(\log d)^{2-p/2}(\log\log d))$ for $p\in[1,2)$. For $p=2$ it holds that $\alpha=O(\sqrt{2d}),\beta=O(\sqrt{2})$, and $(\alpha\beta)^p=O(d)$. Moreover, the $\ell_p$ leverage scores $u_i^{(p)}$ satisfy $u_i^{(p)}\leq \beta^p\norm{a_i R^{-1}}_p^p$, and $\sum_i u_i^{(p)}\leq (\alpha\beta)^p=d^{O(1)}$.
\end{proposition}
Since the above conditioning result uses dense $\ell_p$ subspace embedding matrices which come with the computational bottleneck of the current matrix multiplication time, we remark that there exist sparse alternatives for $\ell_p$ subspace embeddings given in \citealp[Theorems  4.2, 5.2 of][]{WangW22}.
However this comes at the cost of slightly larger $d$ dependence resulting in $(\alpha\beta)^p = O(d^{2+p/2}(\log d)^{1+p/2}).$

Another interesting aspect is that the proof of \citep{WoodruffY23lpSE} uses so called $\ell_p$ spanning sets, relaxing slightly the dimension of well-conditioned bases, which yields an almost optimal linear $(\alpha\beta)^p = O(d\log\log d)$ conditioning. However, their computation is based on repeatedly reweighted $\ell_2$ leverage score calculations. Current non-adaptive/adaptive sketching techniques \citep{Mahabadi2020} are limited to post right-multiplication, but re-weighting would require post left-multiplication. It is thus currently unclear whether the direct construction of $\ell_p$ spanning sets is possible in our setting of turnstile data streams. It seems even less clear whether recent local search and non-constructive improvements \citep{BhaskaraMV23} can be leveraged. Developing a constructive version that operates on turnstile data streams is thus an important and exciting open problem.

\subsection{Idea 4: robustness of various loss functions under small perturbations}
Our final step before applying our new samplers to provide a framework for approximating a broad array of loss functions studied in previous literature, is to show that they can handle the small perturbations that are introduced by replacing the original data samples $a_i$ by their sketched versions $\at_i$ with $\|\at_i - a_i\|_p\leq O(\eps)\|a_i\|_p$. This is not immediate for the considered loss functions, and needs to be verified on a case-wise basis.
We note that the remaining items, i.e., the $(1\pm\eps)$ factor approximations to the sampling probabilities and the corresponding approximations of weights are readily in a form that approximates the entire loss function in the common case where it is simply a summation of single loss functions.
We have the following theorem, which uses a data dependent parameter $\mu$ that is standard in the analysis of asymmetric loss functions \citep{MunteanuSSW18,MunteanuOP22}. 
\begin{theorem}\label{lem:applicationsmain}
    Let $A\in\mathbb{R}^{n\times d}$ be $\mu$-complex (see \cref{def:mu_complex}). 
    Given a leverage score sampling algorithm that constructs an $\eps$-coreset of size $k$, as for the loss functions below (summarized in \cref{pro:leveragescoresampling} in \cref{app:application}), there exists a sampling algorithm that works in the turnstile stream setting that with constant probability outputs a weighted $2\varepsilon$-coreset $(A', w)\in \mathbb{R}^{k' \times d} \times \mathbb{R}_{\geq 1}^{k'}$ of size $k'=\Theta(k)$, such that
    \[
        \forall z  \in \mathbb{R}^d\colon \left| \sum_{i \in [k']}w_i g(a_i' z ) - \sum_{i=1}^n g(a_i z ) \right|\leq 2\varepsilon \sum_{i=1}^n g(a_i z ).
    \]
    The size of the sketching data structure used to generate the sample is $r \cdot s$, where $s=3\ln(36 n/\delta)$ and
    \[ 
        r=
        \begin{cases}
            O\left(k \ln(n) (\alpha^p \beta^p/\varepsilon)^p\right) & \text{if $g(t)=|t|^p $},\\
            O\left(k \ln(n) (\mu \alpha^p \beta^p/\varepsilon)^p\right) &\text{if $g(t)=\max\{0, t\}^p,$}\\
            O\left(k \ln(n) (\mu \alpha \beta/\varepsilon)\right) & \text{if $g(t)=\ln(1+e^t)$,}\\
            O\left(k \ln(n) (p \mu^2 \alpha^p \beta^p/\varepsilon)^p\right) &\text{if $g(t)=-\ln(\Phi_p(-t))$,}
        \end{cases}    
    \]
    where $\Phi_p\colon \mathbb{R} \rightarrow [0,1]$ denotes the CDF of the $p$-generalized normal distribution.
    In particular if the matrix $P:=R^{-1}$ of \cref{pro:Rmatrixmain} is used in \cref{alg:turnstilesampling}, then the overhead is at most $O(\ln(n) (p\mu^2 \alpha^p \beta^p/\varepsilon)^p)=\operatorname{poly}(\mu d/\eps)\log(n)$.
\end{theorem}

We would like to add that our algorithm serves as a general framework, that in principle extends beyond the loss functions presented in \cref{lem:applicationsmain}. It likely works for any loss function which is close to the $\ell_p$ norm.\footnote{A known limitation is that $p>2$ would imply $\tilde\Omega(n^{1-2/p})$ sketch size, although the final sample can be small again.} In particular, any off-line $\ell_p$ leverage score algorithm can be simulated with little overhead. If one could access the original rows $a_i$ for $i$ in the sample, our algorithm serves as a generic black-box. But to work with the approximated samples $\tilde a_i$ one needs to show additionally and on a case-wise basis that the loss function is robust to their perturbation. This last item limits \cref{lem:applicationsmain} to the presented loss functions, since we have proven the robustness property only for those four functions as exemplary applications.

We further note that any improvement of conditioning parameters $\alpha,\beta\in d^{O(1)}$ will reduce the overhead. Additionally, the analysis takes an established subsample size $k$, possibly depending on $d$, and adds $d^{O(1)}$ overhead for the turnstile simulation. Thus, our work conditions the turnstile result on readily available off-line subsampling and matrix conditioning results. It might save some $d$ dependence if all analyses were integrated more directly.

\section{Experimental illustration}
\begin{figure*}[ht!]
\begin{center}
\begin{tabular}{cccc}
{~}&
{\small\hspace{.5cm}\textsc{CoverType}}&{\small\hspace{.5cm}\textsc{WebSpam}}&
{\small\hspace{.5cm}\textsc{KddCup}}\\
\rotatebox{90}{\hspace{6pt}\textsc{Logistic Regression}}&
\includegraphics[width=0.2951\linewidth]{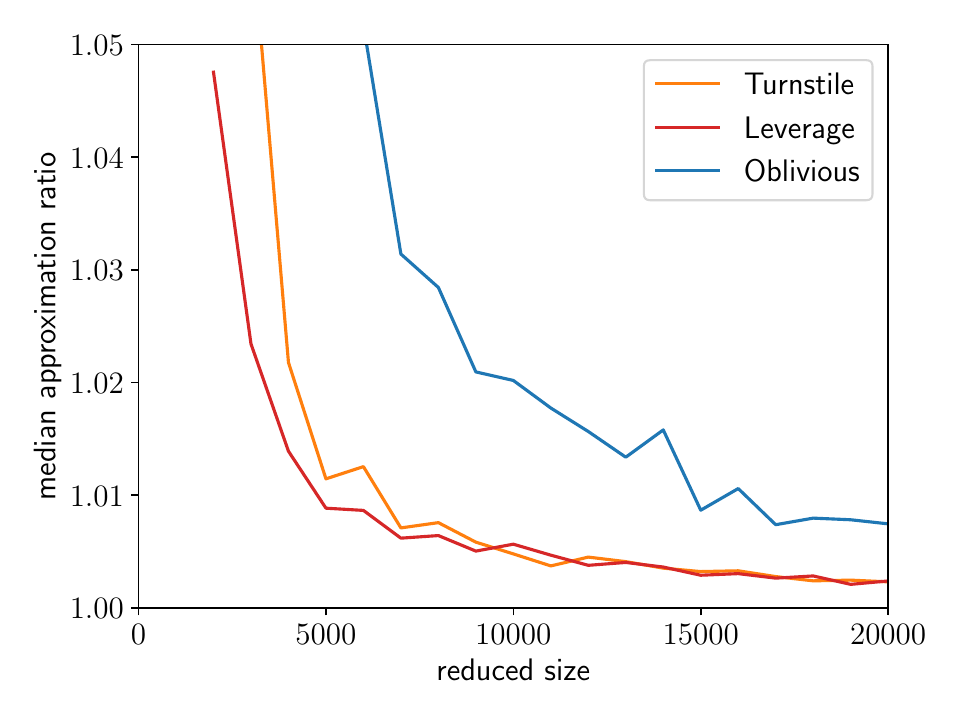}&
\includegraphics[width=0.2951\linewidth]{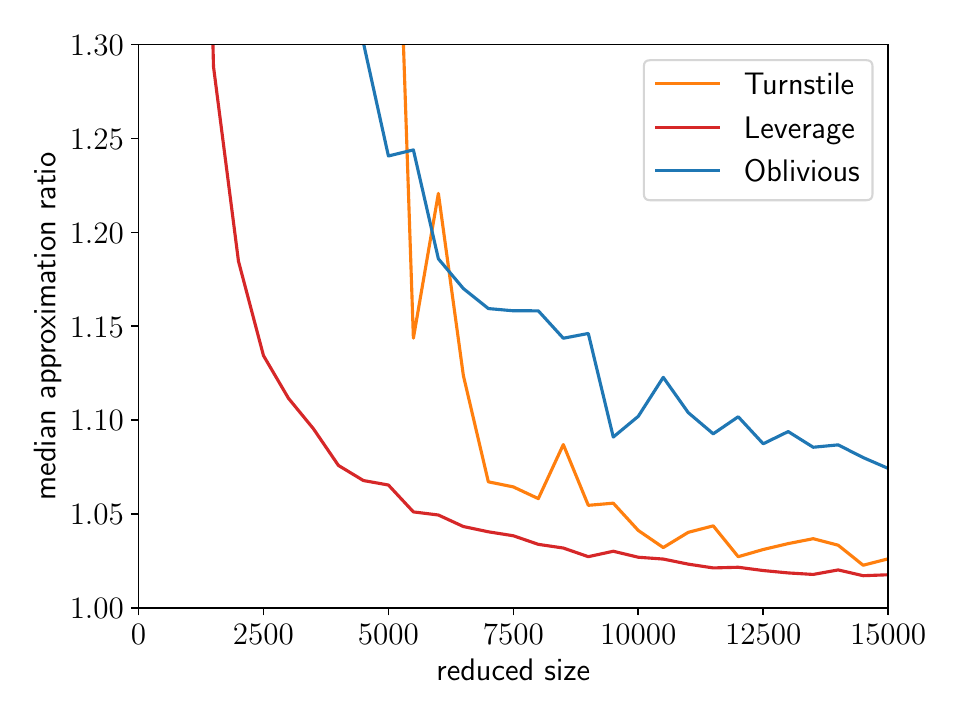}&
\includegraphics[width=0.2951\linewidth]{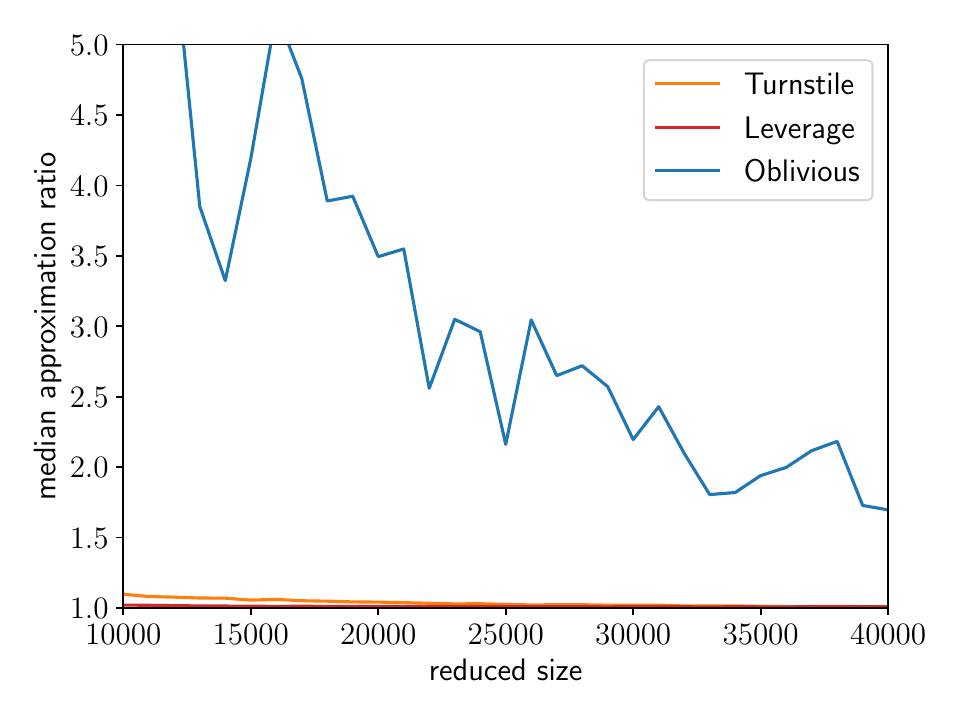}\\
\rotatebox{90}{\hspace{20pt}$\ell_1$ \textsc{Regression}}&
\includegraphics[width=0.2951\linewidth]{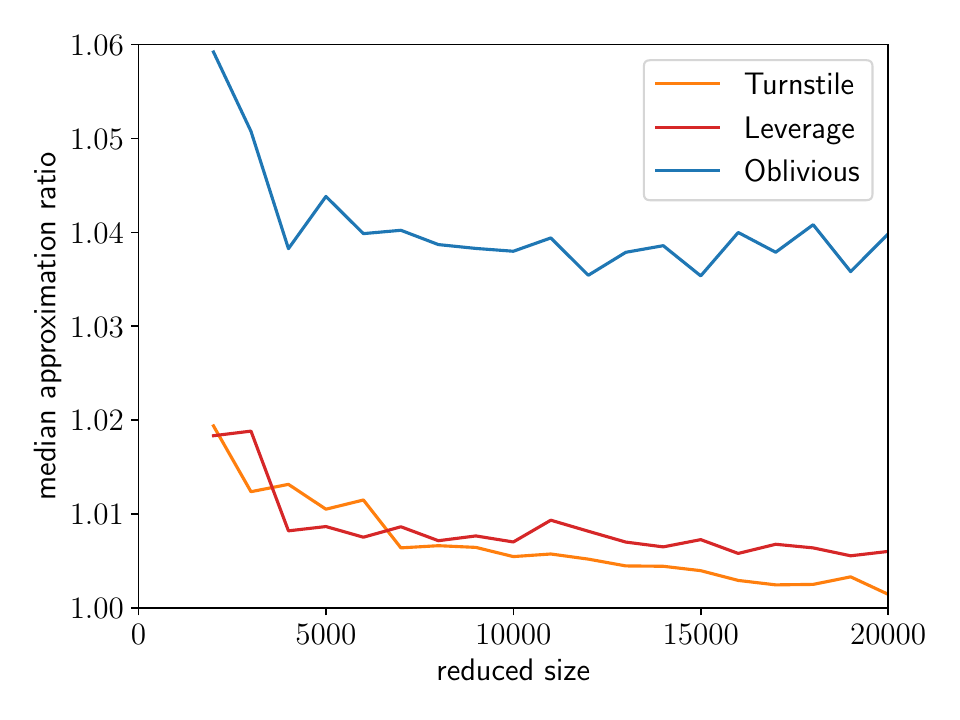}&
\includegraphics[width=0.2951\linewidth]{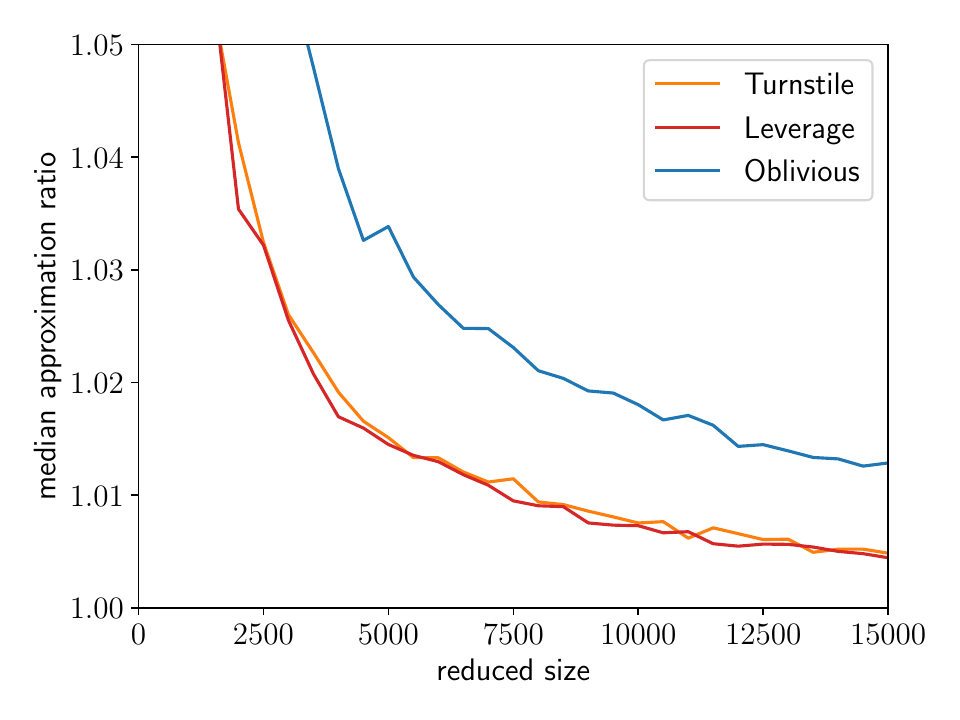}&
\includegraphics[width=0.2951\linewidth]{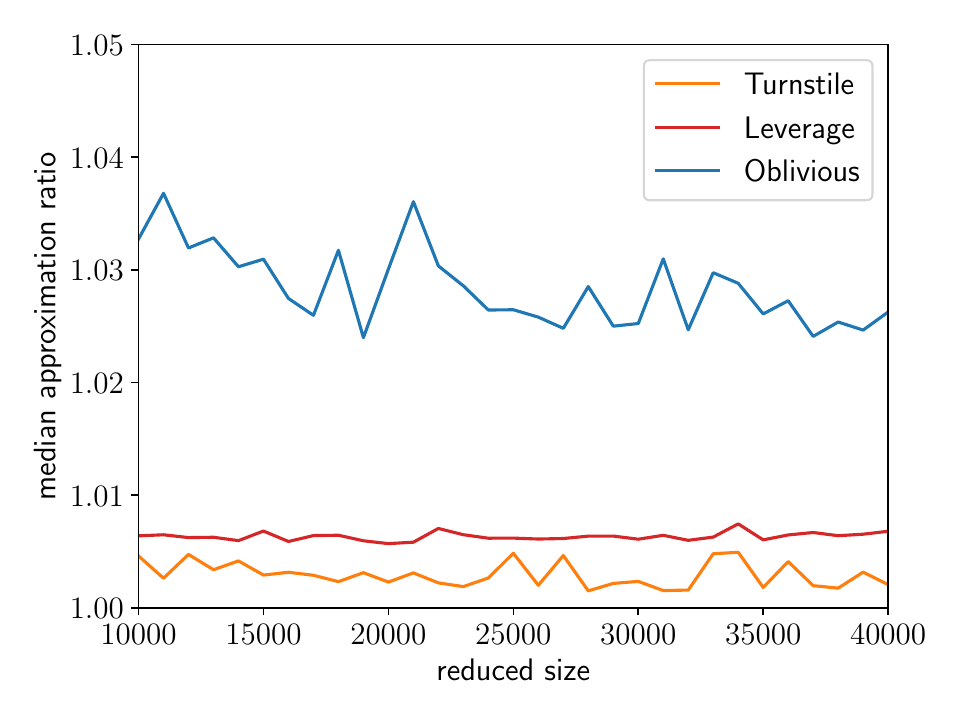}\\
\end{tabular}
\caption{Comparison of the approximation ratios for logistic regression, and $\ell_1$ regression on various real-world datasets. The new turnstile data stream sampler (orange) is compared to plain leverage score sampling (red), and to plain oblivious sketching (blue). The plots indicate the median of approximation ratios taken over 21 repetitions for each reduced size. Best viewed in colors, lower is better.
}
\label{fig:experiments}
\end{center}
\end{figure*}
We demonstrate the performance of our novel turnstile $\ell_p$ sampler. 
Recall, that our algorithm is a hybrid between an oblivious sketch and a leverage score sampling algorithm. It thus makes most sense to compare to pure oblivious sketching as well as to pure off-line leverage score sampling. To this end, we implement our new algorithm into the experimental framework of the near-linear oblivious sketch of \citet{MunteanuOW23}, and add the code of \citet{MunteanuOP22} for $\ell_1$ leverage score sampling.\footnote{Our new code is available at \url{https://github.com/Tim907/turnstile-sampling}.} 

Our a priori hypothesis from the theoretical knowledge on the three regimes is that the performance should be somewhere in the middle between the performances of the competitors. Ideally, we would want our algorithm to perform as closely as possible to off-line leverage score sampling.

The following real-world datasets have become standard baselines to measure the performance of data reduction algorithms for logistic regression and $\ell_1$ regression: Covertype, Webspam, and KDDCup, see \cref{sec:app:experiments_data} for details.
For each dataset, and each of the two problems, we first solve the original large instance to optimality to obtain $z_{opt}$. We then run the data reduction algorithms, for varying target coreset resp. sketch sizes, and solve the reduced and reweighted problem to optimality to obtain the approximation $\tilde z$. For each target size, we repeat this process $21$ times and plot in \cref{fig:experiments} the median of the resulting approximation ratios $f(\tilde z)/f(z_{opt})$. {We experienced convergence problems using the \texttt{scipy} optimizer for the non-differentiable $\ell_1$ loss. Thus, for $\ell_1$ regression, $z_{opt}$ denotes the best (though not necessarily optimal) solution found.}
The results are consistent across all settings: our new turnstile sampler outperforms pure oblivious sketching by a large margin. Its performance lies between the two competitors and is very close to off-line leverage score sampling. In some cases, it even performs slightly better for $\ell_1$ regression, which is likely due to the reported inaccuracies of the \texttt{scipy} optimizer, rather than the reduction algorithms.

The experiments affirm our hypothesis, and corroborate the usefulness of our novel turnstile $\ell_p$ leverage score sampling sketch in practical applications.
We refer to \cref{sec:app:experiments} for more experiments using $p=1.5$, and a mixture of $\ell_1+\ell_2$ leverage scores, as well as details on data, computing environment, running times, and memory requirements.

\section{Conclusion}
We generalize the turnstile $\ell_2$ row sampling algorithm of \citet{Mahabadi2020} to work for all $p\in[1,2]$ using novel statistical tests that rely only on the CountSketch data structure, rather than requiring auxiliary or $p$-specific sketches. This is used to simulate $\ell_p$ leverage score sampling over a turnstile data stream.
The combination of different $\ell_p$ distributions and uniform sampling extends our methods to logistic regression and $\ell_p$ generalizations of linear, ReLU, and probit regression losses. Our experiments show good performance for $\ell_p$ and logistic regression as compared to pure oblivious sketching and off-line sampling.
The most intriguing open question is whether it is possible to simulate the construction of $\ell_p$ spanning sets \cite{WoodruffY23lpSE,BhaskaraMV23} in turnstile data streams, which would bring larger powers of $d$ down to near-optimal linear dependence \cite{MunteanuO24optimallpsampling}.

\section*{Acknowledgements}
The authors would like to thank the anonymous reviewers of ICML 2024 for very valuable comments and discussion. We also thank Tim Novak for helping with the experiments. This work was supported by the German Research Foundation (DFG), grant MU 4662/2-1 (535889065), and by the Federal Ministry of Education and Research of Germany (BMBF) and the state of North Rhine-Westphalia (MKW.NRW) as part of the Lamarr-Institute for Machine Learning and Artificial Intelligence, Dortmund, Germany. Alexander Munteanu was additionally supported by the TU Dortmund - Center for Data Science and Simulation (DoDaS).

\bibliography{arxivmain}
\bibliographystyle{apalike}

\newpage
\appendix
\onecolumn
\allowdisplaybreaks
\section{Preliminaries}
We are given a data matrix $A \in \mathbb{R}^{n \times d}$ with row vectors $a_i, \dots a_n \in \mathbb{R}^d$ presented in a turnstile data stream.
We assume that $n \gg d$.
Further let $p \in [1, 2]$ have a fixed value.
Let $s_i\geq u_i^{(p)}+\frac{1}{n}$ where $u_i^{(p)}=\sup_{x \in \mathbb{R}^d\setminus \{0\}}\frac{|a_ix|^p}{\norm{Ax}_p^p}$ are the $\ell_p$ leverage scores (see \cref{def:leveragescores}).
Our goal is to develop an algorithm that samples row $i$ with probability $p_i\gtrsim \frac{ks_i}{S}$ in one pass over a turnstile data stream and determine weights $w_i \approx \frac{1}{p_i}$.
We allow an error controlled by a parameter $\varepsilon>0$ in both, the sampled vector as well as the weight.

\section{The algorithms}

Our first algorithm (\cref{alg:findhh}) determines heavy rows of a matrix $A$. It is a modification of the CountSketch \citep{CharikarCF04}, that performs additional statistical tests on $s$ repetitions of the sketch to 1) determine a suitable threshold $M_0$ using the $0.65$-percentile among the $s$ repetitions, relative to which any row will be considered 'heavy', 2) estimate the $\ell_p$ norm of the current row up to $(1\pm\eps)$ error using the median among the $s$ repetitions, and compares the estimate to the threshold, and 3) find a representative element among the $s$ repetitions using the median again, to find an approximation of the row that lies close to most other approximations. This will ensure that it also lies close to the original input row, which it represents. See the main text for more details.

Our second algorithm (\cref{alg:turnstilesampling}) multiplies random scaling factors $t_i^{-1/p}$, where $t_i \sim {U}(0, 1)$ to the rows of a matrix $A$ to get a new Matrix $A'=TA$, where $T=\operatorname{diag}(t_1^{-1/p},\ldots,t_n^{-1/p})$ is a diagonal $n \times n$ matrix. Then \cref{alg:findhh} is applied to determine the heavy rows of $A'P$. Hereby $A$ is presented in a turnstile data stream, and $P$ is a conditioning matrix that is obtained in a postprocessing step after the stream has reached its end. This can be done using another turnstile sketching primitive applied to the stream that represents $A$ in parallel to our algorithm. The postprocessing step is then completed by right-multiplication of our sketch with $P$ (in most of our analysis $P=I$; other choices are discussed later in the applications of \cref{sec:applications}).
If $r$ and $s$ are sufficiently large, then we can guarantee that $A'$ has at least a certain number of heavy rows, the (roughly) $k$ largest of which are back transformed to their original sign, scale and basis, and returned as an approximate sample $S$ together with estimated sampling probabilities. This is done by calculating a threshold $\alpha$ which is the smallest approximated $\ell_p$ norm of the $k$ largest elements.
For $(i, \at_i P^{-1}, w_i) \in S $ the first entry is the index of a row $a_i$ of $A$, the second entry is a slightly perturbed row $\at_i P^{-1} \approx a_i$, and the third entry is a weight which is roughly the inverse of the sampling probabilities $p_i\approx\min\{1,{\norm{\at_i}_p^p}/{\alpha}\}\approx\min\{1,{\norm{a_i P}_p^p}/{\norm{A P}_p^p}\}$.

\section{Outline of the analysis}

\begin{itemize}
    \item[1)] We first prove some technical lemmas that are used multiple times and give intuitions about how parts of the analysis work.
    In particular, we analyze sums of Bernoulli random variables, medians and other percentiles, as well as the expected $\ell_p$ norm of a random bucket.
    \item[2)] We analyze \cref{alg:findhh}.
    Here, we show that there is an upper bound for $M_0$ which guarantees that it finds and returns all 'heavy' rows.
    Further, we show that there is a lower bound for the threshold $M_0$, which guarantees that any element returned by the algorithm is approximated up to a relative error of $\varepsilon$.
    \item[3)] We then proceed by analyzing a slightly modified version of \cref{alg:turnstilesampling} (see \cref{mod: alg2} for details).
    We first give a high level intuition of how the algorithm works.
    We prove that the probability of sampling row $i$ is greater
    or equal to $(1 - \varepsilon){\norm{a_i P}_p^p}/{\alpha} \approx c\cdot {k\norm{a_i P}_p^p}/{\norm{A P}_p^p}$ for an appropriate $\alpha$ (and constant $c$) and that the number of samples is in the interval $[k, 2k]$.
    We then use the properties proven in 2) to show that the norm of each row is approximated up to a relative error of $\varepsilon$.
    Finally, we analyze the weights for which we show that they are roughly the inverse sampling probabilities and that they can be used to approximate $\norm{AP}_p^p$ up to a factor $(1\pm\eps)$.
    \item[4)] We show that if we can sample from two distributions $p_i, p_i'$, we can also sample from a joint distribution where the sampling probability is roughly $\frac{p_i + p_i'}{2}$.
    In particular, we use this to combine \cref{alg:turnstilesampling} with uniform sampling to sample with probability proportional to $ \frac{\norm{a_iP}_p^p}{\norm{AP}_p^p}+\frac{1}{n}$.
    \item[5)] We show how our results can be applied to construct an $\eps$-coresets for the $\ell_p$ variants of linear regression, ReLU regression, probit regression, as well as logistic regression. 
\end{itemize}

\section{Tools for the analysis}

Let us start with some facts following from well known results of probability theory.
The first fact is about the median of Bernoulli random variables.
The lemma will be crucial for arguments regarding the median or other percentiles and to obtain bounds on the number of samples.

\begin{lemma}\label{lem: bernolem}
    Let $m \in \mathbb{N}$ and $0 < \delta <1$.
    Let $X_1, \dots, X_m$ be a sequence of independent Bernoulli random variables with $P(X_i=1)=p>0.075$.
    If $m \geq  3\ln(2/\delta)/0.025^3$ then with probability at least $1 -\delta$ it holds that $X=\sum_{i=1}^m X_i = |\{ i \mid X_i=1 \}| \in (1 \pm 0.025)pm$.
\end{lemma}

\begin{proof}
    Let $X=\sum_{i=1}^m X_i $ be the number of $1$'s in $ \{ X_1 \dots X_m\}$.
    Since $X$ is a sum of Bernoulli random variables, the expected value of $X$ equals $\mathbb{E}(\sum_{i=1}^m X_i)=pm$.
    By Chernoff's bound it holds that 
    \begin{align*}
        P(|X - pm| > 0.025 pm ) \leq 2\exp\left( -\frac{0.025^2 pm}{3} \right)
    \leq 2\exp\left( -\frac{0.025^3 m}{3} \right)\leq \delta.
    \end{align*}
\end{proof}

The next lemma is similar to the previous one but handles Bernoulli random variables with small expected sum.

\begin{lemma}\label{lem: altbernolem}
    Let $m \in \mathbb{N}$ and $0 < \delta <1$.
    Let $X_1, \dots, X_m$ be a sequence of independent Bernoulli random variables with $P(X_i=1)=p_i$ and let $k\geq 20\ln(2/\delta)$.
    If $\E(X)\leq 9k $ then with probability at least $1 -\delta$ it holds that 
    \[
        X=\sum\nolimits_{i=1}^m X_i = |\{ i \mid X_i=1 \}| \in [\E(X)-k, \E(X)+k].
    \]
\end{lemma}

\begin{proof}
    We will prove this by using Bernstein's inequality.
    First, note that $\E(\sum_{i=1}^m X_i^2)=\E(\sum_{i=1}^m X_i)=\E(X)\leq 9k$ since $X_i$ are Bernoulli random variables.
    Second, note that $X_i \leq 1$.
    Thus using Bernstein's inequality we get that
    \[
        P( |X-\E(X)| \geq k )\leq 2\exp \left(- \frac{ k^2/2 }{ \E(X)+k/3 } \right)
        \leq 2\exp \left(- \frac{ k }{20 } \right) \leq \delta.
    \]
\end{proof}

An important property of a sum with random signs is that it preserves the $\ell_2$ norm of the entries.
The following lemma uses this fact and shows the relation of the expected value of the $p$th power of a sum with random signs over the elements of a vector $v$ to its $\ell_p$ norm $\norm{v}_p^p $.

\begin{lemma}\label{lem: bucketcontbound}
    Let $v_1, \dots v_n \in \mathbb{R}^d$ and let $\sigma_1, \ldots, \sigma_n \in \{-1, 1\} $ be uniform and pairwise independent random signs.
    If $p\leq 2$ then it holds that $\mathbb{E}(\norm{\sum_{i=1}^n \sigma_i v_i }_p^p)\leq \sum_{i=1}^n \norm{v_i}_p^p$.
\end{lemma}

\begin{proof}
    First note that for uniform and pairwise independent random signs we have that
    \begin{align*}
        \E\left( \left\|\sum_{i=1}^n \sigma_i v_i \right\|_p^p \right)=\E\left(\sum_{j=1}^d  \left| \sum_{i=1}^n   \sigma_i v_{ij}  \right |^p \right)
        = \sum_{j=1}^d \E \left( \left| \sum_{i=1}^n   \sigma_i v_{ij}  \right |^p \right).
    \end{align*}
    Khintchine's inequality \citep[see][]{Haagerup1981} followed by the standard inter-norm inequality yield
    \[ \mathbb{E}\left(\left|\sum_{i=1}^n \sigma_i v_{ij} \right|^p\right)\leq \norm{v^{(j)}}_2^p \leq \norm{v^{(j)}}_p^p \]
    where $v^{(j)} \in \mathbb{R} $ is the vector with coordinates $v_{ij}$ for $i \in [n]$. 
    Combining the previous two inequalities we get that
    \[
        \E\left( \left\|\sum_{i=1}^n \sigma_i v_i \right\|_p^p \right) \leq \sum_{j=1}^d \norm{v^{(j)}}_p^p
        = \sum_{i,j \in [n] \times [d]} |v_{ij}|^p=\sum_{i=1}^n \norm{v_i}_p^p
    \]
\end{proof}

\paragraph{Some notation} Consider a bucket $B$ consisting of a set of indices together with the corresponding set of random signs we define $G_p(B)=\norm{\sum_{i \in B}\sigma_i a_i }_p$. The specific signs $\sigma_i=\sigma_{i,j}, j\in[s]$ will be clear from the context.

\section{Analysis of \texorpdfstring{\cref{alg:findhh}}{}}

\paragraph{High level idea} For $k \in [n]$ let $S_L(k, A)\subset [n] $ be the subset of the $k$ indices of elements with the largest $\ell_p$ norm (ties are broken arbitrarily) and let $S_R(k, A)=[n] \setminus S_L(k, A)$ be the subset of the remaining indices.
If $A$ is clear from the context we simply write $S_L(k) $ and $S_R(k)$. If $k$ is also clear from the context we just write $S_L$ and $S_R$.

The idea of \cref{alg:findhh} is that if we hash the elements to $r$ buckets, then for $k=r/20$, at least $r-r/20$ buckets, do not contain any large element of $S_L(k)$.
Further the expected squared $\ell_2$ norm of a bucket is $M/r$ for $M=\sum_{i \in S_R(k)} \norm{a_i}_p^p$.
Using \cref{lem: bucketcontbound} and the union bound we can extend this result showing that with probability at least $1-1/4-1/20$, the contribution of a bucket $B$ is $G(B)^p \leq 4M/r$.

The argument can also be applied to the buckets containing a certain index $i$, i.e., if we consider a bucket $B_i$ containing the element $i$ then with probability at least $1-1/4-1/20$ we have that $\norm{B_i-\sigma_i a_i}_p^p=\norm{\sum_{j \in B\setminus \{ i\}} \sigma_j a_j }_p^p \leq 4M/r$.
Thus if $\norm{a_i}_p^p \gtrsim \frac{ M}{ \varepsilon^p r}$
then most of the buckets containing element $i$ will be close to $a_i$ and using the median, which is the approximation $\at_i$ calculated by \cref{alg:findhh}, we can approximate the large elements exceeding a fraction of $\gamma M$ up to an error of $\varepsilon$ with respect to their $\ell_p$ norm by setting $r=O(\frac{1}{\gamma \varepsilon^p})$.

In addition to the definitions given in the high level idea, we define $$M'=\inf \{ w \in \mathbb{R}_{\geq 0} \mid P(G(B)^p \leq w)\geq 0.6 \}$$ to be the (theoretical) $.6$-percentile of the $\ell_p^p$ norm contributions of buckets.
The following Lemma yields an upper and a lower bound for $M_0$:

\begin{lemma}\label{lem: Mbound}
    If $s\geq 3\ln(2/\delta)/0.025^3$, then the value of $M_0$ in \cref{alg:findhh} satisfies
    \[
        M' \leq M_0 \leq 4 M/r
    \]
    with failure probability at most $2\delta$.
\end{lemma}

\begin{proof}
    Let $S_L=S_L(r/20)$ be the set of the $r/20$ indices with the largest $\ell_p$ norm and $ S_R=[n]\setminus S_L$.
    Let $M=\sum_{i \in S_R} \norm{a_i}_p^p$.
    Consider any bucket $B$.
    The probability that $B$ contains any specific element is $1/r$. By a union bound, the probability that $B$ contains an element of $S_L$ is bounded by $P(B \cap S_L \neq \emptyset) \leq r/20 \cdot 1/r = 1/20$.
    Further denoting by $P(S)$ for a set $S$ the probability that $S=B \setminus S_L$ and using \cref{lem: bucketcontbound} it holds that
    \begin{align*}
        \E (G_p(B \setminus S_L)^p)&=\sum_{S\subset S_R} P(S)\E\left(\left\Vert \sum_{i\in S} \sigma_i a_i \right\Vert_p^p\right) \leq \sum_{S\subset S_R} P(S) \left( \sum_{i\in S} \norm{a_{i}}_p^p \right).
    \end{align*}
    Now, by double counting the last term, we also have that
    \[
         \sum_{S\subset S_R} P(S) \left( \sum_{i\in S} \norm{a_i}_p^p \right) 
         = \sum_{i\in S_R} \norm{a_i}_p^p \left(\sum_{S\subset S_R, i \in S} P(S) \right)
         = \sum_{i \in S_R} \norm{a_i}_p^p \cdot  P(i \in B)
         = \sum_{i \in S_R} \frac{1}{r}\cdot\norm{a_i}_p^p = M/r.
    \]
    Thus using Markov's inequality we have that $ G_p(B \setminus S_L)^p \leq 4M/r$ with probability at least $1-1/4$.
    Using the union bound we have that with probability at least $1-1/4-1/20=0.7$, an arbitrary bucket $B$ contains no element of $S_L$ and $G_p(B \setminus S_L)^p \leq 4M/r$.

    Since $s \geq 3\ln(2/\delta)/0.025^3$, \cref{lem: bernolem} implies that at least $0.675\cdot s$ many random buckets satisfy these properties with failure probability at most $\delta$, so in particular this holds for the (realized) $.65$-percentile $M_0$. We conclude that $M_0 \leq 4M/r$.

    The lower bound also follows by \cref{lem: bernolem} for $s \geq 3\ln(2/\delta)/0.025^3$, which implies that the (theoretical) $.6$-percentile is not exceeded by more than $.025$. Specifically, this yields $|\{ j\in [s] \mid G(B_{j, 1})^p \leq M' \}| \leq 0.625 s$. Consequently the (realized) $0.65$-percentile $M_0$ is larger than $M'$. The failure probability is again bounded by at most $\delta$, and the overall failure probability is bounded by $2\delta$ by another union bound, which concludes the proof.
\end{proof}

In the following lemma, these bounds will be used to show that with high probability all elements in the output $L$ of \cref{alg:findhh} are close to the original rows. Further it shows that all rows with large $\ell_p$ norm will be in $L$.

\begin{lemma} \label{lem: L-Lemma}
    If $s \geq 3\ln(2n/\delta)/0.025^3$, $r \geq 50$, $0 < \varepsilon \leq 1/3$ and $M' \leq M_0$, then the following holds with failure probability at most $\delta$:
    For any $i  \in [n]$ with $v_i \geq (12/\varepsilon)^p M' $ it holds that $\norm{a_i}_p^p \geq (3/\varepsilon)^p M' $.
    Further, for any $i  \in [n]$ with  $\norm{a_i}_p^p \geq (3/\varepsilon)^p M'$ it holds that $v_i=(1\pm \varepsilon)\norm{a_i}_p^p$.
    In particular, this implies that for any $i$ with $ \norm{a_i}_p^p \geq (12/\varepsilon)^p M'/(1-\varepsilon)$ it holds $i \in L$.
    Finally, it holds for $S_i:=\{j \in [s] \mid \norm{\at_{i, j}-a_i}_p\leq \varepsilon\|a_i\|_p/9\}$ that $|S_i|\geq s/2$.
\end{lemma}
\begin{proof}
    By \cref{lem: Mbound}, it holds that $M' \leq M_0$ with probability $1-\delta$.
    
    We show the first claim by contraposition: rows $a_i$ with small norms, i.e., $\norm{a_i}_p^p < (3/\varepsilon)^p M'$ will not be part of the output $L$.     
    Fix $i \in [n]$ and for each repetition $j \in [s]$ let $B(i, j)$ be the bucket that contains $i$. We set $b_{i, j} = \sum_{l \in B(i, j) \setminus \{i\}}\sigma_{l, j}a_l$ to be the content of the bucket after sketching all data, but with the contribution of $a_i$ removed.
    We set
    $$M''=\inf \{ w \in \mathbb{R}_{\geq 0} \mid P(G(B \setminus\{i\})^p \leq w)\geq 0.575 \}.$$
    Note that for any bucket $B$ it holds that $P(i \in B)=1/r \leq 0.02 $. Thus, we have that
    \[
        P(G(B \setminus\{i\})^p \leq M')\geq P(G(B)^p \leq M')-P(i \notin B) \geq 0.58 > 0.575.
    \] 
    and consequently $M'' \leq M' $.
    
     By definition of the $.575$-percentile $M''$ and applying \cref{lem: bernolem}, we get that
    $$ \norm{b_{i, j}}_p^p \leq M'' \leq M'$$
    holds for at least half of the indices of $j \in [s]$ up to failure probability at most $\delta/n$ which will be assumed in the remainder of the proof.
         
    For all $i$ and $j$ that satisfy $\norm{b_{i, j}}_p^p \leq M'$, we have that
    \begin{align*}
        G(B(i, j))^p &=\norm{\sigma_{i, j}a_i+ b_{i, j}}_p^p
        \leq (\norm{a_i }_p + M'^{1/p})^p\\
        &\leq (2\max\{\norm{a_i }_p, M'^{1/p} \})^p \leq \max\{4\norm{a_i }_p^p , 4M' \}.
    \end{align*}
    Then it also holds that $v_i=\median_{j\in [s]} \norm{\at_{i, j}}_p^p \leq \max\{4\norm{a_i }_p^p , 4M' \}$.
    Thus, we can conclude that if index $i$ satisfies $\norm{a_i}_p^p <  (3/\varepsilon)^p M' \leq (3/\varepsilon)^p M_0$ then it holds that 
        \[
            v_i < \max\{(12/\varepsilon)^p M_0   , 4M' \} \leq  (12/\varepsilon)^p M_0
        \]
   and consequently $i \notin L $.
    
   Next, we show that rows with larger norm $\norm{a_i}_p^p \geq  (3/\varepsilon)^p M' $ are well approximated assuming that $ \norm{b_{i, j}}_p^p \leq M'$. Let $\gamma:=\frac{M'}{\norm{a_i }_p^p} \leq (\varepsilon/3)^p $. 
   Then by the triangle inequality it holds that
    \begin{align*}
       G(B(i, j))^p=\norm{\sigma_{i, j}a_i+ b_{i, j}}_p^p
       \leq (1+ \gamma^{1/p})^p \norm{a_i}_p^p
       \leq (1+3\gamma^{1/p})\norm{a_i}_p^p
       \leq (1+\varepsilon)\norm{a_i}_p^p
   \end{align*}
   and similarly we have
   \begin{align*}
        G(B(i, j))^p=\|\sigma_{i, j}a_i+ b_{i, j}\|_p^p
       \geq (1-3\gamma^{1/p})\norm{a_i}_p^p
       \geq (1-\varepsilon)\norm{a_i}_p^p.
   \end{align*}
    Since $ \|b_{i, j}\|_p^p \leq M'$ holds for at least half of the indices $j \in [s]$ we can conclude that 
    \[
        v_i=\median_{j \in [s]} \norm{\at_{i, j}}_p^p \in \left[(1-\varepsilon)\norm{a_i}_p^p, (1+\varepsilon)\norm{a_i}_p^p\right]. 
    \]

    Finally, we show that for $i$ with $i \in L$ it holds that $ \norm{\at_{i,j}-a_i}_p \leq (\varepsilon/9)\norm{a_i}_p$ and that $|S_i|\geq s/2$.
    Using that $\varepsilon \leq 1/3$ we have for $i \in [n]$ with $v_i \geq (12/\varepsilon)^p M_0$ that
    \[
        \norm{\at_{i,j}-a_i}_p^p = \norm{b_{i, j}}_p^p \leq M' \leq M_0 \leq (\varepsilon/12)^p\,v_i
        \leq (\varepsilon/12)^p(1+\varepsilon)\norm{a_i}_p^p\leq (\varepsilon/9)^p\norm{a_i}_p^p
    \]
    which also yields
    \[
        |\{ j \in [s]  \mid \norm{\at_{i, j}-a_i}_p\leq \varepsilon\|a_i\|_p/9 \} | \geq s/2.
    \]

    By the union bound, these properties hold for all $i$ simultaneously with probability at least $1-O(\delta)$. Rescaling $\delta$ by a constant concludes the proof.
\end{proof}

We are now ready to prove that \cref{alg:findhh} works as intended for the right choice of $r$ and $s$:

\begin{theorem}[copy of \cref{thm_main:findhh}]\label{thm: findhh}
    Let $\eps,\delta\in(0,1/20],\gamma\in(0,1)$. Let $L$ be the list of tuples in the output of \cref{alg:findhh}. Further let $S_R(r/20)$ be the subset of rows excluding the $r/20$ largest $\ell_p$ norms and let $M=\sum_{i \in S_R} \norm{a_i}_p^p$.
    If $r = 8\gamma^{-1} \cdot (12/\varepsilon)^p$ and $s \geq 3\ln(6n /\delta)/0.025^3$ then with probability at least $1-\delta$, the following properties hold: for any element $(i, \at_i) \in L $ it holds that $\norm{\at_i-a_i}_p\leq (\varepsilon/3)\norm{a_i}_p$ and $\norm{\at_i}_p^p=(1 \pm \varepsilon)\norm{a_i}_p^p$.
    Further, for any $i \in [n]$ with $\norm{a_i}_p^p \geq \gamma M$ it holds that $i \in L$.
\end{theorem}

\begin{proof}[Proof of \cref{thm_main:findhh}/\ref{thm: findhh}]
    The statements of \cref{lem: Mbound} and \cref{lem: L-Lemma} hold with failure probability at most $\delta= 2(\delta/3)+ (\delta/3)$ using the union bound.
    Then we have that $M' \leq M_0 \leq 4M/r$ and for any $i \in L$ it holds that $v_i \geq (12/\varepsilon)^p M_0 \geq (12/\varepsilon)^p M' $. \cref{lem: L-Lemma} yields that $v_i=(1\pm \varepsilon)\norm{a_i}_p^p$. 
    For the set $S_i=\{j \in [s] \mid \norm{a_i-\at_{i, j}}_p\leq \varepsilon\norm{a_i}_p/9\}$ we have that $|S_i|\geq s/2$.
    
    For any elements $j,j' \in S_i$ we have $$\norm{\at_{i, j}-\at_{i, j'}}_p \leq \norm{\at_{i, j}-a_i}_p + \norm{a_i-\at_{i, j'}}_p \leq 2\eps\norm{a_i}_p/9$$ by the triangle inequality. It follows that $\median_{j' \in [s]}\{ \norm{\at_{i, j}-\at_{i, j'}}_p \}\leq 2\eps\norm{a_i}_p^p/9$ since $|S_i| \geq s/2$.

    Let $\at_i = \at_{i,j}$ for $j \in [s]$ minimizing $\median_{j' \in [s]}\{ \norm{\at_{i, j}-\at_{i, j'}}_p^p \}$.
    Again since $|S_i| \geq s/2$ there must be at least one element in $j' \in S_i$ with $\norm{\at_{i, j}-\at_{i, j'}}_p \leq 2\varepsilon\norm{a_i}_p/9 $. Using the triangle inequality again we get that
    \[
        \norm{\at_{i}-a_i}_p = \norm{\at_{i, j}-a_i}_p \leq \norm{\at_{i, j}-\at_{i, j'}}_p + \norm{\at_{i, j'}-a_i}_p
        \leq (2\varepsilon/9+\varepsilon/9)\norm{a_i}_p\leq \varepsilon\norm{a_i}_p/3.
    \]
    We note that since $\norm{\tilde{a}_{ij}-a_i}_p\leq {\eps}\norm{a_i}_p/3$ holds, we have by the triangle inequality that
\[
    \norm{\tilde{a}_{ij}}_p^p\leq (\norm{a_i}_p+\norm{\tilde{a}_{ij}-a_i}_p )^p \leq (1+\varepsilon)\norm{a_i}_p^p
\]
and
\[
    \norm{\tilde{a}_{ij}}_p^p\geq (\norm{{a_i}}_p-\norm{\tilde{a}_{ij}-a_i}_p )^p \geq (1-\varepsilon)\norm{a_i}_p^p.
\]
    
    Finally, since $ M' \leq M_0 \leq 4 M/r$, or equivalently $ M'r/4 \leq M_0 r/4 \leq M$, we also have for any $i$ with $\norm{a_i}_p^p \geq \gamma M$ that 
    \[
        \norm{a_i}_p^p \geq {\gamma M}
        \geq \gamma r M_0 /4
    \]
     and thus by \cref{lem: L-Lemma}
    \[
        v_i \geq (1-\varepsilon) \norm{a_i}_p^p\geq \frac{1}{2} \cdot 
        \frac{\gamma r M_0}{4}
        \geq (12/\varepsilon)^p M_0.
    \]
    which implies that $i \in L$. 
\end{proof}

\section{Analysis of \texorpdfstring{\cref{alg:turnstilesampling}}{}}

\paragraph{High level idea} Consider the matrix $A \in \mathbb{R}^{n \times 1}$ consisting of $n$ copies of the row $1$.
If we multiply each row with $t_i^{-1}$ where $t_i \in (0, 1] $ are drawn uniformly at random then what roughly happens is that the new matrix $A'$ with rows $a'_i= a_i/t_i$ consists of the rows $n, n/2, n/3 , \dots , n/(n-1), 1$.
We then have that $\norm{A'}_1= \Theta(n \log(n))$ and the $k$ largest elements of $A'$ are bounded from below by $n/k$.
Or in other words $M=(n \log(n))$ and we want to find all rows with $\ell_1$ norm greater or equal to $n/k$.
If we now apply \cref{alg:findhh} to $A'$ with $r= O( k \log(n) /\varepsilon)$ then all elements with $a_i'\geq {n/k}=\Theta(M/(k\log(n)))$ will be in $L$ with high probability.
The challenge will be to control the randomness of the variables $t_i$ and to generalize the idea to arbitrary instances and different $p$'s.

Instead of analyzing \cref{alg:turnstilesampling} as presented, we analyze a slightly modified version, where \cref{alg:findhh} is applied twice in parallel. The main purpose of the modification is to keep the analysis clean and simple. The presented \cref{alg:turnstilesampling} is likely to have the same properties up to small constant factors but the analysis would require to work with conditional probabilities which only leads to additional technicalities that distract from understanding the main ideas of our algorithm. 

\paragraph{Modification of \cref{alg:findhh}}\label{mod: alg2}
To simplify the analysis, we run \cref{alg:findhh} twice with two independent copies of the scaling random variables $t_i, i\in [n]$. The first copy is used to compute $\alpha$ and the second generates the sample using the value of $\alpha$ from the first copy. This makes the estimate $\alpha$ independent of the sample and avoids purely technical difficulties in the analysis. However, it is likely not necessary and is therefore not presented in the pseudo code. In the first iteration, we use an increased value of $k'=(3/2)k$ and we stop after defining $\alpha$ (line 9). In the second iteration, we skip lines 8-9 and use $\alpha$ from the previous iteration.

We define $S \subseteq L$ to be the set of indices with $\norm{\at_{i}}_p^p \geq \alpha$ returned at the end. We assume that $t_i \in (0, 1] $ are drawn i.i.d. uniformly at random and $A' = TA \in \mathbb{R}^{n \times d}$ is the matrix with rows $a_i'=t_i^{-1/p}a_i$. 

Our main theorem is that given $k \in [n]$ with an appropriate choice of $r, s$ \cref{alg:turnstilesampling} returns a subsample $S \subseteq [n] \times \mathbb{R}^d \times \mathbb{R}_{\geq 1}$ such that $|S|\in [k, 2k] $, index $i$ is sampled with probability at least $\min \{1, \frac{k \norm{a}_p^p}{\norm{A}_p^p}\}$ and for $(i, \at_i, w_i) \in S $ we have that $\norm{\at_i-a_i}_p=(\varepsilon/3)\norm{a_i}_p$ and $w_i=(1\pm \varepsilon)P(i \in S)^{-1}$.
Further we can use the weights to approximate $\norm{A}_p^p $ up to a factor of $(1\pm\varepsilon)$.

\begin{theorem}[copy of \cref{thm_main:alg2}]\label{thm: alg2}
    If we apply the modified version of \cref{alg:turnstilesampling} (see \cref{mod: alg2}) with $0 < \varepsilon,\delta \leq 1/20$, $k \geq 160\ln(12/\delta)$, $r \geq 32 k\ln(n) \cdot (72/\varepsilon)^p$, and $s \geq 3\ln(36n/\delta)/0.025^3$, then with probability at least $1- \delta$ it holds that
    \begin{itemize}
        \item[1)] $|S| \in [k, 2k]$,
        \item[2)] index $ i \in S$ is sampled with probability\\[5pt]\hspace*{33pt}$p_i:= P(i \in S) \geq \min \left\{1, \frac{k \norm{a}_p^p}{\norm{A}_p^p}\right\},$
        \item[3)] if $i \in S$ then $\norm{\at_i-a_i}_p \leq (\varepsilon/3) \norm{a_i}_p $,
        \item[4)] if $i \in S$ then $w_i =(1 \pm \varepsilon)\frac{1}{p_i}$,
        \item[5)] $\sum_{i \in S} w_i \norm{\at_i}_p^p= (1 \pm \varepsilon)\norm{A}_p^p$.
    \end{itemize}
\end{theorem}

To support readability, the proof of \cref{thm_main:alg2}/\ref{thm: alg2} is divided into multiple Lemmas.

Our first Lemma considers the unique number $N(k) \in \mathbb{R}_{\geq 0}$ such that the expected number of elements $i \in [n]$ with $ \norm{a_i'}_p^p \geq N(k) $ is $k$.
The properties that we show in this Lemma will allow to show that the number of elements is $|S|\in[k, 2k]$.
Further it will be used later to show that the largest $2k$ rows of $A'$ have a norm large enough to be in $L$ with failure probability at most $\delta$.
Before we state the lemma, we need to give some more definitions:

Recall that $S_L(k, A) \subseteq [n]$ is the set of indices of the elements with the $k$ largest norms (of $A$) and $S_R(k, A)=[n] \setminus S_L$.
We set $M(A, k):= \sum_{i \in S_R(k, A)}\norm{a_i}_p^p$.

We will show that all indices where $\norm{a_i}_p^p \geq \norm{A}_p^p/k $ will be sampled with probability at least $1-\delta$.
The exact value of $\norm{a_i}_p^p$ does not matter but if it gets large, it makes the analysis more complicated.
Since we want to provide a good understanding of our analysis, instead of assuming that $ \norm{a_i}_p^p \geq \norm{A}_p^p/k$ we define $A(k) \in \mathbb{R}^{n \times d}$ to be the truncated matrix that we get by scaling down the largest rows of $A$ so that all rows $a_i(k)$ of $A(k)$ satisfy $\norm{a_i(k)}_p^p \geq \norm{A(k)}_p^p/k $.

\begin{definition}\label{def:truncate}
    Let $u_k \in \mathbb{R} $ be the solution\footnote{We note that $u_k$ can be computed by scaling down the largest row(s). If there are multiple largest rows, we scale all of them down. $u_k$ exists if and only if the number of non-zero rows is larger or equal to $k$.} of the equation
    \[
        \frac{u_k}{\sum_{i=1}^n \min\{ u_k, \norm{a_i}_p^p \}}=\frac{1}{k}.
    \]
    Then we define $A(k)$ to be the matrix with
    \begin{align*}
        a_i(k)=\begin{cases}
            \frac{u_k^{1/p}}{\norm{a_i}_p}\cdot a_i & \norm{a_i}_p^p > u_k \\
            a_i & \norm{a_i}_p^p\leq u_k.
        \end{cases}
    \end{align*}
\end{definition}
In particular note that all elements $\norm{a_i}_p^p > u_k$ are truncated to $\norm{a_i(k)}_p^p = u_k$.

We already note the following properties of $A(k)$:
it holds that $S_R(A(k), k)=S_R(A, k)$ and $\norm{A(k)}_p^p \leq 2 \cdot \sum_{i \in S_R(A, k/2)}  \norm{a_i}_p^p$.
The first one follows immediately since there can be at most $k$ large rows that contribute $\norm{a_i}_p^p \geq \norm{A}_p^p/k$ and all others remain unchanged. The second claim will be proven in the following lemma.

\begin{lemma}\label{lem: newN(k)lemma}
    For $k \in [n]$ we set $N(k) \in \mathbb{R}_{\geq 0}$ to be the unique number such that the expected number of elements $i \in [n]$ with $ \norm{a_i'}_p^p \geq N(k) $ is $k$.
    Then it holds that
    \[
        \min \left\{\frac{\norm{A}_p^p}{k}, 2 M(A, k/2) \right\} \geq \frac{\norm{A(k)}_p^p}{k} = N(k) \geq  M(A, k)/k.
    \]
\end{lemma}

\begin{proof}
    We first prove that $N(k) =  \frac{\norm{A(k)}_p^p}{k}$. For $i \in [n]$ define the Bernoulli random variable $X_i=1$ if $t_i\leq k \norm{a_i(k)}_p^p/\norm{A(k)}_p^p$ and $X_i=0$ otherwise. Note that $X_i=1$ iff $\norm{a_i'(k)}_p^p = \norm{a_i(k)}_p^p/t_i \geq \frac{\norm{A(k)}_p^p}{k}$. Thus, $X_i=1$ holds with probability $p_i= \min \{1,  k \norm{a_i(k)}_p^p/\norm{A(k)}_p^p \}=  k \norm{a_i(k)}_p^p/\norm{A(k)}_p^p$ by definition of $A(k)$.
    Let $X=\sum_{i=1}^n X_i$.
    Observe that 
    \[
        \E(X)=\sum_{i=1}^n p_i = \sum_{i=1}^n k \norm{a_i(k)}_p^p/\norm{A(k)}_p^p = k.
    \]
    To see this, note that the truncated largest rows satisfy $\norm{a_i(k)}_p^p/\norm{A(k)}_p^p = 1/k$ by \cref{def:truncate}. Therefore their probability equals $p_i=1$. Now, if we increase their norms back to their original size, then the probabilities remain truncated at $1$, and thus do not change. Therefore $\E(X)=k$ holds also for the original matrix $A$. By definition of $N(k)$ we get that $N(k) = \frac{\norm{A(k)}_p^p}{k}.$
    
    Since $S_R(A(k), k)=S_R(A, k)$ it holds that
    \[
        \frac{\norm{A}_p^p}{k} \geq \frac{\norm{A(k)}_p^p}{k} \geq \frac{M(A(k), k)}{k} =\frac{M(A, k)}{k}.
    \]
    Further since $ \frac{\norm{a_i(k)}_p^p}{\norm{A(k)}_p^p}\leq \frac{1}{k}$ we have that
    \begin{align*}
        \sum_{i \in S_L(A(k), k/2)}\norm{a_i(k)}_p^p
        = \sum_{i \in S_L(A(k), k/2)}\norm{A(k)}_p^p \cdot \frac{\norm{a_i(k)}_p^p}{\norm{A(k)}_p^p}
        \leq \sum_{i \in S_L(A(k), k/2)}\norm{A(k)}_p^p \cdot \frac{1}{k}
        = \frac{\norm{A(k)}_p^p}{2}
    \end{align*}
    and consequently $M(A(k), k/2) = \sum_{i \in S_R(A(k), k/2)}\norm{a_i}_p^p  \geq \frac{\norm{A(k)}_p^p}{2}$. We conclude that
    \begin{align*}
        \frac{\norm{A(k)}_p^p}{k} \leq \norm{A(k)}_p^p \leq 2 \cdot M(A(k), k/2)\leq 2 \cdot M(A, k/2).
    \end{align*}
\end{proof}

Our next Lemma shows that if $k$ is large enough then the number of rows with $ \norm{a_i'}_p^p \geq N(k)$ is roughly $k$.

\begin{lemma}\label{lem: numberofN(k)}
    Assume that $k \geq 160\ln(2/\delta)$.
    Then it holds that $||\{ i\in L \mid \norm{a_i'}_p^p \geq N(k)\}|-k|\leq k/8$ with failure probability at most $\delta$.
\end{lemma}

\begin{proof}
    
    For $i \in [n]$ define the Bernoulli random variable $X_i=1$ if $t_i\leq  \norm{a_i}_p^p/N(k)$ and $X_i=0$ otherwise. 
    Let $X=\sum_{i=1}^n X_i$.
    First notice that by definition of $N(k)$ we have that
    \[
        \E(X)=k.
    \]
    By \cref{lem: altbernolem} it holds that $P(|X-k| \geq k/8) \leq \delta $.
\end{proof}

After looking at the heavy hitters and large rows of $A'$ that we would like to sample, we will now show that the total sum $\sum_{S_R(r/20)}\norm{a_i'}_p^p$ is small enough to guarantee that the rows of $A'$ with the $k$ largest norms are in $L$.
When proving that this is indeed the case, we need to take care of one complication. Namely, the expected value of $\norm{a_i'}_p^p=\norm{a_i/t_i^{1/p}}_p^p=\norm{a_i}_p^p/t_i$ is unbounded.
However if we know that $t_i>\max\{ 1/n,\norm{a_i}_p^p/u \}$ for some $u \in \mathbb{R}_{\geq 0}$ then we can bound the expected value of $ \norm{a_i'}_p^p$ by $\ln(n)\norm{a_i}_p^p$ and the variance by $ 2u \norm{a_i}_p^p$.
Using these properties, we can prove that the total contribution of the elements that are not large is bounded by $O(\ln(n))$ times the original value, as already indicated in the introductory example.

The following Lemma shows that with high probability $M(A', 3k)$ is bounded by $O(\log(n) M(A, k) ) $.

\begin{lemma}\label{lem: A'Mound}
    Assume that $k \geq 160\ln(2/\delta)$.
    Set $M=M(A', 3k)=\sum_{i \in S_R(3k, A')}\norm{a_i'}_p^p$ and  $M(A)=M(A, k)= \sum_{i \in S_R(k, A)}\norm{a_i}_p^p$.
    Then it holds that $M \leq 2 \ln(n)M(A)$ with failure probability at most $2\delta$.
\end{lemma}

\begin{proof}
    We define $S_0=\{ i \in [n] \mid t_i < 1/n\}$ and we set $S_1=S_L((5/2)k, A')\cup S_0$ and $S_2=[n] \setminus S_1 $.

    In this proof we assume that we have $\norm{a_i}_p^p=M(A)/k $ for all $i \in S_L(A, k)$:
    If $ \norm{a_i}_p^p<M(A)/k $ then increasing the norm of $a_i$ can only increase $M(A', 3k) $.
    Further if $ \norm{a_i}_p^p> M(A)/k  $ then following argumentation shows that $i \in S_1$ and thus decreasing the norm of $a_i$ has no effect on $S_2$:
    By the upper bound in the first item of \cref{lem: newN(k)lemma} $N(2k) \leq \norm{A}_p^p/(2k) \leq M(A)/k$.
    Further by \cref{lem: numberofN(k)} we have that
    \[
    ||\{ i\in [n] \mid \norm{a_i'}_p^p \geq N(2k)\}|-2k|\leq (2k/8)=k/4
    \]
    with probability at least $1-\delta$.
    Then $S_L((5/2)k, A') \subseteq S_1$ contains all $i \in [n]$ with $\norm{a_i'}_p^p \geq M(A)/k \geq N(2k)$.

    Notice that $\norm{a_i'}_p^p\geq \norm{a_i}_p^p$ and by the above assumption $\norm{a_i}_p^p=M(A)/k$ for all $i \in S_L(A, k)$, we get that $S_L(A, k) \subseteq S_1$ and thus $\sum_{i \in S_2}\|a_i\|_p^p  \leq M(A, k)$.

    Further, note that the expected number of indices $i \in [n]$ with $t_i<1/n$ is smaller than one. By \cref{lem: altbernolem} the number of such indices is bounded above by $k/2$ with failure probability at most $\delta$.
    Thus $|S_0| \leq k/2$ and $|S_1| \leq (5/2)k+k/2=3k$. 

    For $i \in S_2$ define the random variable $X_i=\norm{a_i'}_p^p < M(A)/k =:u$.
    Recall that $X_i= t_i^{-1}\norm{a_i}_p^p$ where $t_i \in (\max\{\norm{a_i}_p^p/u, 1/n\} ,1)$ is drawn uniformly at random as we already know that $t_i > \max\{\norm{a_i}_p^p/u, 1/n\} $ for all $i \in S_2$.
    This implies that
    \begin{align*}
        \E(X_i)&\leq \frac{1}{1-1/n} \cdot\int_{1/n}^1 \norm{a_i}_p^p t^{-1} ~dt 
        \leq (3/2)\norm{a_i}_p^p \Big[\ln(t)\Big]_{1/n}^1
        = (3/2)\norm{a_i}_p^p\ln(n)
    \end{align*}
    for any element in $i \in S_2$.
    Consequently we have for $X=\sum_{i \in S_2} X_i$ that 
    \[
        \E(X)=\sum_{i \in S_2}\E(X_i) \leq \sum_{i \in S_2} (3/2)\norm{a_i}_p^p\ln(n) \leq (3/2) M(A) \ln(n).
    \]
    Further since $ \norm{a_i'}_p^p \leq u$ we have that 
    \begin{align*}
        \E(X_i^2)&=\frac{1}{1-\norm{a_i}_p^p/u} \cdot\int_{\norm{a_i}_p^p/u}^1 \norm{a_i}_p^{2p} t^{-2} ~dt \\
        &\leq \frac{1}{1-\norm{a_i}_p^p/u} \cdot \Big[(-t)^{-1}\Big]_{\norm{a_i}_p^p/u}^1 \norm{a_i}_p^{2p} \\
        &= \frac{1}{1-\norm{a_i}_p^p/u} \cdot\norm{a_i}_p^{p}(u - \norm{a_i}_p^p) \\
        &= \frac{u(u - \norm{a_i}_p^p) }{u-\norm{a_i}_p^p} \cdot\norm{a_i}_p^{p}
        \leq \norm{a_i}_p^{p}u
    \end{align*}
    and thus
    \begin{align*}
        \sum_{i \in S_2} \E( X_i^2)
        \leq \sum_{i \in S_2} \norm{a_i}_p^{p}u
        \leq M(A)u = 2M(A)^2/k
    \end{align*}
    Using Bernstein's inequality with $ t=M(A)/2$ we get that
    \begin{align*}
        P(X \geq 4 M(A) \ln(n))&\leq P(X \geq (3/2) M(A) \ln(n)+t)\\
        &\leq \exp\left( -\frac{t^2/2}{M(A)^2/k+t M(A)/(3k)}\right)\\
        &\leq \exp\left( -\frac{k}{6}\right) \leq \delta.
    \end{align*}
    This shows with the claimed probability that 
    \[4 M(A) \ln(n) > X= \sum_{i \in S_2}\norm{a_i'}_p^p \geq  \sum_{i \in S_R(3k, A')}\norm{a_i'}_p^p = M, \]
    where we have used that $|S_1|\leq 3k$, thus $|S_2|\geq n-3k$, and the right hand side sums over the smallest possible set of $n-3k$ elements. This concludes the proof.
\end{proof}

We do not know the exact value of $a_i'$, but only have access to their sketched approximations $\at_i'$. Thus, we define $\tilde{N}(k)$ to be the unique number such that the expected number of elements $i \in L$ with $ \norm{\at_i'}_p^p \geq \tilde{N}(k) $ is $k$.
The following Lemma shows that there is only a small difference between $N(k)$ and $\tilde{N}(k)  $.

\begin{lemma}\label{lem: NtildeN}
    Let $\varepsilon > 0$ and $k \geq 160\ln(2/\delta)$.
    Further assume that $\norm{\at_i'}_p^p =(1 \pm \varepsilon)\norm{a_i'}_p^p$.
    Then
    \[
        N((1-\varepsilon)k) \geq \tilde{N} (k) \geq N((1+\varepsilon)k).
    \]
\end{lemma}

\begin{proof}
    Let $X_i=1$ if $\norm{a_i}_p^p/t_i \geq \tilde{N} (k) $ and $X_i=0$ otherwise.

    For the inequality $ N((1-\varepsilon)k) \geq \tilde{N} (k)$ notice that by assumption we have that $\norm{\at_i'}_p^p \geq (1 - \varepsilon)\norm{a_i'}_p^p$.
    Let $X_i'=1$ if $\norm{a_i}_p^p/t_i \geq \tilde{N} (k)/(1-\varepsilon)$ and $X_i'=0$ otherwise.
    Note that $P(X_i'=1) \geq P(X_i=1) \cdot (1-\varepsilon)$ and that the probability that $t_i \in (1-\varepsilon, 1) \cdot \frac{\tilde{N} (k)/(1-\varepsilon)}{\norm{a_i}_p^p}$ given that $t_i \leq \frac{\tilde{N} (k)/(1-\varepsilon)}{\norm{a_i}_p^p}$ is $\varepsilon$.
    Thus the expected number of indices with $X_i'=1$ is at least $(1-\varepsilon) $ times the number of indices with $X_i=1$ and consequently $ N((1-\varepsilon)k) \geq \tilde{N} (k)$.

    Now let $X_i'=1$ if $\norm{a_i}_p^p/t_i \geq \tilde{N} (k)/(1+\varepsilon)$.
    Note that $P(X_i'=1)\leq P(X_i=1) \cdot (1+\varepsilon)$ and that the probability that $t_i \in (1/(1+\varepsilon), 1) \cdot \frac{\tilde{N} (k)/(1-\varepsilon)}{\norm{a_i}_p^p}$ is $1-\frac{1}{1+\varepsilon}=\frac{\varepsilon}{1+\varepsilon}\leq \varepsilon$.
    
    Thus the expected number of indices with $X_i'=1$ is at most $(1+\varepsilon)$ times the number of indices with $X_i=1$ and consequently $ N((1+\varepsilon)k) \leq \tilde{N} (k)$.

\end{proof}

We are now ready to prove the first three statements of \cref{thm_main:alg2}/\ref{thm: alg2} along with some more technical claims.

\begin{corollary}\label{cor: alg2}
    If $\varepsilon \leq 1/20$, $r \geq \max\{ 32 \ln(n)k \cdot (12/\varepsilon)^p, 120k \}, s \geq 3\ln(6n/\delta)/0.025^3$ and $k \geq 160\ln(2/\delta)$ then with failure probability at most $5\delta$ it holds that
    \begin{itemize}
        \item[1)] $L$ contains all indices $i$ with $\norm{\at_i} \geq \tilde{N}(2k) $;
        \item[2)] $ \frac{\norm{A}_p^p}{k}\geq \tilde{N}((10/8)k)\geq \alpha \geq \tilde{N}((14/8)k)$;
        \item[3)] $\norm{\at_i-a_i}_p \leq (\varepsilon/3) \norm{a_i}_p  $ holds for all elements in $S$
        \item[4)] $|S|\in [k, 2k]$;
        \item[5)] $P(i \in S) \in [(1-\varepsilon) \cdot \frac{\norm{a_i}_p^p}{\alpha}, (1+\varepsilon) \cdot \frac{\norm{a_i}_p^p}{\alpha}]$ if $ (1-\varepsilon) \cdot \frac{\norm{a_i}_p^p}{\alpha} \leq 1$ and $P(i \in S)=1$ otherwise.
        \item[6)] $P(i \in S)\geq \min\{ 1, \frac{k \norm{a_i}_p^p}{\norm{A}_p^p} \}$
    \end{itemize}
\end{corollary}

\begin{proof}
    The first part of this corollary is to prove that $L$ contains all the important elements.

    By \cref{lem: newN(k)lemma} we have that
    \[
        N(2k) \geq M(A, 2k)/(2k).
    \]
    By \cref{lem: A'Mound} it holds that $M(A', 6k) \leq 2\ln(n)M(A, 2k)$ with failure probability at most $2\delta$.
    Applying Theorem \ref{thm: findhh} to $A'$ with $r = \max\{ 32 \ln(n)k \cdot (12/\varepsilon)^p, 120k \}, s \geq 3\ln(n \delta^{-1}/6)/0.025^3$ we get that with failure probability at most $\delta$ all indices $i$ with $\norm{a_i'}_p^p \geq M(A, 2k)/(2k)$ are in $L$ and $\norm{\at_i-a_i}_p \leq (\varepsilon/3) \norm{a_i}_p  $ holds for all elements in $L$ and thus in particular for any element in $S \subseteq L$ proving $1)$ and $3)$.

    Next we look at the number of elements in $S$.
    First note that it holds that
    \[
        ||\{ i\in L \mid \norm{\at_i'}_p^p \geq \tilde{N}(k')\}|-k'|\leq k'/8
    \]
    with failure probability at mos $\delta$. The proof of this is exactly as the proof of \cref{lem: numberofN(k)}, just replacing $N$ by $\tilde{N}$.
    We apply this twice, for $k'=(14/8)k$ to see that $ \alpha \geq \tilde{N}((14/8)k)$ with failure probability at most $\delta$ and for $k'=(10/8)k$ to see that $\alpha \leq \tilde{N}((10/8)k)$ with failure probability at most $\delta$.
    Combining both results we get that $\alpha=N(k_\alpha)$ with $k_\alpha \in [(10/8)k, (14/8)k] $

    As we apply our algorithm the second time with fixed $\alpha$, we apply the same argument to prove that
    \[
        ||\{ i\in L \mid \norm{\at_i'}_p^p \geq \tilde{N}(k_\alpha)\}|-k_\alpha|\leq k_\alpha/8
    \]
    implying that $|S| \in [k, 2k] $.
    Further by \cref{lem: NtildeN} and \cref{lem: newN(k)lemma}, and using that $\eps \leq 1/20$, we have that
    \[
        \alpha = \tilde{N}(k_\alpha) \leq N((1-\varepsilon)(10/8)k)
        \leq N((9/8)k)\leq \frac{\norm{A}_p^p}{(9/8)k}
    \]

    Finally, we consider the sampling probabilities.
    We note that $i$ is sampled if $i \in L$ and $\norm{\at_i'}_p^p \geq \alpha$.
    Since $i \in L$, we have that $\norm{\at_i'}_p^p=(1 \pm \varepsilon)\norm{a_i'}_p^p$.
    Thus $i$ is sampled if $\norm{\at_i'}_p^p \geq \frac{\alpha}{1-\varepsilon} $ and $i$ is not in $S$ if $\norm{\at_i'}_p^p \leq \frac{\alpha}{1+\varepsilon}$.
    Thus the probability $P(i \in S)$ is at least $\frac{(1-\varepsilon)\norm{a_i}_p^p}{\alpha}$ and at most $\frac{(1+\varepsilon)\norm{a_i}_p^p}{\alpha}$ proving $5)$.
    For the $6)$ observe that by our previous arguments, \cref{lem: newN(k)lemma}, and again using $\eps \leq 1/20$, we have that
    \[
        P(i \in S) \geq  \frac{(1-\varepsilon)\norm{a_i}_p^p}{\alpha} \geq \frac{(1-\varepsilon)\norm{a_i}_p^p}{N((9/8)k)} \geq \frac{(1-\varepsilon)(9/8)k\norm{a_i}_p^p}{\norm{A}_p^p}
        \geq \frac{k \norm{a_i}_p^p}{\norm{A}_p^p}.
    \]
\end{proof}

The following Lemma completes the proof of \cref{thm_main:alg2}/\ref{thm: alg2}:

\begin{lemma}
    Assume that the statements of Corollary \ref{cor: alg2} hold.
    For all elements $(i, \at_i, w_i)$ it holds that $w_i=(1 \pm \varepsilon)P(i \in S)^{-1} $.
    Further it holds that $\sum_{i \in S} w_i \norm{\at_i}_p^p= (1 \pm \varepsilon)\norm{A}_p^p$ with failure probability at most $\delta$.
\end{lemma}

\begin{proof}
    Assuming that for any element $i\in L $ it holds $\norm{\at_i'}_p^p = (1\pm \varepsilon) \norm{a_i'}_p^p$ we have
\begin{align*}
    P(i \in S)&=P( \norm{\at_i'}_p^p \geq \alpha)\geq (1/2) P(\norm{a_i'}_p^p \geq \alpha ) + (1/2)P(\norm{a_i'}_p^p \geq \alpha/(1-\varepsilon) )\\
    &=(1/2)\cdot \frac{\norm{a_i}_p^p}{\alpha}+ (1/2)\cdot \frac{(1-\varepsilon)\norm{a_i}_p^p}{\alpha}\\
    &=(1-\varepsilon/2)\frac{\norm{a_i}_p^p}{\alpha}.
\end{align*}
Here the first inequality uses the fact that with probability $1/2$ we have that $\norm{\at_i'}_p^p \geq \norm{a_i'}_p^p $, since the vector added to $a_i$ in its respective bucket has a $0.5$ chance to point in the same direction as $a_i$.

Similarly, we have that
\begin{align*}
    P(i \in S)&=P( \norm{\at_i'}_p^p \geq \alpha)\leq (1/2) P(\norm{a_i'}_p^p \geq \alpha ) + (1/2)P(\norm{a_i'}_p^p \geq \alpha/(1+\varepsilon) )\\
    &=(1/2)\cdot \frac{\norm{a_i}_p^p}{\alpha}+ (1/2)\cdot \frac{(1+\varepsilon)\norm{a_i}_p^p}{\alpha}\\
    &=(1+\varepsilon/2)\frac{\norm{a_i}_p^p}{\alpha}.
\end{align*}
Since $(1 \pm \varepsilon/2)/(1 \pm \varepsilon/2) =  (1\pm 3\varepsilon)$ this proves that
\[
    w_i=\frac{\alpha}{\norm{\at_i}_p^p}=\frac{\alpha}{(1 \pm \varepsilon)\norm{a_i}_p^p}
    =(1 \pm 2 \varepsilon)\frac{\alpha}{\norm{a_i}_p^p}= \frac{1 \pm 2 \varepsilon}{(1 \pm \varepsilon/2)P(i \in S)}
    =(1\pm 3\varepsilon)P(i \in S)^{-1}.
\]

Now consider the random variable that takes the value $X_i=\frac{\norm{a_i}_p^p}{\norm{A}_p^p} \cdot P(i \in S)^{-1}$ with probability $ P(i \in S)$ and $X_i = 0$ otherwise.
Assume without loss of generality that $ (1-\varepsilon) \cdot \frac{\norm{a_i}_p^p}{\alpha} \leq 1$ holds for all $i \in [n]$. Indices with $ (1-\varepsilon) \cdot \frac{\norm{a_i}_p^p}{\alpha} > 1$ we have that $P(i \in S)=1$ and would only add a special case where the variance of $X_i$ is zero. Then by Corollary \ref{cor: alg2} item 5)  we have that
\[
    \frac{\norm{a_i}_p^p}{\norm{A}_p^p} \cdot P(i \in S)^{-1}
    \leq \frac{\norm{a_i}_p^p}{\norm{A}_p^p} \frac{\alpha}{(1- \varepsilon)\norm{\at_i}_p^p}
    \leq  \frac{\norm{a_i}_p^p}{\norm{A}_p^p} \frac{\alpha}{(1- 3\varepsilon)\norm{a_i}_p^p} 
    =  \frac{\alpha}{(1- 3\varepsilon)\norm{A}_p^p} \leq \frac{2\alpha}{\norm{A}_p^p}
    \leq 2/k.
\]
Further we have that $\E(\sum_{i=1}^n P(i \in S)X_i )=1$ and
\[ \sum_{i=1}^n P(i \in S)X_i^2 \leq \frac{2}{k} \cdot \sum_{i=1}^n P(i \in S)X_i =\frac{2}{k} \]
Using Bernstein's inequality we get that
\begin{align*}
    P(|\sum_{i=1}^n P(i \in S)X_i ~ -1 |\geq \varepsilon)
    \leq  \exp\left( -\frac{\varepsilon^2/2}{2/k + 2/(3k)  } \right)
    \leq \exp\left( -\frac{k \varepsilon^2}{6  } \right) \leq \delta.
\end{align*}
Since we do not know $P(i \in S )^{-1}$ but rather $w_i = (1 \pm 3\varepsilon)P(i \in S )^{-1}$ we get that
\begin{align*}
    &\sum_{i \in S} w_i \norm{\at_i}_p^p
    =\sum_{i \in S} (1\pm 3\varepsilon)P(i \in S)^{-1} (1\pm \varepsilon) \norm{a_i}_p^p\\
    &=\sum_{i \in [n]} (1\pm 3\varepsilon)(1\pm \varepsilon) X_i \norm{A}_p^p
    = (1 \pm \varepsilon)\norm{A}_p^p (1\pm 3\varepsilon)(1\pm \varepsilon)
    = (1 \pm 6\varepsilon)\norm{A}_p^p
\end{align*}
with failure probability at most $\delta$.
\end{proof}

\cref{thm_main:alg2}/\ref{thm: alg2} follows by substituting $\varepsilon$ by $\varepsilon/6$ and $\delta$ by $\delta/6$.

\section{Weighted sampling from multiple distributions}\label{sec: multsampldist}

Assume that we want to sample an index $i$ with probability $p_i+p_i'$ but we only have access to a sampling algorithm that samples with probability $p_i$ and another sampling algorithm that samples with probability $p_i'$.
The question is whether this is sufficient to sample with probability roughly $p_i + p_i'$ for some constants
$c_1 p_i +c_2 p_i'$.

\begin{lemma}\label{lem:pqsampling}
    Let $S_1 \subseteq [n]$ (resp $S_2$) be a sample where index $i\in[n]$ is sampled with probability $p_i$ (resp $p_i'$).
    Then $S=S_1 \cup S_2$ is a sample where $i$ is sampled with probability $(p_i + p_i') \geq P(i \in S) \geq (1/2) (p_i + p_i')$.
    Further, if both $p_i$ and $p_i'$ are known up to a factor of $(1\pm\varepsilon)$, i.e., we have $\tilde{p}_i=(1\pm\varepsilon) p_i $ and $\tilde{p}_i'=(1\pm\varepsilon) p_i' $, then we can compute the probability $ P(i \in S)$ up to a factor of $(1\pm\varepsilon)$.
\end{lemma}

\begin{proof}
    First note that the probability that $i \notin S$ is given by
    \[
        P(i \notin S)=(1-p_i)(1-p_i')=1-p_i-p_i'+p_ip_i'
    \]
    and consequently
    \[
        P(i \in S)=p_i+p_i'-p_ip_i'.
    \]
    Since $0\leq p_ip_i' = \frac{p_ip_i'}{2}+ \frac{p_ip_i'}{2}\leq \frac{p_i}{2}+ \frac{p_i'}{2}$ this implies that
    \[
        p_i + p_i' \geq P(i \in S)\geq \frac{1}{2}\cdot  (p_i + p_i').
    \]
    Further let $\tilde{p}_i=c_1  {p}_i$ and $\tilde{p}_i'=c_2  {p}_i'$.
    Using elementary calculus and using the fact that $\tilde{p}_i'\geq 0 $ and $\tilde{p}_i \geq 0$ one can verify that the probabilities are maximized, respectively minimized at the approximation boundaries, i.e., when $c_1,c_2 = (1\pm\eps)$.
    
    We thus get that
    \begin{align*}
        c_1 p_i + c_2 p_i'- c_1 c_2 p_ip_i' \leq (1 + \varepsilon)(p_i + p_i')- (1 + \varepsilon)^2p_ip_i' \leq (1 + \varepsilon)(p_i + p_i'- p_ip_i')=(1 + \varepsilon)P(i \in S).
    \end{align*}
    and similarly
    \begin{align*}
        c_1 p_i + c_2 p_i'- c_1 c_2 p_ip_i' \geq (1 - \varepsilon)(p_i + p_i')- (1 - \varepsilon)^2p_ip_i' \geq (1 - \varepsilon)(p_i + p_i'- p_ip_i')=(1 - \varepsilon)P(i \in S).
    \end{align*}
\end{proof}

We get the following corollary:

\begin{corollary}[copy of \cref{cor:combinedsamplingmain}]\label{cor: combinedsampling}
    Combining a sample $S_1$ from \cref{alg:turnstilesampling} with parameter $k$ and a uniform sample $S_2$ with sampling probability $k/n$ we get a sample $S_1 \cup S_2$ of size $\Theta(k)$ and the sampling probability of $i$ is $\Omega\left(k \left(\frac{\norm{a_i}_p^p}{\norm{A}_p^p} +1/n\right)\right)$, for any sample $\at_i$ we have that $\norm{\at_i-a_i}_p\leq (\varepsilon/3) \norm{a_i}_p $.
    Further, the sampling probability and thus appropriate weights can be approximated up to a factor of $(1\pm\eps)$.
\end{corollary}

For the sake of completeness note that if we want to sample with probability $\Omega\left(k \left(\frac{\norm{a_i}_p^p}{\norm{A}_p^p} +1/n\right)\right) $ then for this particular sampling probability there is another even simpler approach, which is to not sketch indices with $t_i \geq k/n$ in \cref{alg:turnstilesampling}, but instead include the original rows $a_i$ into a separate uniform sample. In this case, their weights $w_i$ need to be adapted to $w_i=p_i^{-1}=(\max\{\frac{k}{n}, \frac{\norm{a_i}_p^p}{\alpha}\})^{-1}$.

\section{Application to \texorpdfstring{$\ell_p$}{lp} leverage score sampling for regression loss functions}\label{app:application}

We now show how \cref{alg:turnstilesampling} can be used to get an $\eps$-coreset by simulating known results based on $\ell_p$ leverage score sampling. We first need a few more definitions.

\begin{definition}[$\ell_p$ leverage scores]\label{def:leveragescores}
     For fixed $p \in [1, 2]$ we set $u_i^{(p)}= \sup_{z \neq 0} \frac{|a_iz |^p}{\|Az \|^p_p}$ to be the $i$-th leverage score of $A$.
\end{definition}

\begin{definition}[\citealt{DasguptaDHKM09}, copy of \cref{def:good_basismain}]
\label{def:good_basis}
Let $A$ be an $n\times d$ matrix, let $p \in [1,\infty)$, and let $q \in (1, \infty]$ be its dual norm, satisfying $\frac{1}{p}+\frac{1}{q}=1$.
Then an $n \times d$ matrix $V$ is an \emph{$(\alpha,\beta,p)$-well-conditioned basis} for the column space of $A$ if\\
(1) $\Vert V \Vert_p:=\left( \sum_{i \leq n, j \leq d}|V_{ij}|^p\right)^{1/p}\leq \alpha$, and\\
(2) for all $z\in\mathbb{R}^d$, $\Vert z \Vert_q \leq \beta \Vert V z\Vert_p$.

We say that $V$ is an \emph{$\ell_p$-well-conditioned basis} for the column
space of $A$ if $\alpha$ and $\beta$ are $d^{O(1)}$,
independent of $n$.
\end{definition}

\begin{proposition}[copy of \cref{pro:Rmatrixmain}]\label{pro: Rmatrix}
    There exists a turnstile sketching algorithm that for a given $p\in[1,2]$ computes an invertible matrix $R$ such that $AR^{-1}$ is $(\alpha,\beta,p)$-well-conditioned with $\alpha=O(d^{2/p-1/2}(\log d)^{1/p-1/2}),\beta=O((d(\log d)(\log\log d))^{1/p})$, and $(\alpha\beta)^p = O(d^{3-p/2}(\log d)^{2-p/2}(\log\log d))$ for $p\in[1,2)$. For $p=2$ it holds that $\alpha=O(\sqrt{2d}),\beta=O(\sqrt{2})$, and $(\alpha\beta)^p=O(d)$. Moreover, the $\ell_p$ leverage scores $u_i^{(p)}$ satisfy $u_i^{(p)}\leq \beta^p\norm{a_i R^{-1}}_p^p$, and $\sum_i u_i^{(p)}\leq (\alpha\beta)^p=d^{O(1)}$.
\end{proposition}

\begin{proof}[Proof of Proposition \ref{pro:Rmatrixmain}/\ref{pro: Rmatrix}]
    Let $\Pi\in \mathbb{R}^{r\times n}$ be an $\ell_p$ subspace embedding satisfying \begin{align} \label{eqn:subspaceembedding}
    \forall x\in \mathbb{R}^d\colon \|A x\|_p / \eta \leq \|\Pi A x\|_p \leq \gamma \|A x\|_p
    \end{align}

    We show that if $\Pi A = QR$ is the $QR$ decomposition, then $U=AR^{-1}$ is a $(\eta d r^{1/2},\gamma,p)$-well-conditioned basis for the column space of $A$. Note that $q\geq 2 \geq p \geq 1$. Then
    \begin{align*}
        \|z\|_q \leq \|z\|_2 = \|Qz\|_2 = \|\Pi A R^{-1} z\|_2 \leq \|\Pi A R^{-1} z\|_p \leq \gamma \|A R^{-1} z\|_p = \gamma \|U z\|_p
    \end{align*}
    and noting that $Q\in \mathbb{R}^{r\times d}$ has orthonormal columns, we also have that
    \begin{align*}
        \|U\|_p^p &= \sum_{i=1}^d \|AR^{-1}_i\|_p^p \leq \eta^p \sum_{i=1}^d \|\Pi AR^{-1}_i\|_p^p = \eta^p \sum_{i=1}^d \|Q_i\|_p^p \\
        &\leq \eta^p {d}^{1/2} \left( \sum_{i=1}^d \|Q_i\|_p^{2p} \right)^{1/2} \leq \eta^p {d}^{1/2} \left( \sum_{i=1}^d ({r}^{1/p-1/2})^{2p} \|Q_i\|_2^{2p} \right)^{1/2} \\
        &\leq \eta^p {d}^{1/2} ({r}^{1/p-1/2})^{p} \left( \sum_{i=1}^d \|Q_i\|_2^{2p} \right)^{1/2} = \eta^p d ({r}^{1/p-1/2})^{p} \\
    \end{align*}
    Taking the $p$-th root on both sides yields $\|U\|_p \leq \eta d^{1/p} {r}^{1/p-1/2}.$
    
    Next, we choose for $\Pi$ the oblivious subspace embeddings given in \citealp[Corollary 1.12 of][]{WoodruffY23lpSE}, that allow for the following parameterization: if $1 \leq p < 2$ then \cref{eqn:subspaceembedding} holds with $\eta = O(1), \gamma = O((d(\log d)(\log\log d))^{1/p}),$ and $r=O(d\log d)$. It is thus $(\alpha,\beta,p)$-well-conditioned with $\alpha = \eta d^{1/p} r^{1/p-1/2} = O(d^{2/p-1/2}(\log d)^{1/p-1/2})$, and $\beta = \gamma = O((d(\log d)(\log\log d))^{1/p})$. Thus, $(\alpha\beta)^p = O(d^{3-p/2}(\log d)^{2-p/2}(\log\log d))$.
    
    In the special case $p=2$, it is known \citep{ClarksonW17} that the CountSketch directly yields an $(1\pm\eps)$-error oblivious subspace embedding with sparsity $s=1$, thus it can be applied in $O(\texttt{nnz}(A))$ time, and was shown in \citealp[Lemma 2.14 of][]{MunteanuOP22} that it yields a $(\alpha,\beta,2)$-well-conditioned basis with $\alpha=\sqrt{2d}, \beta=\sqrt{2}$ using the $QR$ decomposition as above. Thus, $(\alpha\beta)^p = 4d$ in this case.

    Finally, \citealp[Lemma 2.12 of][]{MunteanuOP22} yields that $u_i^{(p)}\leq \beta^p \|U_i\|_p^p = \beta^p \|a_iR^{-1}\|_p^p$, and $\sum_i u_i^{(p)}\leq (\alpha\beta)^p$.
\end{proof}

We remark that there exist sparse alternatives for $\ell_p$ subspace embeddings given in \citealp[Theorems  4.2, 5.2 of][]{WangW22} that admit a sparsity of $s=O(\log d)$. These apply to the data in $O(\texttt{nnz}(A)\log d)$ time (much faster than dense matrix multiplication) where $\texttt{nnz}(A)$ denotes the number of non-zero entries of $A$. However this comes at the cost of slightly larger $(\alpha\beta)^p = O(d^{2+p/2}(\log d)^{1+p/2})$.

For asymmetric loss functions (all of \cref{pro:leveragescoresampling} except $g(t)=|t|^p$), we require an additional parameter $\mu$ that has been introduced for logistic regression by \citet{MunteanuSSW18} and generalized to arbitrary $p$ \citep{MunteanuOP22}.

\begin{definition}[$\mu$-complexity, \citealt{MunteanuOP22}]\label{def:mu_complex}
	Let $A \in \mathbb{R}^{n \times d }$ be any matrix. For a fixed $p\geq 1$ we define
	\[ \mu_p(A)=\sup_{z  \in \mathbb{R}^d\setminus\{0\}} \frac{\sum_{a_i z  >0}|a_i z |^p}{\sum_{a_i z  <0}|a_i z |^p}. \]
	We say that $A$ is $\mu$-complex if $\mu_p(A)\leq \mu < \infty$.
\end{definition}

We summarize a (non-exclusive) list of leverage score sampling results for various loss functions in the following proposition:

\begin{proposition}\label{pro:leveragescoresampling}
    Let $A\in\mathbb{R}^{n\times d}$ be $\mu$-complex.
    If we sample $S\subset [n]$ of a certain size $k:=|S|=\operatorname{poly}(\mu d/\eps)$ proportional to sampling probabilities $p_i \geq c(\norm{a_i R^{-1}}_p^p + 1/n)$ where $R$ is the matrix from Proposition \ref{pro: Rmatrix} and weights $w_i = (kp_i)^{-1}$ then with constant probability the weighted subsample is an $\eps$-coreset, i.e., it holds that
    \[
        \forall z \in \mathbb{R}^d\colon \sum_{i \in S}w_i g(a_i z ) = (1 \pm \varepsilon) \sum_{i \in [n]}g(a_i z)
    \]
    where $g(\cdot)$ denotes one of the following loss functions:
    \begin{itemize}
        \item $g(t)=|t|^p$ (here $k=\operatorname{poly}(d/\eps)$ is independent of $\mu$),
        \item $g(t)=\max\{0, t\}^p$,
        \item $g(t)=-\ln(\Phi_p(-t))$, where $\Phi_p\colon \mathbb{R} \rightarrow [0,1]$ denotes the CDF of the $p$-generalized normal distribution,
        \item $g(t)=\ln(1+e^t)$.
    \end{itemize}
\end{proposition}

\begin{proof}
    For the first item, $g(t)=|t|^p$, which is known as the loss function for linear $\ell_p$ regression, the result is known for $p=2$ \citep{DrineasMM2006a}, and has been generalized to general $p\in [1,2]$ \citep{DasguptaDHKM09}, and improved using sketching techniques \citep{SohlerW11,DrineasMMW12,WoodruffZ13}.
    
    For the second item, we refer to \citep{MunteanuOP22} who solved the problem for $g(t)=\max\{0, t\}^p$ as a means to approximate the third item, i.e., the $p$-generalized probit regression problem.
    
    The fourth item $g(t)=\ln(1+e^t)$ is known as logistic regression \citep{MunteanuSSW18,MaiMR21}, that can be handled by means of $\ell_1$ leverage score sampling \citep{MunteanuOP22}.
\end{proof}

Using these results we show that we can construct an $\eps$-coreset in the turnstile stream setting using our algorithm with only $\operatorname{poly}(\mu d/\eps)\log n$ overhead. The main challenge here is to show that the perturbation incurred from the fact that $\at_i$ is not exactly $a_i$, does not cause a large error for the loss function.

\begin{theorem}[copy of \cref{lem:applicationsmain}]\label{lem:applications}
    Let $A\in\mathbb{R}^{n\times d}$ be $\mu$-complex (see \cref{def:mu_complex}). 
    Given a leverage score sampling algorithm that constructs an $\eps$-coreset of size $k$, as for the loss functions below (summarized in \cref{pro:leveragescoresampling}), there exists a sampling algorithm that works in the turnstile stream setting that with constant probability outputs a weighted $2\varepsilon$-coreset $(A', w)\in \mathbb{R}^{k' \times d} \times \mathbb{R}_{\geq 1}$ of size $k'=Theta(k)$, such that
    \[
        \forall z  \in \mathbb{R}^d\colon \left| \sum_{i \in [k']}w_i g(a_i' z ) - \sum_{i=1}^n g(a_i z ) \right|\leq 2\varepsilon \sum_{i=1}^n g(a_i z ).
    \]
    The size of the sketching data structure used to generate the sample is $r \cdot s$, where $s=3\ln(36 n/\delta)$ and
    \[ 
        r=
        \begin{cases}
            O\left(k \ln(n) (\alpha^p \beta^p/\varepsilon)^p\right) & \text{if $g(t)=|t|^p $},\\
            O\left(k \ln(n) (\mu \alpha^p \beta^p/\varepsilon)^p\right) &\text{if $g(t)=\max\{0, t\}^p,$}\\
            O\left(k \ln(n) (\mu \alpha \beta/\varepsilon)\right) & \text{if $g(t)=\ln(1+e^t)$,}\\
            O\left(k \ln(n) (p \mu^2 \alpha^p \beta^p/\varepsilon)^p\right) &\text{if $g(t)=-\ln(\Phi_p(-t))$,}
        \end{cases}    
    \]
    where $\Phi_p\colon \mathbb{R} \rightarrow [0,1]$ denotes the CDF of the $p$-generalized normal distribution.
    In particular if the matrix $P:=R^{-1}$ of \cref{pro:Rmatrixmain} is used in \cref{alg:turnstilesampling}, then the overhead is at most $O(\ln(n) (\mu^2 \alpha^p \beta^p/\varepsilon)^p)=\operatorname{poly}(\mu d/\eps)\log(n)$.
\end{theorem}

\begin{proof}[Proof of Theorem \ref{lem:applicationsmain}/\ref{lem:applications}]
    We use the algorithm from Proposition \ref{pro: Rmatrix} and \cref{alg:turnstilesampling} in parallel.
    From the algorithm of Proposition \ref{pro: Rmatrix} we get a matrix $R$ such that $u_i^{(p)}\leq c_R\norm{a_i R^{-1}}_p^p$.
    Using \cref{alg:turnstilesampling} with the modification described in Section \ref{sec: multsampldist} and parameters $r \geq \max\{ 32 k\ln(n) \cdot (72/\varepsilon')^p, 120k \}$, $ s \geq 3\ln(36n/\delta)/0.025^3$, and $\varepsilon'=\varepsilon/(\alpha \beta)^p$, we get a sample $S$ of size $2k \geq |S| \geq k$ by \cref{thm_main:alg2}/\ref{thm: alg2} resp. Corollary \ref{cor: combinedsampling}.
    Thus $S$ consists of $\Theta(k)$ (weighted) samples $ (i, \at_i, \tilde w_i)$, where $\norm{\at_iR-a_i}_p \leq (\varepsilon'/3) \norm{a_i}_p$ and $\tilde w_i =(1 \pm \varepsilon') w_i = (1 \pm \varepsilon') P(i \in S)^{-1}$ with $P(i \in S) \geq c(u_i^{(p)}+ 1/n)$.

    Using \cref{pro:leveragescoresampling}, with constant probability it holds that
    \[
        \sum_{i \in S}w_i g(a_i z ) = (1 \pm \varepsilon')^2 \sum_{i \in [n]}g(a_i z ).
    \]
    Here, the additional factor of $(1 \pm \varepsilon')$ comes from the approximation of the weights in the output of our algorithm, up to which we can assume in the following we have the exact weights of \cref{pro:leveragescoresampling}.
    The remaining part of the proof is to show that the error incurred by replacing $a_i$ with the output rows $\at_i P^{-1} = \at_i R$ is small.

    $\bullet$ First we consider $g(t)=|t|^p$.
    Recall that $AR^{-1}$ is an $(\alpha,\beta,p)$-well-conditioned basis.
    We aim to use a variant of Bernoulli's inequality in the following form, which follows using the mean value theorem:
    $(|a| + |b|)^p -|a|^p \leq p|b|(|a| + |b|)^{p-1}$.
    We also use that $\norm{\at_i - a_iR^{-1}}_p \leq (\varepsilon'/3) \norm{a_i R^{-1}}_p$. For $\eps'=\eps/(\alpha\beta)^p$ this yields
    \begin{align*}
        \left| \sum_{i \in S}w_i g(\at_i R z) - \sum_{i \in S}w_i g(a_i z) \right|
        &=\left| \sum_{i \in S}w_i |\langle \at_i R, z \rangle|^p - |\langle a_i, z \rangle|^p \right|\\
        & \leq   \sum_{i \in S}w_i \left| | \langle \at_i R, z \rangle|^p - |\langle a_i, z \rangle|^p \right|\\
        & =   \sum_{i \in S}w_i \left| | \langle \at_i , R z \rangle|^p - |\langle a_i R^{-1}, R z \rangle|^p \right|\\
        &\leq \sum_{i \in S}w_i \left| | \langle a_i R^{-1} +\at_i - a_i R^{-1} , R z \rangle|^p - |\langle a_i R^{-1}, R z \rangle|^p \right|\\
        &\leq \sum_{i \in S}w_i \left| \left( | \langle a_iR , R^{-1} z \rangle |  + | \langle \at_i-a_iR^{-1} , R z  \rangle| \right)^p - |\langle a_i R^{-1}, Rz \rangle|^p \right|\\
        &\leq \sum_{i \in S}w_i p \left| \langle \at_i-a_iR^{-1} , R z  \rangle\right| \left( | \langle a_iR^{-1} , R z \rangle |  + | \langle \at_i-a_iR^{-1} , R z  \rangle| \right)^{p-1} \\ 
        &\leq \sum_{i \in S}w_i p  \norm{\at_i-a_iR^{-1}}_p\norm{R z}_q \left( \norm{a_iR}_p \norm{R^{-1} z}_q  + \norm{\at_i-a_iR^{-1}}_p\norm{R z}_q \right)^{p-1} \\
        &\leq \sum_{i \in S}w_i p (\eps'/3) \norm{a_iR^{-1}}_p\norm{R z}_q \left( (1+\eps'/3)\norm{a_iR}_p \norm{R^{-1} z}_q  \right)^{p-1} \\
        &\leq \sum_{i \in S}w_i p (2\eps'/3) \norm{a_iR^{-1}}_p^p\norm{R z}_q^p\\
        &\leq \sum_{i \in S}w_i \norm{a_iR^{-1}}_p^p (4\eps'/3) \beta^p \norm{AR^{-1} R z}_p^p\\
        &\leq (1+\eps'/3) \norm{AR^{-1}}_p^p (4\eps'/3) \beta^p \norm{AR^{-1} R z}_p^p\\
        &\leq (1+\eps'/3) (4\eps'/3) (\alpha\beta)^p \norm{Az}_p^p\\
        &\leq 2\eps \norm{Az}_p^p = \sum\limits_{i\in[n]} |a_i z|^p.        
    \end{align*}

    $\bullet$ Now let $ g(t)=\max\{0, t\}^p$. Using $\eps'=\eps/((\mu+1)(\alpha\beta)^p)$, we have very similarly to the case $|t|^p$ above (the $\ldots$ indicate that these steps are verbatim). Consider the cases where $\max\{\at_i R z,a_i z\}\leq 0$, or $\min\{\at_i R z,a_i z\}\geq 0$. In both cases we have that
    \begin{align}\label{eq:powerReLU1}
        \left| \sum_{i \in S}w_i g(\at_i R z) - \sum_{i \in S}w_i g(a_i z) \right|
        &\leq \sum_{i \in S}w_i \left| \max\{0, \at_i R z\}^p - \max\{0, a_i z\}^p \right|\nonumber\\
        & \leq   \sum_{i \in S}w_i \left||\langle \at_i R^{-1}, z \rangle|^p - |\langle a_i, z \rangle|^p \right|\nonumber\\
        & \leq \ldots \text{(verbatim to the previous calculation)}\nonumber\\
        &\leq 2 \varepsilon \norm{A z}_p^p /(\mu+1)\nonumber\\
        &=2\eps \sum\limits_{a_i z > 0} |a_i z|^p = 2\eps \sum\limits_{i\in[n]} \max\{0,a_i z\}^p \leq 2\eps \sum\limits_{i\in [n]} g(a_i z).
    \end{align}
    Consider the remaining case where $\max\{\at_i R z,a_i z\}\geq 0\geq \min\{\at_i R z,a_i z\}$. By Hölder's inequality we have that
    \begin{align}\label{eq:hoelder}
        \left| \langle \at_i, R z \rangle - \langle a_i , z \rangle \right| 
        &\leq\left| \langle \at_i, R z \rangle - \langle a_i R^{-1}, R z \rangle \right| \nonumber\\
        &=\left| \langle \at_i -  a_i R^{-1}, R z \rangle \right| \nonumber \\
        &\leq \norm{\at_i -  a_i R^{-1}}_p \norm{R z}_q \nonumber \\
        &\leq (\eps'/3)\norm{a_i R^{-1}}_p \beta \norm{A z}_p.
    \end{align}
    
    Consequently, we get the same overall bound in this case
    \begin{align}\label{eq:powerReLU2}
        \left| \sum_{i \in S}w_i g(\at_i R z) - \sum_{i \in S}w_i g(a_i z) \right| & \leq \sum_{i \in S}w_i \left| \max\{0, \at_i R z\}^p - \max\{0, a_i z\}^p \right| \nonumber\\
        &\leq \sum_{i \in S}w_i \max\{\at_i R z , a_i z\}^p \nonumber\\
        &\leq \sum_{i \in S}w_i \left| \langle \at_i, R z \rangle - \langle a_i , z \rangle \right|^p \nonumber\\
        &\leq \sum_{i \in S}w_i (\eps'/3)\norm{a_i R^{-1}}_p^p \beta \norm{A z}_p^p\nonumber\\
        &\leq (1+\eps'/3)(\eps'/3) (\alpha\beta)^p \norm{Az}_p^p\nonumber\\
        &\leq 2 \varepsilon \norm{A z}_p^p /(\mu+1)\nonumber\\
        &=2\eps \sum\limits_{a_i z > 0} |a_i z|^p = 2\eps \sum\limits_{i\in[n]} \max\{0,a_iz\}^p\leq 2\eps \sum\limits_{i\in [n]} g(a_i z).
    \end{align}

    $\bullet$ Now let $ g(t)=\ln(1+\exp(t))=\ln(\exp(t)(1+\exp(-t)))=t+g(-t)$. Note that $g(t)\geq \max\{0,t\}$. For the derivative, we have that $0 \leq g'(t) = \frac{\exp(t)}{1+\exp(t)} \leq 1$ for all $t \in \mathbb{R}$.
    Let $p=1$, and $\eps'=\eps/((\mu+1)(\alpha\beta))$. 
    Using \cref{eq:hoelder} again with $p=1$, we get the following overall bound 
    \begin{align}\label{eq:logreg}
        \left| \sum_{i \in S}w_i g(\at_i R z) - \sum_{i \in S}w_i g(a_i z) \right|
        &\leq  \sum_{i \in S}w_i \left|\int_{\at_i R z}^{a_i z}g'(t) \,dt \right| \leq \sum_{i \in S}w_i \left|\int_{\at_i R z}^{a_i z} 1 \,dt \right|\nonumber\\
        &= \sum_{i \in S}w_i \left| \langle \at_i, R z \rangle - \langle a_i , z \rangle \right| \nonumber\\
        &\leq \sum_{i \in S}w_i (\eps'/3)\norm{a_i R^{-1}}_1 \beta \norm{A z}_1\nonumber\\
        &\leq (1+\eps'/3)(\eps'/3) (\alpha\beta) \norm{Az}_1\nonumber\\
        &\leq 2 \varepsilon \norm{A z}_1 /(\mu+1)\nonumber\\
        &=2\eps \sum\limits_{a_i z > 0} |a_i z| = 2\eps \sum\limits_{i\in[n]} \max\{0,a_iz\}\leq 2\eps \sum\limits_{i\in [n]} g(a_i z).
    \end{align}

    $\bullet$ Finally, consider $g(t)=-\ln(\Phi_p(-t))$. For this loss function, we run \cref{alg:turnstilesampling} twice in parallel, once with the given parameter $p$ and once with $p=1$. We combine the samples using \cref{lem:pqsampling}, and add a uniform component using \cref{cor: combinedsampling}.
    
    By \citep[Lemma 2.8]{MunteanuOP22}, we have that $f(Az)=\sum\nolimits_{i\in[n]} g(a_iz) \geq \frac{n}{\mu}$. Further by \citep[Lemma 2.6]{MunteanuOP22} it holds that $g(t)$ is monotonically non-decreasing and convex, and further for any $t\geq 1$ it holds that $t^{p-1} \leq g'(t)\leq t^{p-1}+ \frac{p-1}{t}$. The lower bounds of the cited lemma also imply that $g(t)\geq \max\{0,t\}^p/p$. 
    Note that for $t\leq 1$ convexity yields $0 \leq g'(t) \leq g'(1) \leq 2$, and for $t\geq 1$, we get $0< t^{p-1} \leq g'(t)\leq t^{p-1} + 2$.

    Then, we get for $\eps'=\eps/(6p\mu(\mu+1)(\alpha\beta)^p)$ that
    \begin{align}\label{eq:combinedintegral}
        \left| \sum_{i \in S}w_i g(\at_i R z) - \sum_{i \in S}w_i g(a_i z) \right|
        &\leq  \sum_{i \in S}w_i \left|\int_{\at_i R z}^{a_i z}g'(t) \,dt \right| \nonumber\\
        & \leq \sum_{i \in S}w_i \left( \left|\int_{\at_i R z}^{a_i z} 2 \,dt \right| + \left|\int_{\max\{1,\min\{\at_i R z, a_i z\}\}}^{\max\{1,\max\{\at_i R z, a_i z\}\}} t^{p-1} \,dt \right| \right)
    \end{align}
    Note, that the first integral is the same up to a factor of $2$ as the one we used to handle logistic regression, and $\eps'$ is smaller by a factor of $6p\mu$ now. We thus get verbatim to \cref{eq:logreg} that 
    \[
         \sum_{i \in S}w_i \left|\int_{\at_i R z}^{a_i z} 2 \,dt \right| \leq 4\eps/(6p\mu) \sum\limits_{i\in[n]} \max\{0,a_iz\}.
    \]
    
    Next, note that the second integral satisfies $|\int_b^a t^{p-1}~dt| \leq |a^p-b^p|$, and we see that it can be handled verbatim to the case distinction for the $\ell_p$ ReLU function, i.e., as in \cref{eq:powerReLU1,eq:powerReLU2}. Recall that $\eps'$ is smaller by a factor of $6p\mu$. Thus
    \begin{align*}
        \sum_{i \in S}w_i \left|\int_{\max\{1,\min\{\at_i R z, a_i z\}\}}^{\max\{1,\max\{\at_i R z, a_i z\}\}} t^{p-1} \,dt \right| &\leq \sum_{i \in S}w_i \left| \max\{1,\max\{\at_i R z, a_i z\}\}^p - \max\{1,\min\{\at_i R z, a_i z\}\}^p \right|\\
        &
        \leq 2\eps/(6p\mu) \sum\limits_{i\in[n]} {\max\{0,a_iz\}^p}
    \end{align*}
    To conclude, we note that for all $t\in \mathbb{R}\setminus (0,1)$ we have that $\max\{0,t\}\leq \max\{0,t\}^p$, and for $t\in (0,1)$ it holds that $\max\{0,t\}\leq 1$. Thus $\max\{0,t\}\leq \max\{0,t\}^p + 1$. Consequently, we can resume our calculation of \cref{eq:combinedintegral}
    \begin{align*}
        \eqref{eq:combinedintegral} &\leq 4\eps/(6p\mu) \sum\limits_{i\in [n]} \max\{0,a_i z\} + 2\eps/(6p\mu) \sum\limits_{i\in [n]} \max\{0,a_i z\}^p \\
        & \leq \eps/\mu \sum\limits_{i\in [n]} \frac{\max\{0,a_i z\}^p}{p} + \eps n/\mu\\
        & \leq \eps \sum\limits_{i\in [n]} g(a_i z) + \eps \sum\limits_{i\in [n]} g(a_i z) = 2\eps \sum\limits_{i\in [n]} g(a_i z).
    \end{align*}
\end{proof}

\section{Additional details on experiments and data}\label{sec:app:experiments}
\subsection{Computing environment}
All experiments were run on a workstation with AMD Ryzen Threadripper PRO 5975WX, 32 cores at 3.6GHz, 512GB DDR4-3200.

\subsection{Details on datasets}\label{sec:app:experiments_data}
The datasets were automatically downloaded and preprocessed by the Python code. We give a short description of the data for completeness of presentation. These descriptions were copied from \citet{MunteanuOP22,MunteanuOW23}:
the {Covertype}
  ~data consists of $581,012$ cartographic observations of different forests with $54$ features. The task is to predict the type of trees at each location ($49\%$ positive).
  The {Webspam}
  ~data consists of $350,000$ unigrams with $127$ features from web pages, which have to be classified as spam or normal pages ($61\%$ positive).
 The {Kddcup}
 ~data consists of $494,021$ network connections with $41$ features and the task is to detect network intrusions ($20\%$ positive).

\subsection{Experimental focus}
We demonstrate the performance of our novel turnstile $\ell_p$ sampler. 
Recall, that our algorithm is a hybrid between an oblivious sketch and a leverage score sampling algorithm. It thus makes most sense to compare to pure oblivious sketching as well as to pure off-line leverage score sampling. We refer to \citep{MaiMR21,MunteanuOP22} for comparisons between $\ell_p$ leverage scores and Lewis weights, which are not the focus of this paper.

We implement our new algorithm into the experimental framework of the near-linear oblivious sketch of \citet{MunteanuOW23}, and add the code of \citet{MunteanuOP22} for $\ell_1$ leverage score sampling. Our new and combined code is available at \url{https://github.com/Tim907/turnstile-sampling}.

Our a priori hypothesis from the theoretical knowledge on the three regimes is that the performance should be somewhere in the middle between the performances of the competitors. Ideally, we would want our algorithm to perform as closely as possible to off-line leverage score sampling.

\subsection{Details on space requirements and running times}
The required space is $r\cdot s \cdot d$ to store the $r \cdot s$ many $d$-dimensional vectors, where the values of $r$ and $s$ are as stated in all theorems. In bit complexity, we need to add another $\log(n)$  factor under the standard assumption that all values considered in the data stream are polynomially bounded in $n$ and $d$, and $n>d$. Oblivious sketching uses exactly $k$ rows of $d$-dimensional vectors. Leverage score sampling uses $\Theta(n)$ space, since we compute all $n$ leverage scores in main memory, before sampling. In our implementation, the values of $r$ and $s$ were initially evaluated and fixed to $r=\left\lceil k\cdot \max\{30, \log(n) \}\right\rceil$, and $s= 2\cdot \left\lceil \max\{5,\log(n)/2\}\right\rceil$

For turnstile sketching, the running time is $O(nnz(A)\log n)$ where $nnz(A)$ denotes the number of non-zero entries in the representation of $A$. Oblivious sketching requires $O(nnz(A)\log d)$. Offline leverage scores require $O(nd^2)$. However, our turnstile sampler requires an additional \emph{extraction} which dominates the running time requiring $O(nds +ks^2+kd^2)$. The main goal is to get turnstile updates, $(1+\epsilon)$ error, and poly$(d, \epsilon, \log n)$ space, which the comparison methods cannot provide. However, this comes at the cost of increased running time. Clearly, the oblivious sketch cannot be outperformed but it has limitations in terms of accuracy. In our experimants, the sketching and extraction time of the turnstile sampler is larger than the other methods by a factor of 8-15. However the total running time including optimization is usually increased by only a factor 3-6.

\subsection{Experiments for logistic regression}
\begin{figure*}[ht!]
\begin{center}
\textbf{\textsc{Logistic Regression}}
\vskip 0.2in
\begin{tabular}{cccc}
{~}&
{\small\hspace{.5cm}\textsc{CoverType}}&{\small\hspace{.5cm}\textsc{WebSpam}}&
{\small\hspace{.5cm}\textsc{KddCup}}\\
\rotatebox{90}{\hspace{6pt}\textsc{Approx. Ratio}}&
\includegraphics[width=0.2951\linewidth]{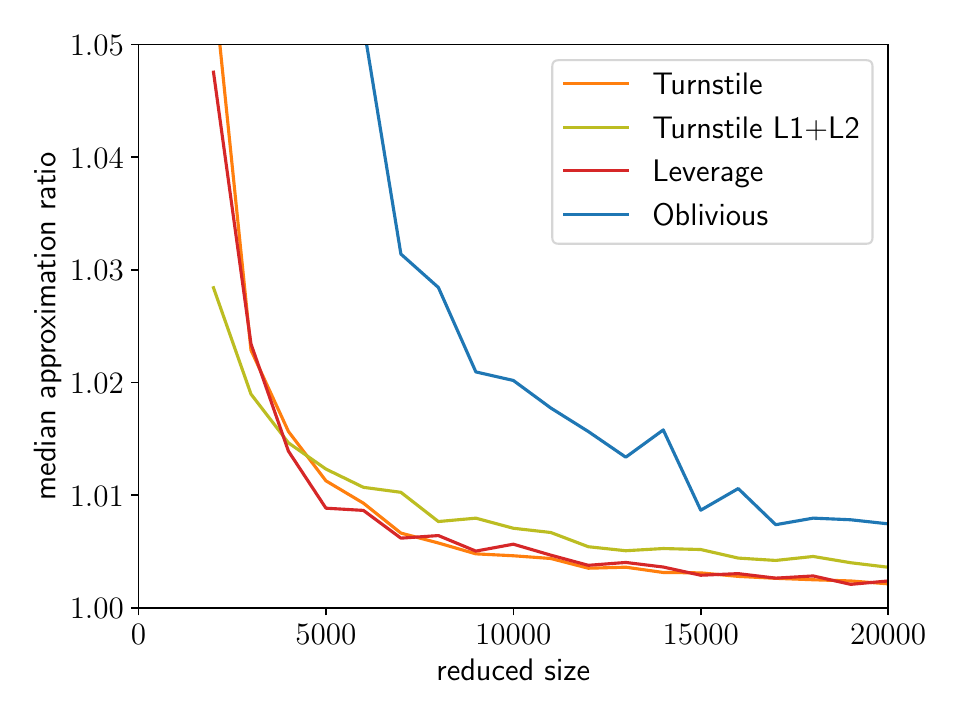}&
\includegraphics[width=0.2951\linewidth]{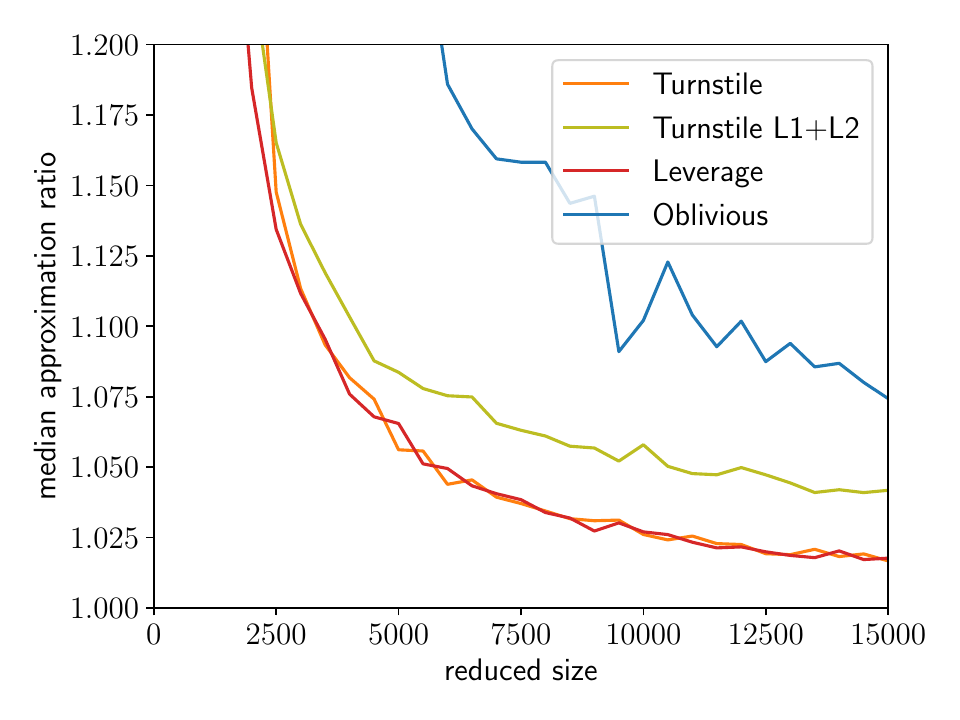}&
\includegraphics[width=0.2951\linewidth]{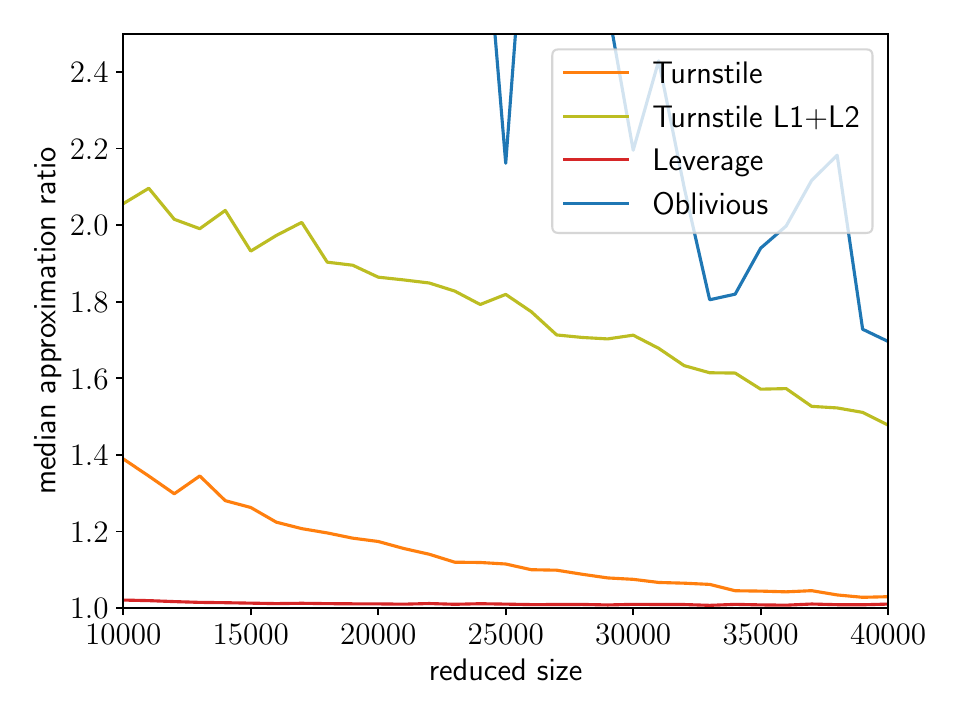}\\
\rotatebox{90}{\hspace{20pt}\textsc{Sampling Time}}&
\includegraphics[width=0.2951\linewidth]{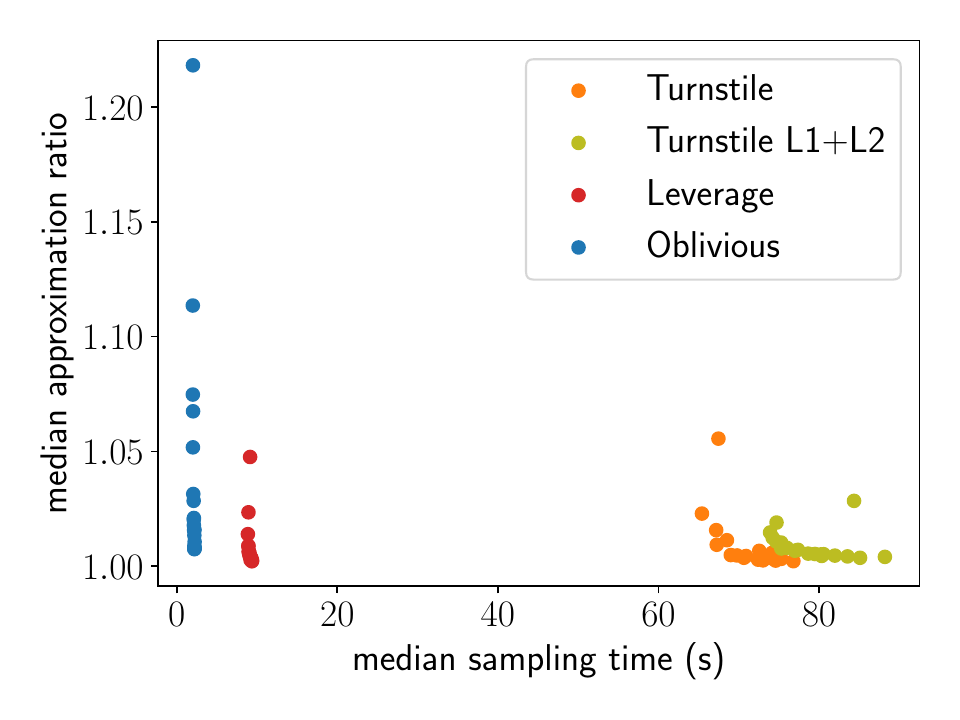}&
\includegraphics[width=0.2951\linewidth]{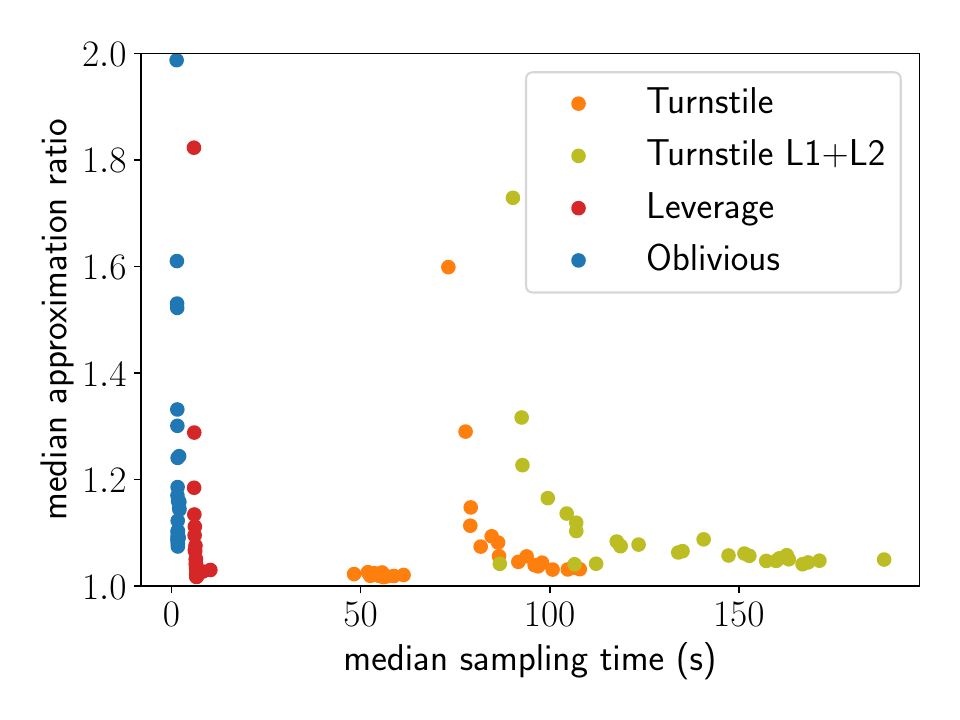}&
\includegraphics[width=0.2951\linewidth]{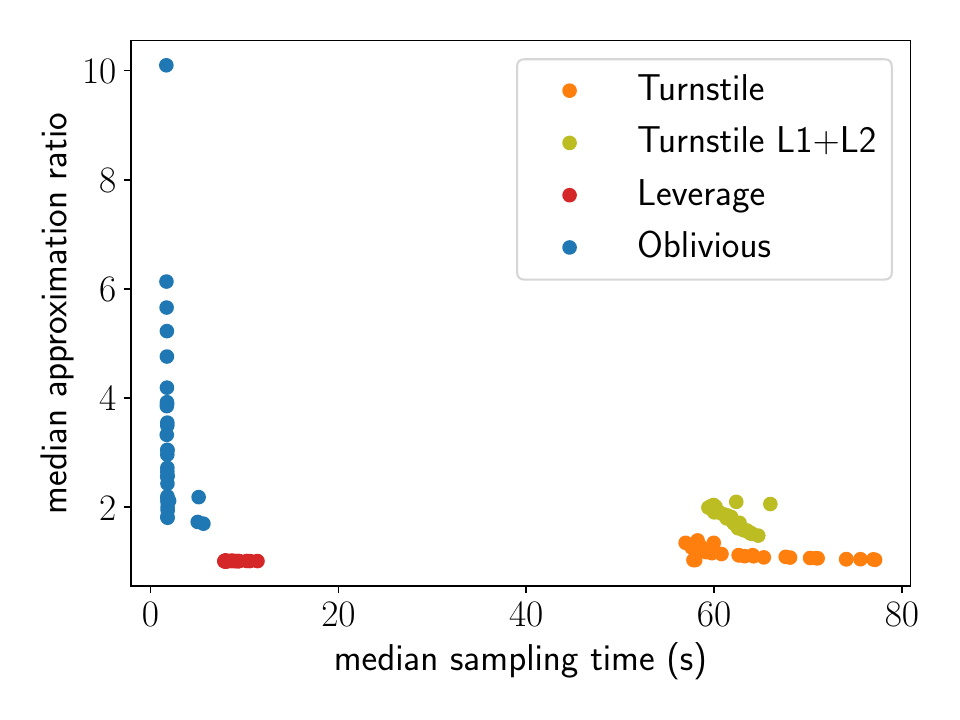}\\
\rotatebox{90}{\hspace{20pt}\textsc{Total Time}}&
\includegraphics[width=0.2951\linewidth]{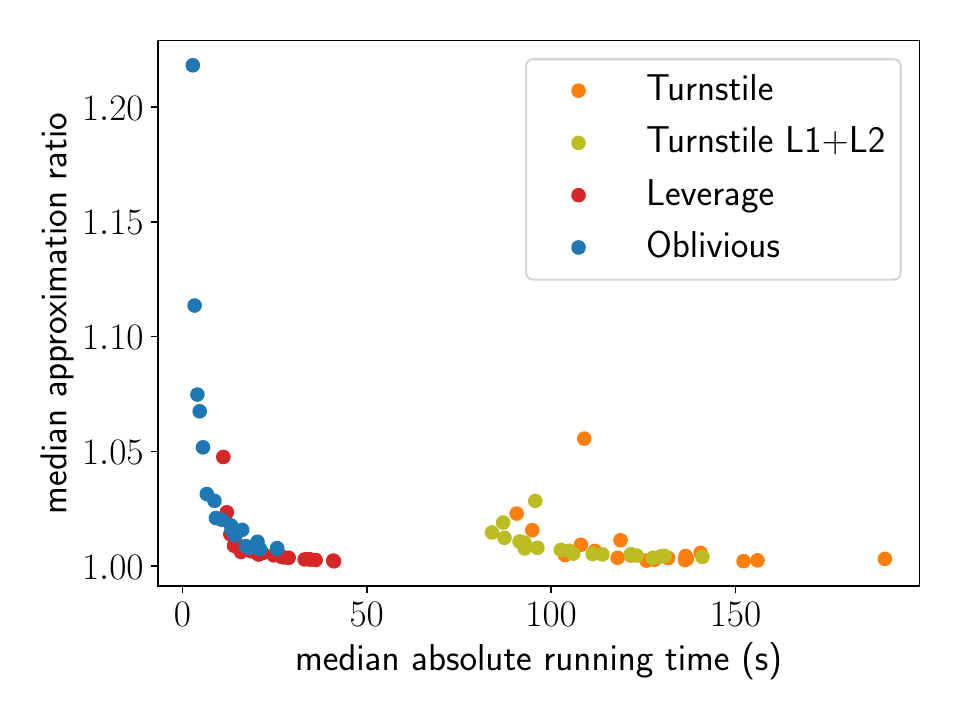}&
\includegraphics[width=0.2951\linewidth]{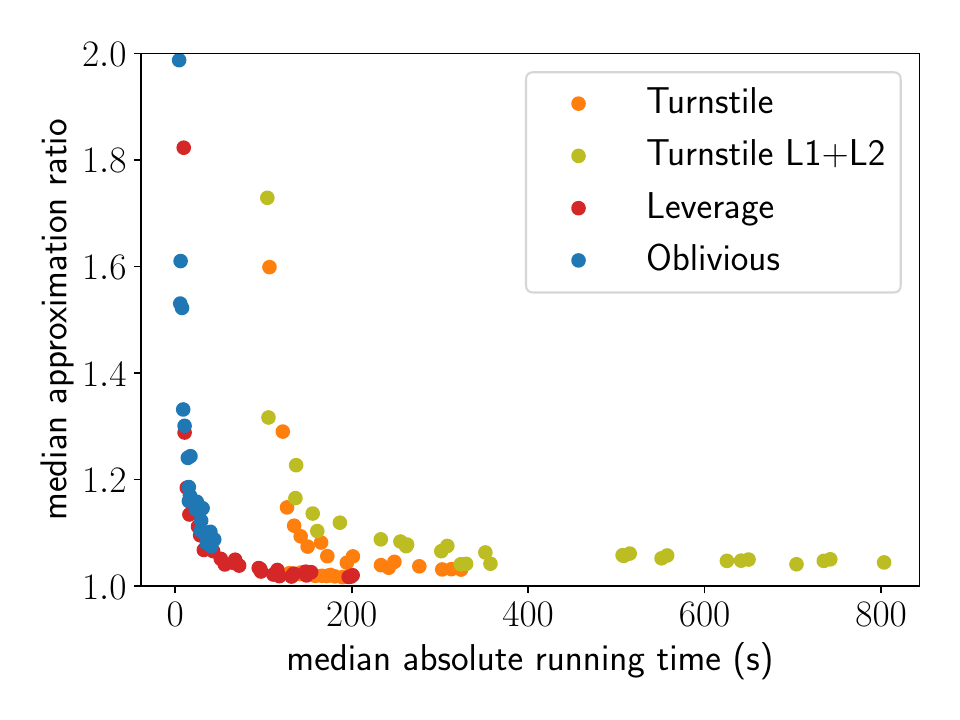}&
\includegraphics[width=0.2951\linewidth]{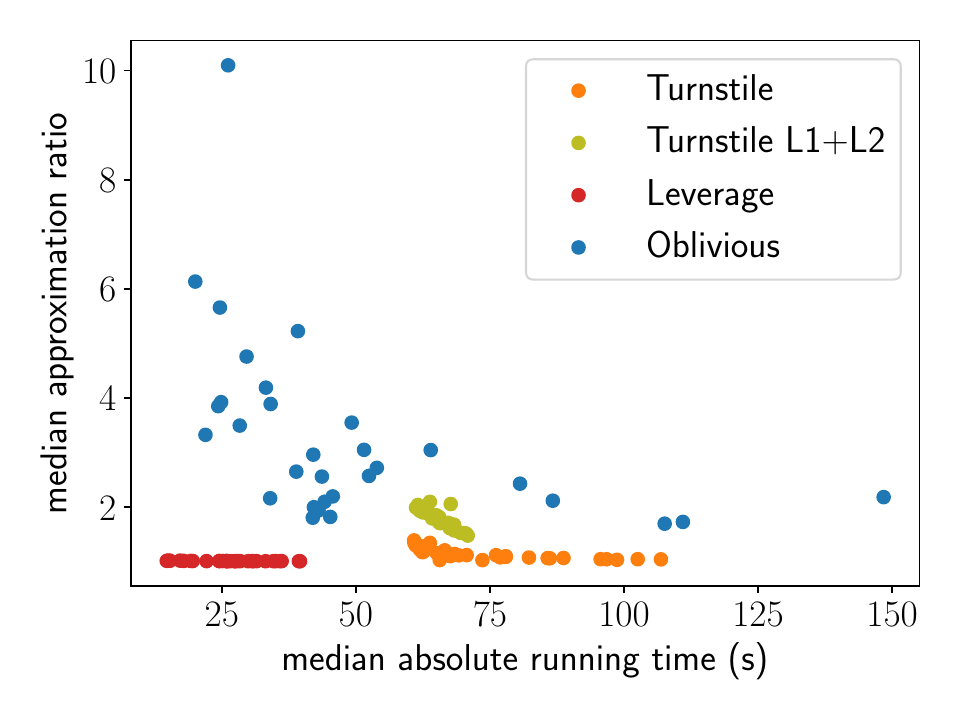}\\
\end{tabular}
\caption{Comparison of the approximation ratios and running times for logistic regression on various real-world datasets. The new turnstile data stream sampler for $p=1$ (orange) and a mixture $p=1,q=2$ (lime) is compared to plain leverage score sampling (red), and to plain oblivious sketching (blue). The plots indicate the median of approximation ratios taken over 21 repetitions for each reduced size. Best viewed in colors, lower is better.
}
\label{fig:experiments:logreg}
\end{center}
\end{figure*}

\clearpage
\subsection{Experiments for \texorpdfstring{$\ell_1$}{l1} regression}
\begin{figure*}[ht!]
\begin{center}
\textbf{\textsc{Linear $\ell_{1}$ Regression}}
\vskip 0.2in
\begin{tabular}{cccc}
{~}&
{\small\hspace{.5cm}\textsc{CoverType}}&{\small\hspace{.5cm}\textsc{WebSpam}}&
{\small\hspace{.5cm}\textsc{KddCup}}\\
\rotatebox{90}{\hspace{6pt}\textsc{Approx. Ratio}}&
\includegraphics[width=0.2951\linewidth]{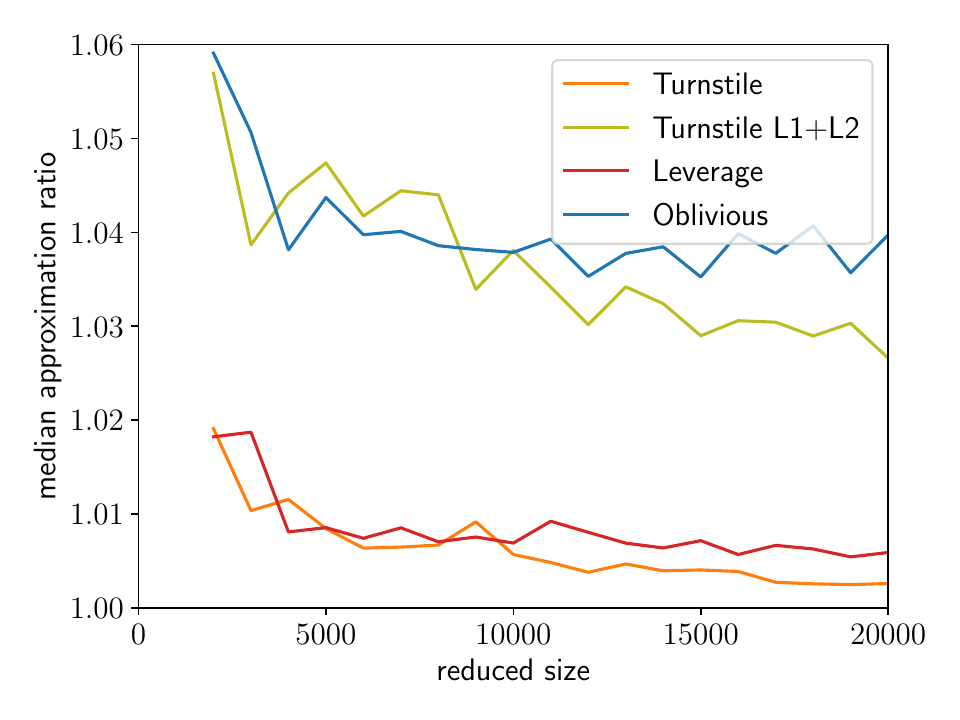}&
\includegraphics[width=0.2951\linewidth]{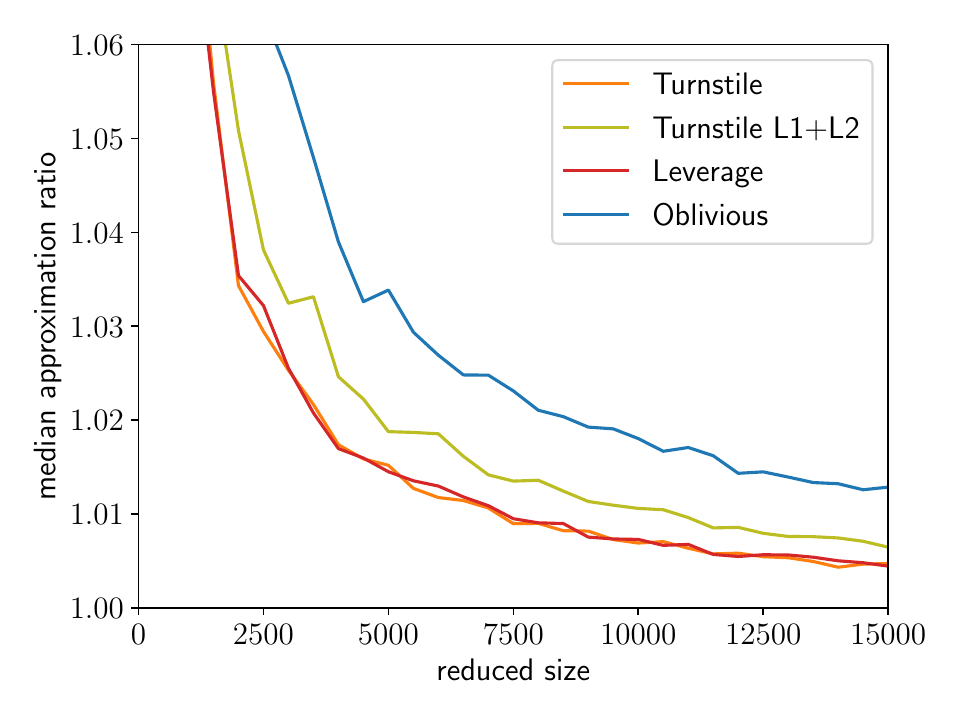}&
\includegraphics[width=0.2951\linewidth]{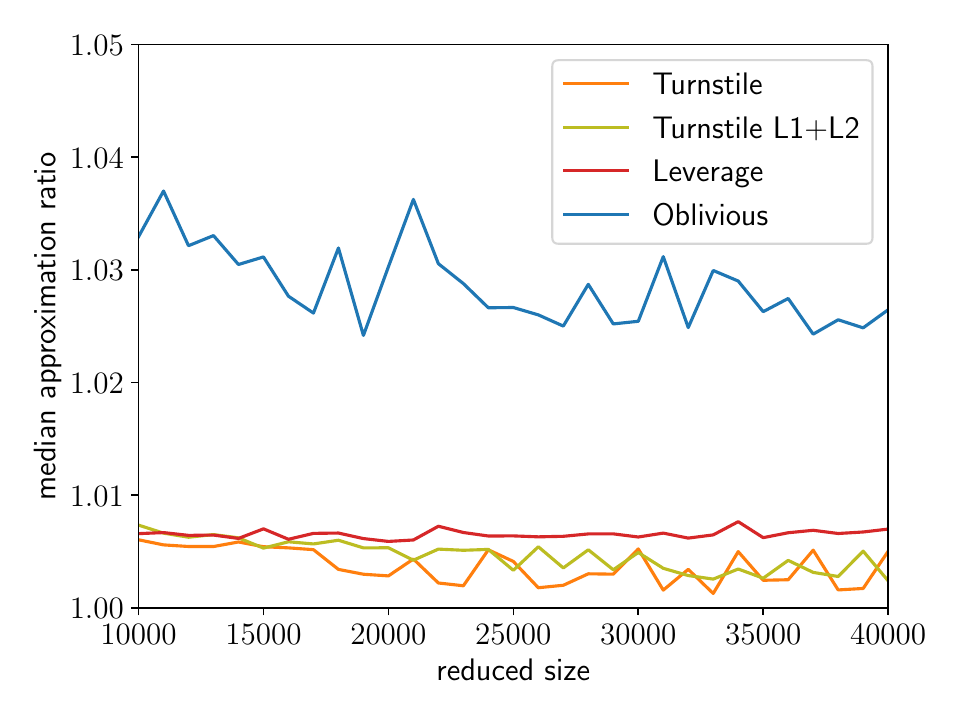}\\
\rotatebox{90}{\hspace{20pt}\textsc{Sampling Time}}&
\includegraphics[width=0.2951\linewidth]{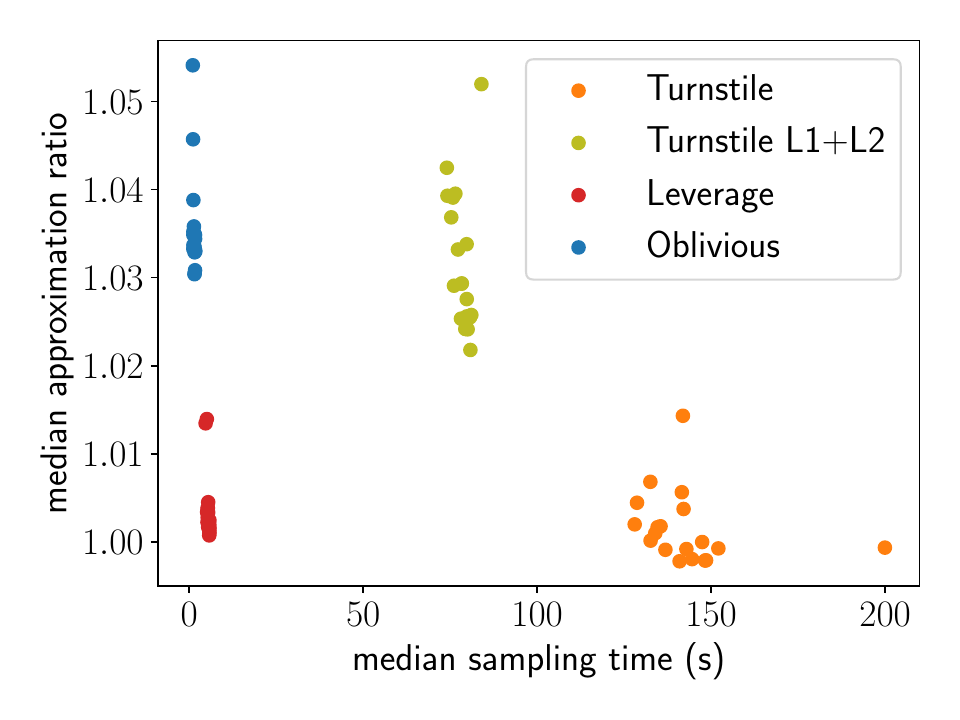}&
\includegraphics[width=0.2951\linewidth]{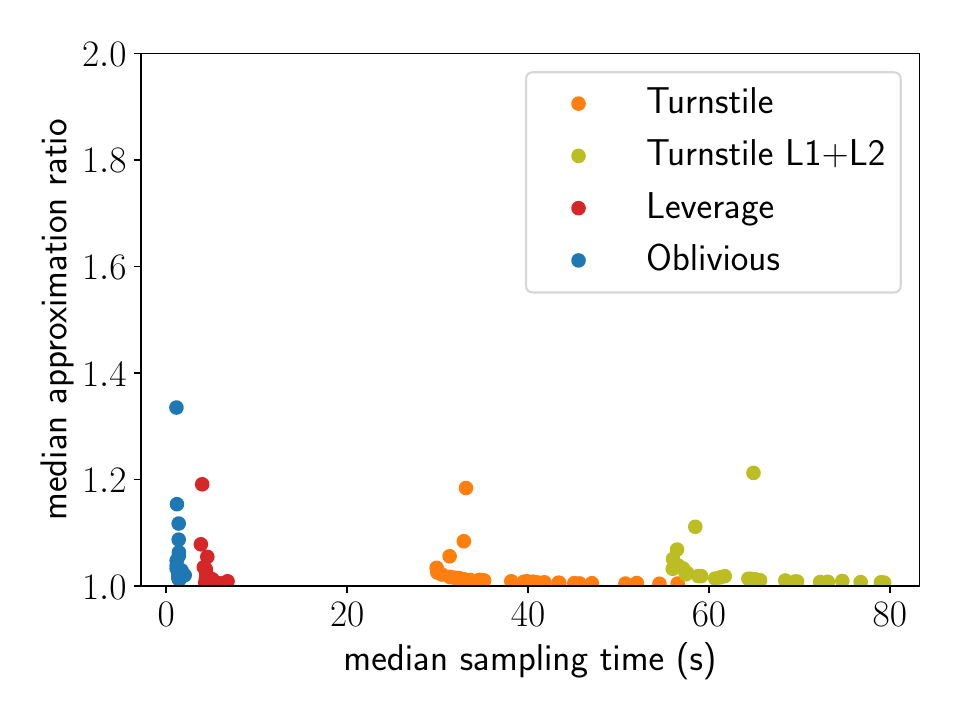}&
\includegraphics[width=0.2951\linewidth]{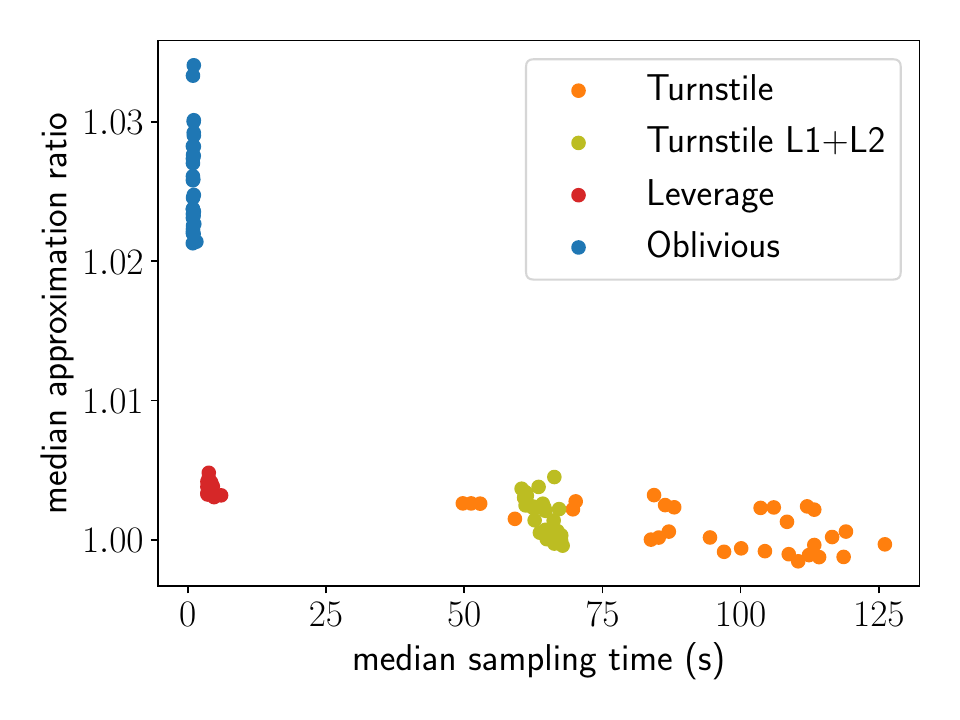}\\
\rotatebox{90}{\hspace{20pt}\textsc{Total Time}}&
\includegraphics[width=0.2951\linewidth]{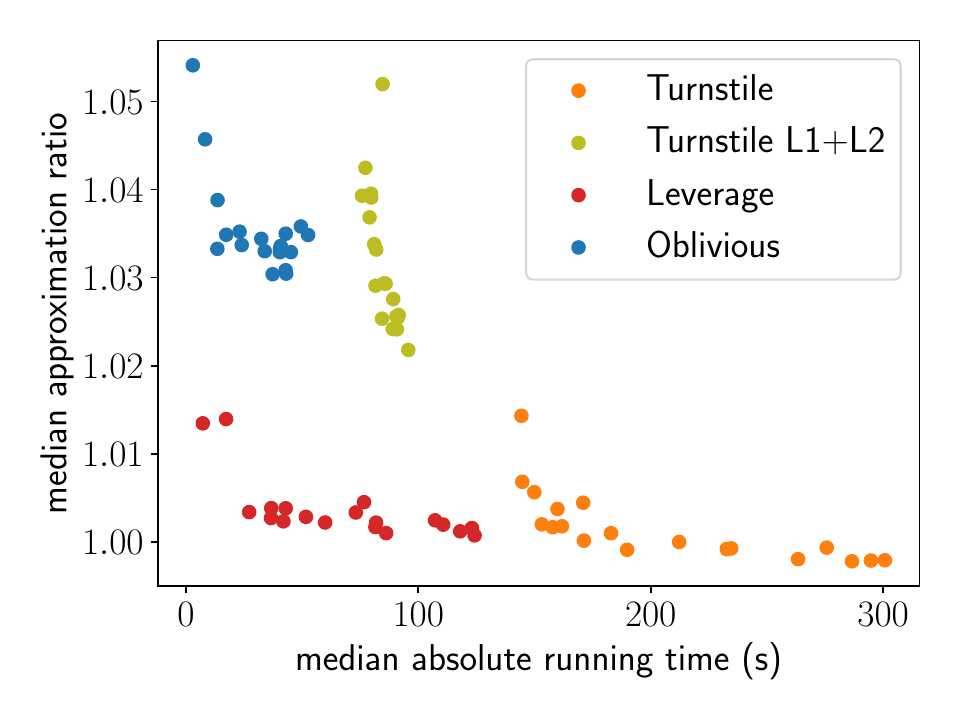}&
\includegraphics[width=0.2951\linewidth]{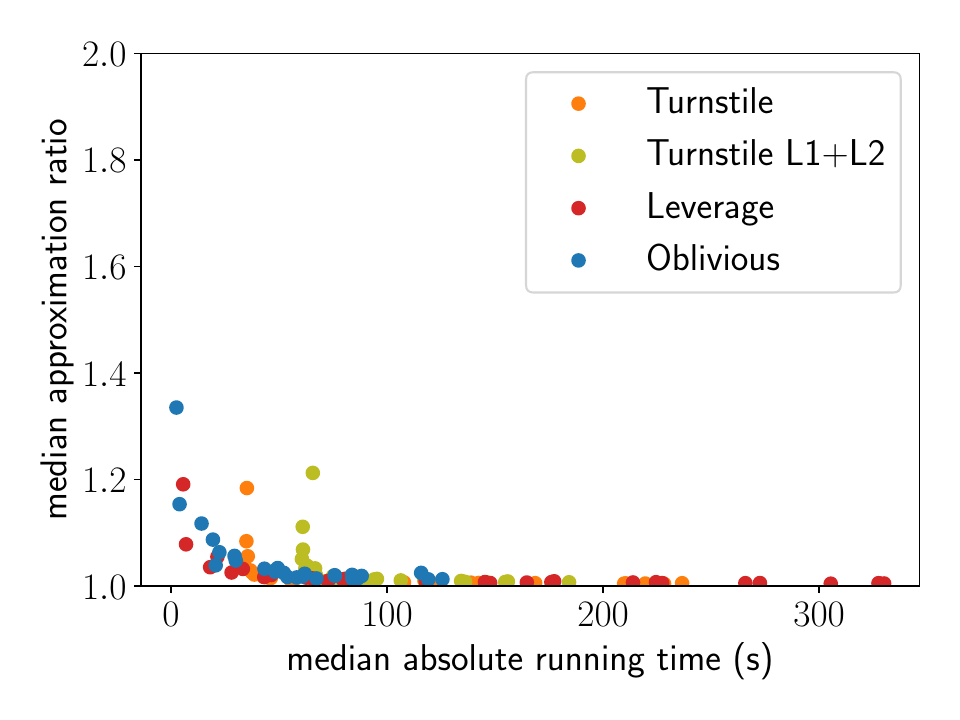}&
\includegraphics[width=0.2951\linewidth]{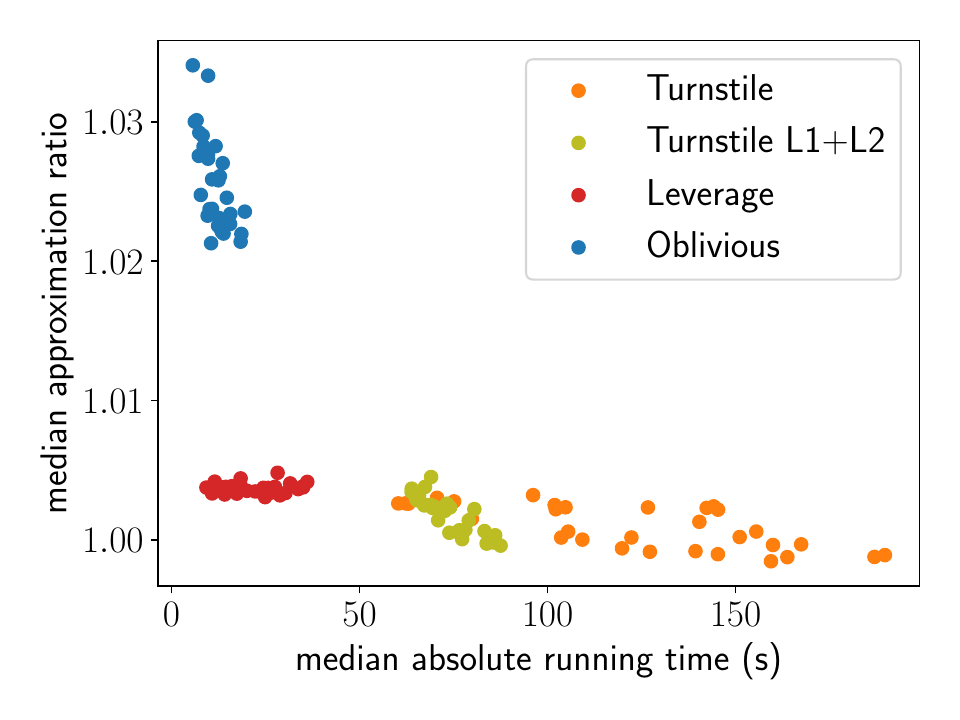}\\
\end{tabular}
\caption{Comparison of the approximation ratios and running times for $\ell_1$ regression on various real-world datasets. The new turnstile data stream sampler for $p=1$ (orange) and a mixture $p=1,q=2$ (lime) is compared to plain leverage score sampling (red), and to plain oblivious sketching (blue). The plots indicate the median of approximation ratios taken over 21 repetitions for each reduced size. Best viewed in colors, lower is better.
}
\label{fig:experiments:l1}
\end{center}
\end{figure*}

\clearpage
\subsection{Experiments for \texorpdfstring{$\ell_{1.5}$}{l1.5} regression}
\begin{figure*}[ht!]
\begin{center}
\textbf{\textsc{Linear $\ell_{1.5}$ Regression}}
\vskip 0.2in
\begin{tabular}{cccc}
{~}&
{\small\hspace{.5cm}\textsc{CoverType}}&{\small\hspace{.5cm}\textsc{WebSpam}}&
{\small\hspace{.5cm}\textsc{KddCup}}\\
\rotatebox{90}{\hspace{6pt}\textsc{Approx. Ratio}}&
\includegraphics[width=0.2951\linewidth]{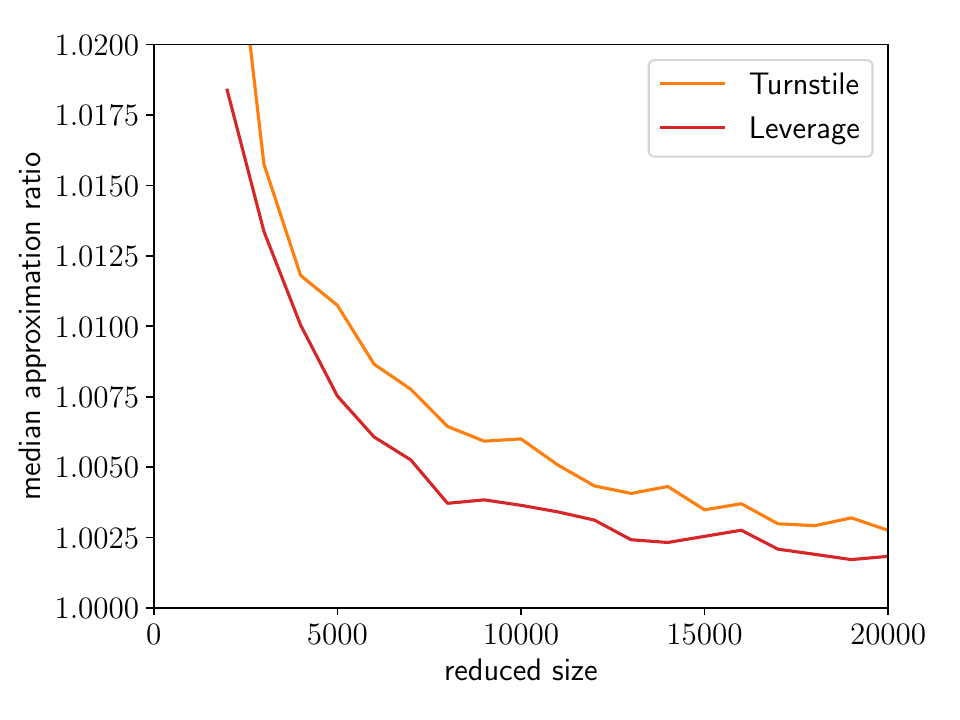}&
\includegraphics[width=0.2951\linewidth]{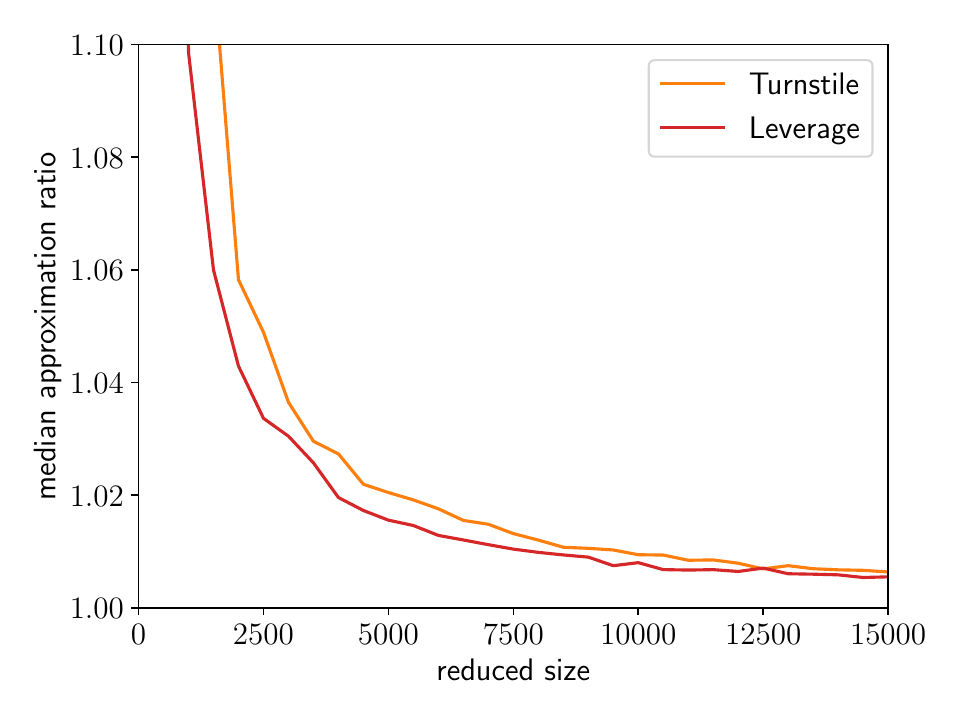}&
\includegraphics[width=0.2951\linewidth]{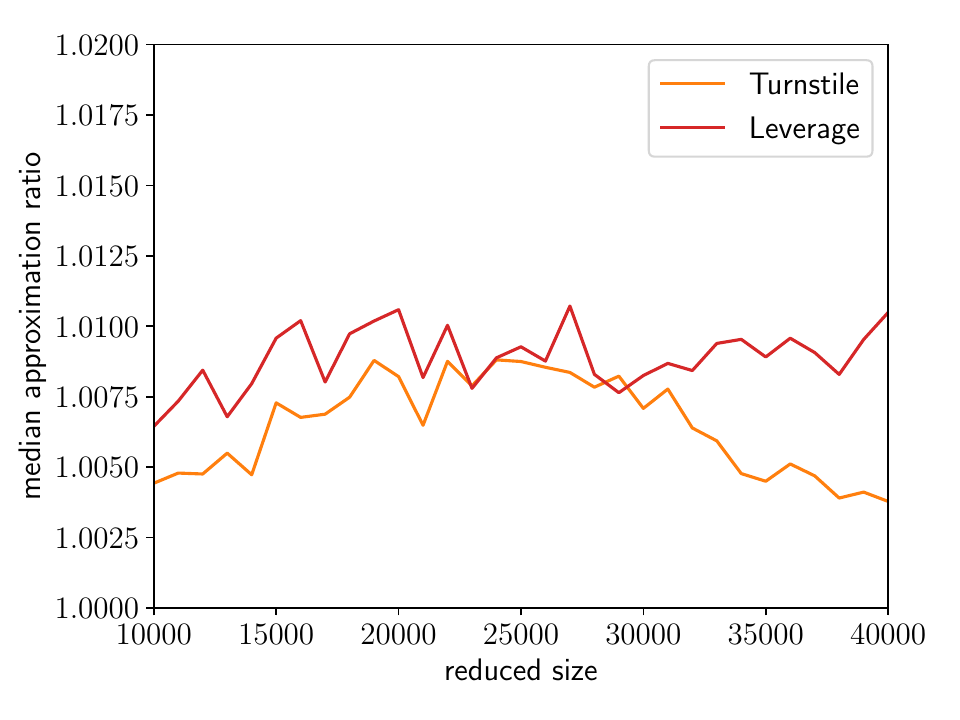}\\
\rotatebox{90}{\hspace{20pt}\textsc{Sampling Time}}&
\includegraphics[width=0.2951\linewidth]{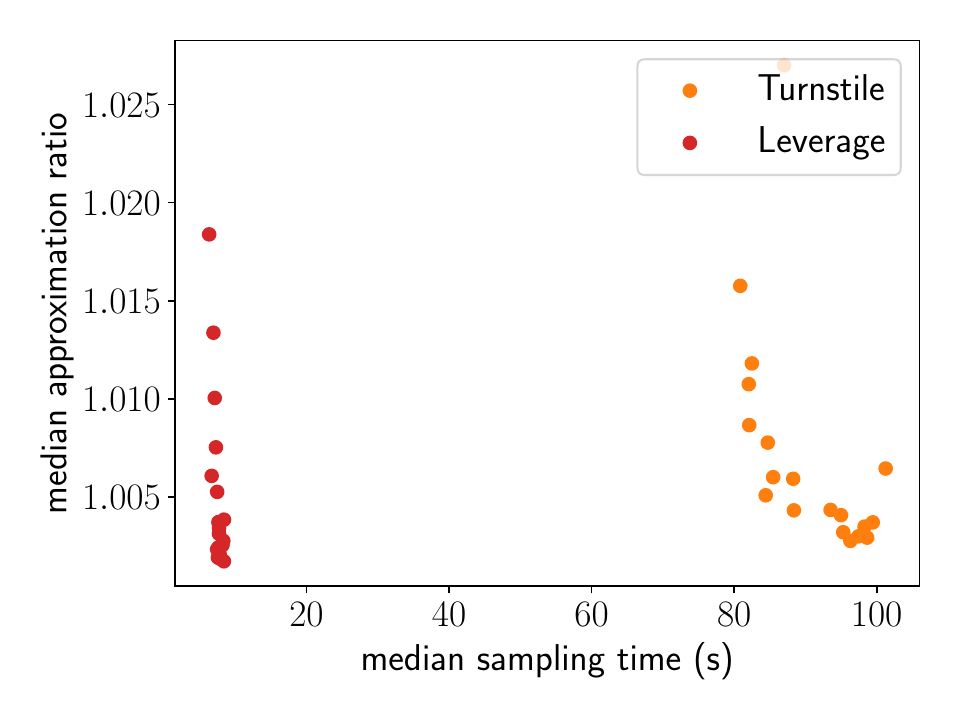}&
\includegraphics[width=0.2951\linewidth]{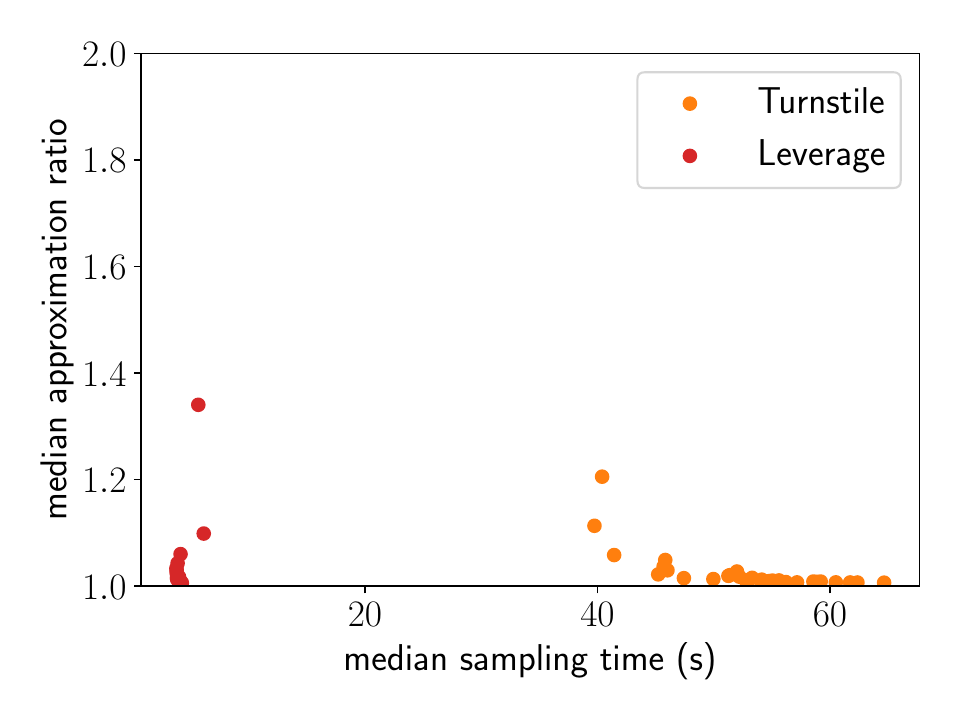}&
\includegraphics[width=0.2951\linewidth]{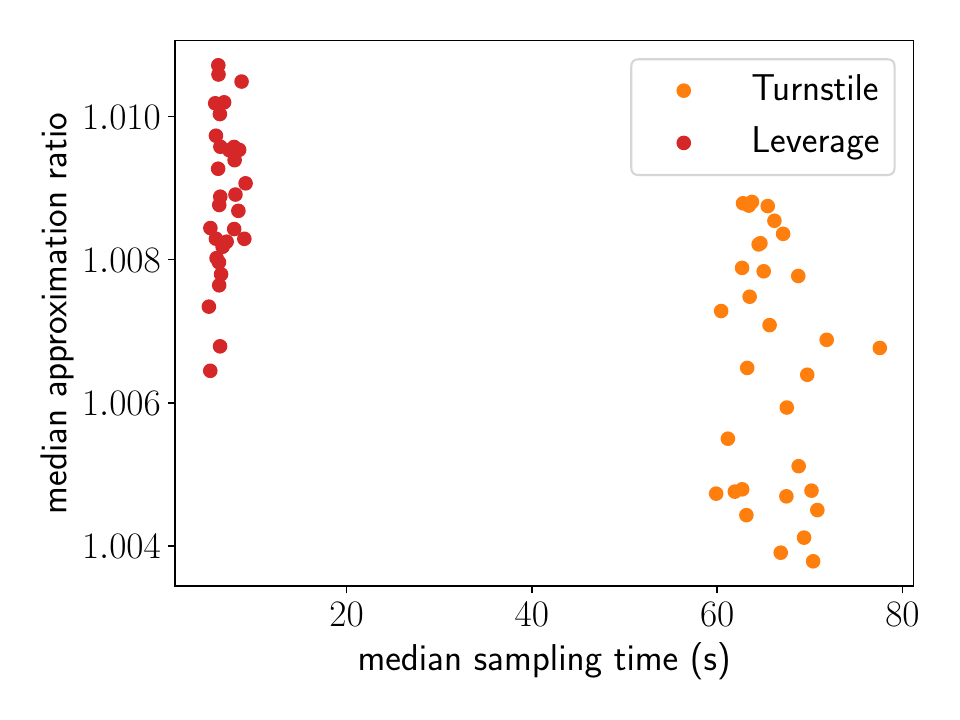}\\
\rotatebox{90}{\hspace{20pt}\textsc{Total Time}}&
\includegraphics[width=0.2951\linewidth]{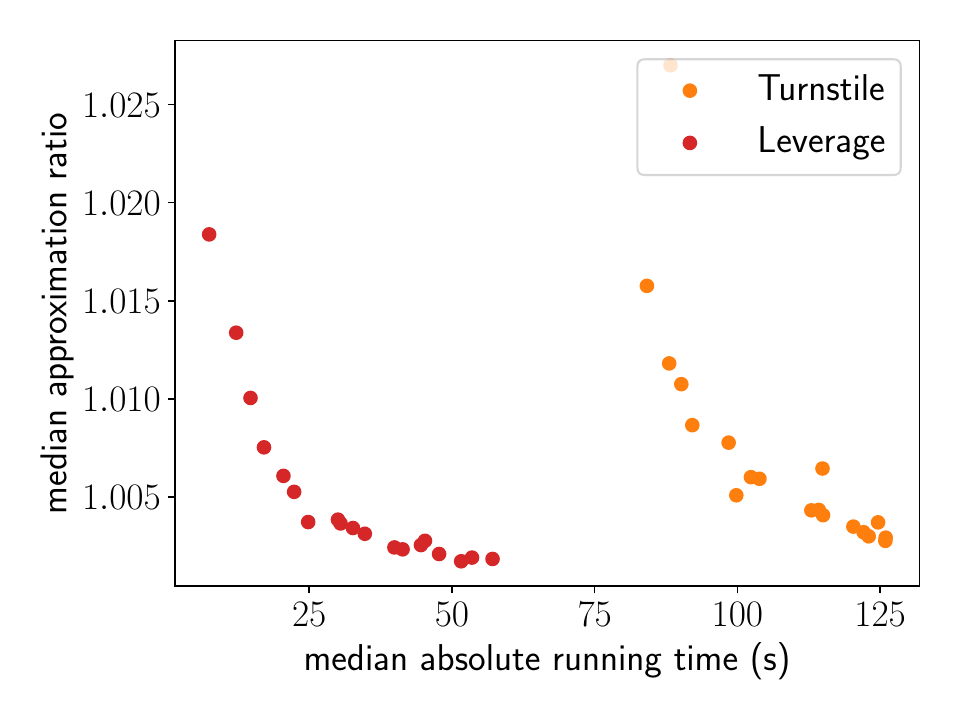}&
\includegraphics[width=0.2951\linewidth]{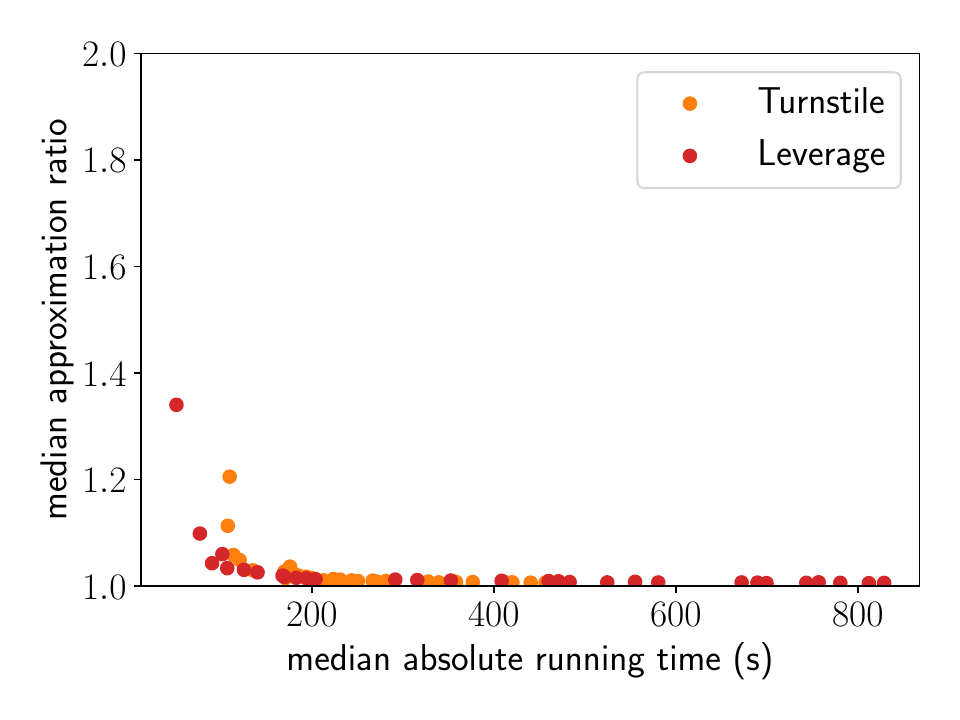}&
\includegraphics[width=0.2951\linewidth]{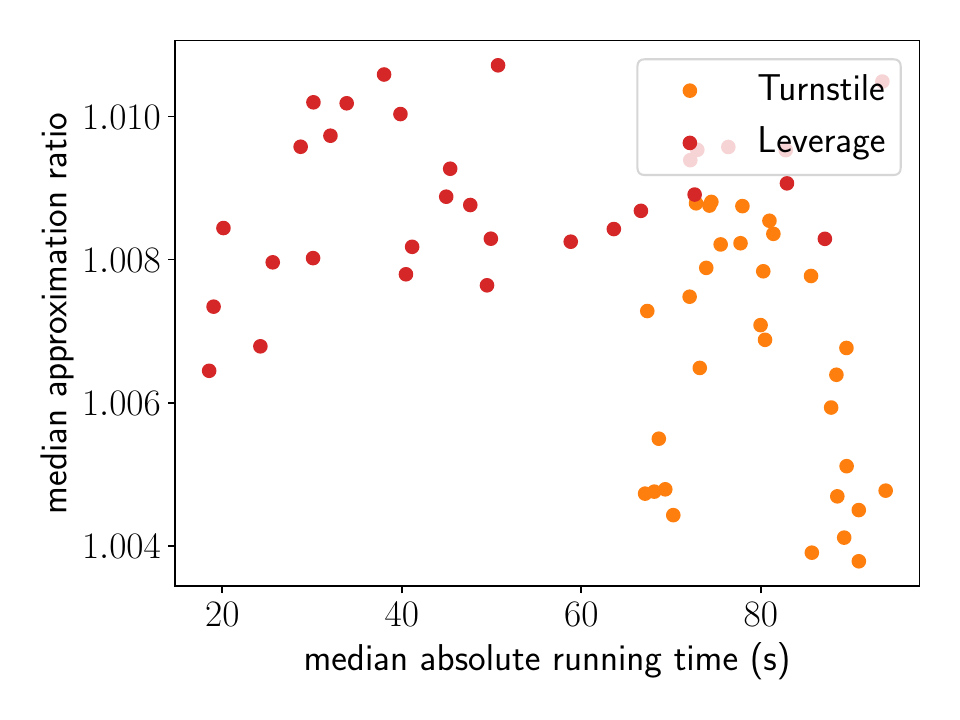}\\
\end{tabular}
\caption{Comparison of the approximation ratios and running times for $\ell_{1.5}$ regression on various real-world datasets. The new turnstile data stream sampler for $p=1.5$ (orange) is compared to plain leverage score sampling for $p=1.5$ (red). The plots indicate the median of approximation ratios taken over 21 repetitions for each reduced size. Best viewed in colors, lower is better.
}
\label{fig:experiments:l15}
\end{center}
\end{figure*}

\end{document}